\documentclass[letterpaper,11pt]{article}

\usepackage[margin=1in]{geometry}
\usepackage{amsfonts}
\usepackage{amsmath}
\usepackage{amssymb}
\usepackage{amsthm}
\usepackage{xspace}
\usepackage{xcolor}
\usepackage{hyperref}
\usepackage{graphicx}
\usepackage{url}
\usepackage{paralist}
\usepackage{verbatim}
\usepackage[final]{pdfpages}
\usepackage{titling}

\hypersetup{
    breaklinks,
    colorlinks=true,
    linkcolor=black,
    citecolor=[rgb]{0,0,0.4},
    filecolor=black,
    urlcolor=black,
    pdftitle={}, pdfauthor={} }

\newcommand{\KTZero}{\ensuremath{\mathsf{KT}_0}\xspace}

\newcommand{\sd}{\textsc{SetDisjointness}\xspace}

\newcommand{\ktzero}{\ensuremath{\mathsf{KT}_0}}
\newcommand{\ktone}{\ensuremath{\mathsf{KT}_1}}

\newcommand{\ktrho}{\ensuremath{\mathsf{KT}_\rho}}
\newcommand{\OT}{\ensuremath{\tilde{O}}}

\newcommand{\mxis}{\text{MaxIS}\xspace}
\newcommand{\mxm}{\text{MaxM}\xspace}
\newcommand{\congest}{\ensuremath{\mathsf{CONGEST}\xspace}}
\newcommand{\local}{\ensuremath{\mathsf{LOCAL}\xspace}}

\newcommand{\lb}{\left}
\newcommand{\rb}{\right}
\newcommand{\set}[1]{\{#1\}}

\def\polylog{\operatorname{polylog}}
\def\poly{\operatorname{poly}}

\newtheorem{theorem}{Theorem}[section]
 \newtheorem{lemma}[theorem]{Lemma}
 \newtheorem{claim}[theorem]{Claim}
 \newtheorem{definition}[theorem]{Definition}
 \newtheorem{remark}[theorem]{Remark}
 \newtheorem{corollary}[theorem]{Corollary}

\usepackage[utf8]{inputenc}

\usepackage{algorithm}
\usepackage{algorithmicx}
\usepackage{algpseudocode}
\algnewcommand{\IIf}[1]{\State\algorithmicif\ #1\ \algorithmicthen}
\algnewcommand{\EndIIf}{\unskip\ \algorithmicend\ \algorithmicif}
\algtext*{EndWhile} \algtext*{EndIf} \algtext*{EndFor} \algdef{SE}[SUBALG]{Indent}{EndIndent}{}{\algorithmicend\ }\algtext*{Indent}
\algtext*{EndIndent}

\usepackage{bbm}
\newcommand{\phaseBound}{\lceil \log(q^{-1}) \rceil}
\newcommand{\E}{\mathbb{E}}
\newcommand{\Ind}{\mathbbm{1}}
\newcommand{\cproba}{100 \log(n)}  \newcommand{\withProba}[1]{#1}

\title{The Message Complexity of Distributed Graph Optimization}
\date{}

\newcommand*\email[1]{\href{mailto:#1}{\url{#1}}}
\author{
Fabien Dufoulon\thanks{Part of the work was done while Fabien Dufoulon was a postdoctoral fellow at the University of Houston in Houston, USA. During that time, F. Dufoulon was supported in part by National Science Foundation (NSF) grants CCF-1540512, IIS-1633720, and CCF-1717075 and U.S.-Israel Binational Science Foundation (BSF) grant 2016419.} \\ Lancaster University \\ 
\email{f.dufoulon@lancaster.ac.uk}\and Shreyas Pai\thanks{Supported in part by Research Council of Finland Grant 334238 and Helsinki Institute for Information Technology HIIT.} \\ Aalto University \\ \email{shreyas.pai@aalto.fi}\and Gopal Pandurangan\thanks{Supported in part by National Science Foundation (NSF) grants CCF-1540512, IIS-1633720, and CCF-1717075 and U.S.-Israel Binational Science Foundation (BSF) grant 2016419.} \\ University of Houston \\ \email{gopal@cs.uh.edu}\and Sriram V.~Pemmaraju\thanks{Partially supported by NSF grant IIS-1955939.} \\ University of Iowa \\ \email{sriram-pemmaraju@uiowa.edu}\and Peter Robinson \\ Augusta University \\ \email{perobinson@augusta.edu}}

\begin{document}

\maketitle

\vspace{-1em}
\begin{abstract}
The {\em message complexity} of a distributed algorithm is the total number of messages  sent by all nodes over the course of the algorithm.
This paper studies the message complexity of distributed algorithms for fundamental graph optimization problems. 
We focus on four classical graph optimization problems: Maximum Matching (\mxm{}), Minimum Vertex Cover (MVC), Minimum Dominating Set (MDS), and Maximum Independent Set (\mxis{}). In the sequential setting, these problems are representative of a wide spectrum of hardness of approximation. While there has been some progress
 in understanding the round complexity of distributed algorithms (for both exact and approximate versions) for these problems, much less is known about their message complexity and its relation with the quality of approximation. 
 We  almost fully quantify  the message complexity of distributed graph optimization by showing the following results:
\begin{enumerate}
\item {\em Cubic regime:} Our first main contribution is showing essentially {\em cubic}, i.e.,  $\tilde{\Omega}(n^3)$ lower bounds\footnote{$\tilde{\Omega}$ and $\tilde{O}$ hide a $1/\polylog{n}$ and $\polylog{n}$ factor respectively.} (where $n$ is the number of
nodes in the graph)  on the message complexity
of distributed {\em exact} computation of  Minimum Vertex Cover (MVC), Minimum Dominating Set (MDS), and Maximum Independent Set (\mxis{}). Our lower bounds apply to any distributed algorithm that runs in polynomial number of rounds (a mild and necessary restriction). Our result is significant since, to the best of our knowledge, this are the first
$\omega(m)$ (where $m$ is the number of edges in the graph) message lower bound known
for distributed computation of such classical graph optimization problems. Our bounds are essentially tight, as all these problems
can be solved trivially using $O(n^3)$ messages in polynomial rounds. All these bounds hold in the standard  \congest{} model  of distributed computation in which messages are of $O(\log n)$ size.

\item {\em Quadratic regime:} In contrast, we show that if we allow {\em approximate} computation then $\tilde{\Theta}(n^2)$ messages are both necessary and sufficient. Specifically, we show that $\tilde{\Omega}(n^2)$ messages are required for {\em constant-factor} approximation algorithms for all four problems. 
For \mxm{} and MVC, these bounds hold for any constant-factor approximation, whereas for MDS and \mxis{} they hold for any approximation factor better than some specific constants. 
These lower bounds hold even in the \local{} model (in which messages can be arbitrarily large) and they even apply to algorithms that take arbitrarily many rounds. 
We show that our lower bounds are essentially tight, by showing that
 if we allow approximation to within an arbitrarily small constant factor, then all these problems can be solved using $\tilde{O}(n^2)$ messages even in the \congest{} model.

\item {\em Linear regime:} We complement the above lower bounds by showing distributed algorithms with  $\tilde{O}(n)$  message complexity that run in polylogarithmic rounds and give constant-factor approximations for all four  problems on \emph{random} graphs. These results imply that almost linear (in $n$) message complexity is achievable on almost all (connected) graphs of every
edge density.
\end{enumerate}
\end{abstract} 
\thispagestyle{empty}

\newpage
\thispagestyle{empty}
\tableofcontents
\newpage
\setcounter{page}{1}

\section{Introduction}
\label{section:introduction}
The focus of this paper is understanding the communication cost of distributively solving  graph optimization problems.
The  communication cost of distributed computation has been studied extensively  in theoretical computer science since the seminal work of Yao \cite{yao79_some_compl_quest_relat_distr}.
This line of work  studies the communication cost of computing boolean functions of the form $f(x,y)$, where the input $x \in \{0,1\}^n$ is given to Alice and the input $y \in \{0,1\}^n$ is given to Bob, and these two players (nodes) 
jointly compute $f(x,y)$ by communicating across a communication link (edge). The \textit{communication complexity} of $f$ is measured by the minimum number of bits exchanged by Alice and Bob to compute $f$. Boolean functions that have been studied extensively include {\em equality}, {\em set disjointness} etc.; we refer to \cite{KushilevitzN:book96,rao_yehudayoff_2020} for a comprehensive treatment.  
We note that the communication complexity of functions on graphs, e.g., connectivity, bipartiteness, maximum matching, etc., has also been studied; see \cite{10.1145/62212.62228,huang_et_al:LIPIcs:2015:4934,DBLP:conf/fsttcs/IvanyosKLSW12,vdurivs1989communication,goel2012communication} for some examples.
In the context of a graph problem, each player gets a {\em portion of an input graph} $G$. The graph may be edge-partitioned into two parts or may be partitioned in some other arbitrary way.

Over the years many extensions and variants of this basic ``two-party'' communication complexity model have been studied.
One important early variant is by Tiwari \cite{TiwariJACM1987} who studied the same problem, but instead of the two players communicating via an edge, the two players communicate via an arbitrary network.
Another important variant is {\em multi-party} communication complexity, introduced in the work of Chandra, Furst and Lipton \cite{ChandraFLSTOC1983}, where $k$ players are provided inputs $x_1, \dots,  x_k \in \{0, 1\}^n$ and these $k$ players want to
compute some joint boolean function $f : (\{0, 1\}^n)^k \rightarrow \{0, 1\}$ with the goal of minimizing the
total communication between the $k$ players. In the multi-party model, the $k$ players are connected by a communication network, which is typically a clique, but arbitrary topologies have also been considered (see for e.g., \cite{WoodruffZhangDIST2013,ChattopadhyayRRFOCS2014} and the references therein).

A related, yet different line of work  has been the study of the \textit{message complexity} in distributed computing (see for e.g., \cite{kuttenjacm15,PaiPPRPODC2021,AwerbuchGPV90,afek1991time,DBLP:journals/talg/Pandurangan0S20,10.1145/3380546}).
In this line of work, we are given an input graph $G$, which {\em also} serves as the communication network. The nodes of this network (which can be viewed as players) communicate along the edges of $G$ to solve problems {\em defined on} $G$. 
The message complexity is simply the total number of messages sent by all nodes over the course of the algorithm. 
Usually, only messages of small size (say, $O(\log n)$ bits) are allowed, hence in such cases, message complexity is essentially the total number of bits communicated (up to a small factor).
A wide variety of graph problems have been studied in this setting, including classical problems such as breadth-first search, minimum spanning tree, minimum cut, maximal independent set, $(\Delta+1)$-coloring, shortest paths, etc., and NP-complete problems such as minimum vertex cover, minimum dominating set, etc. 
A key difference between this line of work and the previously described research on two-party and multi-party communication complexity models is that here the input graph $G$ is {\em also} the communication network and so the structure of $G$ simultaneously determines both the difficulty of the problem and the difficulty of communication needed to solve the problem. 
For example, it may be that a particular problem is easier to solve on a sparse input graph $G$, but the sparsity of $G$ also limits the volume of information that can be exchanged along the edges of $G$.
A second difference is that most prior work on two-party and multi-party communication complexity
by default assumes an {\em asynchronous} communication model. Whereas in  distributed computing, a {\em synchronous} model
of computation, i.e., a model with a global clock, is extensively used.  
While there has been a significant progress in our understanding of communication complexity in the 2-party and multi-party settings, our understanding of the message complexity of distributed computation of graph problems is significantly limited. 
We refer to Section \ref{section:related} for more details comparing and contrasting the above lines of research.

The focus of this paper is gaining a deeper understanding of the \textit{message complexity} of distributed algorithms for fundamental graph optimization problems. 
Besides message complexity,  {\em round complexity} is also a key measure of the performance of distributed algorithms.
While a rich body of literature exists on the round complexity of distributed exact and approximation algorithms for graph optimization problems \cite{Bar-YehudaCSJACM2017,KMWJACM2016,KMWPODC2004,Censor-HillelD18, JiaRS02, DKM19, GK18,PanditPemmarajuPODC2010,BachrachCDELP19,Censor-HillelKP17}, much less is known about the message complexity of these problems and the possibility of tradeoffs between the message complexity and the quality of approximation  that can be achieved for these problems.

We focus on four classical graph optimization problems: Maximum Matching (\mxm), Minimum Vertex Cover (MVC), Minimum Dominating Set (MDS), and Maximum Independent Set (\mxis). In the sequential setting, these problems are representative of a wide spectrum of hardness of approximability. 
\mxm can be solved exactly in polynomial time. MVC has a simple 2-approximation algorithm, but it does not have a better than $1.3606$-approximation~\cite{DinurSafraAnnMath2005}. A simple greedy algorithm provides a $O(\log \Delta)$-approximation to MDS \cite{Vaziranibook}, though it is known that MDS does not have a $(1-\epsilon) \cdot \ln \Delta$-approximation (where $\Delta$ is the maximum degree of the graph) for any $0 < \epsilon < 1$ \cite{DinurSteurerSTOC2014}. Finally, \mxis is known to be even harder; it does not even have an $O(n^{1-\epsilon})$-approximation for any $0 < \epsilon < 1$ \cite{Vaziranibook}. All of these hardness of approximation results are conditional on P $\not=$ NP.

In the standard models of distributed computing such as \local{} and \congest, it is assumed that processors have infinite computational power. This means that hardness of approximation results in the sequential setting do not directly translate to the distributed setting. Hardness of approximation in the distributed setting is, roughly speaking, due to the distance information has to travel or the volume of information that has to travel for nodes to produce a solution that is close enough to optimal.
Researchers are starting to better understand the distributed hardness of approximation from a \emph{round complexity} point of view, but intriguing gaps remain.
For example, there is a $(2+\epsilon)$-approximation algorithm for MVC (even vertex-weighted MVC) running in $O(\log \Delta/(\epsilon \log \log \Delta))$ rounds in the \congest{} model \cite{Bar-YehudaCSJACM2017}. 
The approximation factor was reduced to exactly 2 in \cite{Ben-BasatEKSSIROCCO18}, but at the cost of polylogarithmic factor extra rounds.
These upper bounds are complemented by an $\tilde{\Omega}(n^2)$ \emph{round} lower bound for solving MVC \emph{exactly} in the \congest{} model \cite{Censor-HillelKP17}. Currently, the round complexity of obtaining an $\alpha$-approximation for MVC, for $1 < \alpha < 2$, is unknown.

Even this level of understanding is lacking about the \emph{message complexity} of distributed graph optimization. The key question addressed in this paper is this: \textit{how does the message complexity of fundamental distributed graph optimization problems change as we move from exact algorithms to approximation algorithms?}
We almost fully answer this question and the key takeaway from our results is that there is a sharp separation between the message complexity of exact and approximate solutions. Specifically, we show that $\tilde{\Theta}(n^3)$ messages are necessary and sufficient for the distributed {\em exact} computation of  MVC, MDS, and \mxis for algorithms that runs in a polynomial number of rounds. In contrast, we show that if we allow {\em approximate} computation, for a constant-approximation factor, then $\tilde{\Theta}(n^2)$ messages are both necessary and sufficient for algorithms that run in polynomial rounds for all four problems, \mxm, MVC, MDS, and \mxis. 
We note that focusing on algorithms that run in polynomial rounds is hardly restrictive because any problem can be solved in polynomial rounds in standard distributed computing models (e.g., \congest{} and \local{}) by gathering the entire input
at a single node.

\subsection{Distributed Computing Models}
\label{section:background}

We primarily work in the {\em synchronous}  version of a standard message-passing model of distributed computing known as \congest{} \cite{peleg00}.
In this model, the input is a graph \(G=(V, E)\), $n = |V|$, $m = |E|$, which also serves as the communication network.
Nodes in the graph are processors with unique \texttt{ID}s from a space whose size is polynomial in $n$. 
In the \emph{synchronous} version of this model, it is assumed that all nodes have access to have a common global clock, and   both computation and communication proceed in lockstep, i.e.,  in discrete time steps called \textit{rounds}.\footnote{We note that all our lower bounds also hold in the more general {\em asynchronous} model  where there is no such assumption of a common clock. See Section \ref{section:cubicLowerBounds}.}
In each round, each node (i) receives messages (if any) sent to it in the previous round, (ii) performs local computation based on information it has, and (iii) sends a message (possibly different) to each of its neighbors in the graph.
Processors are assumed to be arbitrarily powerful and can perform arbitrary (e.g., exponential-time) local computations in a round.
We allow only small, i.e., $O(\log n)$-sized messages, to be sent per edge per round. 
Since each \texttt{ID} can be represented with $O(\log n)$ bits, each message in the $\congest$ model is large enough to contain $O(1)$ \texttt{ID}s.
We note that some of our message lower bounds also hold in the less restrictive \local{} model, where messages sent per edge per round can be of arbitrary size.

We primarily work in the standard $\ktzero$ (\textit{\textbf{K}nowledge \textbf{T}ill radius 0}) model, also called the {\em clean network model} \cite{peleg00}, in which nodes have initial local knowledge of only themselves and do not know anything else about the network; specifically, nodes know nothing about their neighbors (e.g., \texttt{ID}s of neighbors).
As we note later, some of our lower bounds even extend to the $\ktone$ model, in which each node has initial knowledge of itself and the \texttt{ID}s of its neighbors.
The point is significant because knowledge of neighbors' \texttt{ID}s can be used in surprising ways to reduce the message complexity of algorithms (see for e.g., the Minimum Spanning Tree (MST) algorithm of  King, Kutten, and Thorup \cite{king15_const_improm_repair_mst_distr}). 
Unless explicitly specified otherwise, all the results we present are in the $\ktzero$ \congest{} model.

\subsection{Our Contributions}
\label{section:results}

Our main results, which are summarized in Table \ref{tb:results}, can be organized into 3 categories.
Column 2 shows essentially tight almost {\em cubic} (i.e., $\tilde{\Omega}(n^3)$) lower bounds in the \congest{} model on the message complexity of computing {\em exact} solutions for MVC, MDS, and \mxis in polynomial number of rounds.\footnote{It is open whether the message complexity of exactly computing \mxm is $\tilde{\Omega}(n^3)$. See Section \ref{sec:conc}.} 
This is significant since, to the best of our knowledge, these are the {\em first $\omega(m)$  message lower bounds} (where $m$ is the number of edges in the graph) known for distributed computation of graph problems in the \congest{} model.
Tight message lower bounds are known for a wide variety of problems in the $\ktzero$ \congest{} model including important global problems such as broadcast, leader election (LE), and minimum spanning tree (MST) \cite{kuttenjacm15} as well as for local symmetry breaking problems such as maximal independent set (MIS), ruling sets, and $(\Delta+1)$-coloring \cite{pai17_symmet_break_conges_model,PaiPPRPODC2021}.
But all of these lower bounds are either of the form $\Omega(m)$ or $\Omega(n^2)$.

Column 3 shows {\em quadratic} lower bounds on the message complexity of {\em constant-factor approximations} for \mxm, MVC, MDS, and \mxis. These bounds hold not just for polynomial-round algorithms, but even for algorithms that use arbitrarily many rounds. Furthermore, they hold not just in the \congest{} model, but also in the \local{} model, in which message sizes can be arbitrarily large. 
These quadratic lower bounds are tight because we are also able to show $\tilde{O}(n^2)$ message upper bounds for constant-approximation algorithms for all four problems, for arbitrarily small constant.
To the best of our knowledge, of the four problems we consider, only MDS has been previously studied from a message complexity perspective. 
In \cite{GotteKSWAlgoSensors2021,GOTTE2023113756}, the authors show an (expected) $O(\log \Delta)$-approximation algorithm to MDS in the $\ktzero$ \congest{} model that uses $O(n^{1.5})$ messages, running in polylogarithmic rounds.
This upper bound result shows that  non-trivial approximation for MDS can be achieved in the $\ktzero$ \congest{} model {\em without} communicating over most edges.
The authors also show a $\tilde{\Omega}(n^{1.5})$ message lower bound for algorithms that yield an $O(1)$-approximation for MDS. Our work significantly improves on this lower bound result by showing that $\Omega(n^2)$ is a lower bound for $5/4-\epsilon$ approximation of MDS. This indicates that message complexity can be quite sensitive to the {\em quality of approximation} for some problems and hence the $\Omega(n^2)$ message lower bounds shown in this paper for approximation algorithms cannot be taken for granted.

Finally, we note that we are able to extend these lower bounds (both cubic and quadratic), which are in the $\ktzero$ model, to the $\ktone$ model also. We present these results in the appendix so as to keep the main text of the paper focused on the $\ktzero$ model.

We complement our lower  bounds by presenting almost-quadratic, i.e.,
$\tilde{O}(n^2)$, message upper bounds for computing $(1 \pm \epsilon)$-approximations to all four problems (Column 4) and essentially linear, i.e.,
$\tilde{O}(n)$, message complexity  algorithms (Column 5) on $G(n,p)$ (Erd\"{o}s-R\'{e}nyi) random graphs\cite{bollobas}
for all four problems that give constant-factor approximations with high probability. 
We now describe the techniques used to obtain our results in more detail.

\begin{table}[tbh]
\renewcommand{\arraystretch}{1.2}
\centering
\noindent\makebox[\textwidth]{\begin{tabular}{|c|c|c|c|c|} \hline
     Problem & Lower Bounds   & Lower Bounds  & Upper Bounds  &  Upper Bounds\\
             & Exact   & Approximate  & Approximate & in Random Graphs \\ \hline
     \mxm{} & Open & $\Omega(\epsilon^{3}n^{2})$ for $\epsilon$-apx & $\tilde{O}(n^2/\epsilon)$ for $(1-\epsilon)$-apx & $\tilde{O}(n)$ for exact \\ \hline
     MVC & $\tilde{\Omega}(n^3)$  & $\Omega(n^2/c)$ for $c$-apx & $\tilde{O}(n^2/\epsilon)$ for $(1+\epsilon)$-apx & $\tilde{O}(n)$ for $(2-o(1))$-apx\\ \hline
MDS & $\tilde{\Omega}(n^3)$ & $\Omega(n^2)$ for $(\frac{5}{4}-\epsilon)$-apx & $\tilde{O}(n^2/\epsilon)$ for $(1+\epsilon)$-apx & $\tilde{O}(n)$ for $(1+o(1))$-apx\\ \hline
     \mxis{} & $\tilde{\Omega}(n^3)$ & $\Omega(n^2)$ for $\left(\frac{1}{2}+\epsilon\right)$-apx  & $\tilde{O}(n^2/\epsilon)$ for $(1-\epsilon)$-apx & $\tilde{O}(n)$ for $(\frac{1}{2} - o(1))$-apx\\ \hline
\end{tabular}
}
\caption{A summary of our message complexity lower and upper bound results. All of these results hold in the synchronous $\ktzero$ \congest{} model, with all the message complexity lower bounds applying to all algorithms that run in polynomial rounds. Additionally, the quadratic lower bounds for approximation algorithms (Column 3) even apply in the $\ktzero$ \local{} model and also to algorithms that take arbitrarily many rounds. The cubic lower bounds for exact algorithms (Column 2) are tight because any problem can be trivially solved in the $\ktzero$ \congest{} model in polynomial rounds using $O(n^3)$ messages by gathering the entire graph topology at a node. Moreover, this gathering algorithm implies that, in {\em random graphs}, all problems can be solved trivially in $\tilde{O}(n^2)$ messages with high probability since such graphs have $O(\log n)$ diameter with high probability. While our focus is not on round complexity, we note that our approximation algorithms for arbitrary graphs take polynomial number of rounds, while those for random graphs take polylogarithmic number of rounds.
For the lower bound for approximate \mxm (Columnn 3), $\epsilon \in (\frac{1}{n^{1/3}}, 1)$; everywhere else the only restriction on $\epsilon$ is $0 < \epsilon < 1$. The lower bound for approximate MVC holds for any $c \ge 1$.
}
\label{tb:results}
\end{table}
\renewcommand{\arraystretch}{1}

\paragraph*{A. Tight Cubic Lower Bounds for Exact Computation:}
The starting point for our cubic message lower bounds is the communication-complexity-based approach in \cite{BachrachCDELP19,Censor-HillelKP17} that is used to show $\tilde{\Omega}(n^2)$ \textit{round} lower bounds for exact MVC and MDS. 
In \cite{Censor-HillelKP17}, the authors present a reduction from the 2-party communication complexity problem \sd{} to MVC. For any positive integer $n$ that is a power of $2$ and bit-vectors $x, y \in \{0, 1\}^{n^2}$, this reduction maps an instance $(x, y)$ of \sd{} to a graph $G_{x,y}$ with $\Theta(n)$ vertices and $\Theta(n^2)$ edges such that
$\sd(x, y) = \mathrm{FALSE}$ iff $G_{x,y}$ has a vertex cover of size at most $4n + 4 \log n - 4$.
Furthermore, $G_{x, y}$ has the property that its vertex set can be partitioned into sets $V_x$ and $V_y$ where the subgraph $G_{x, y}[V_x]$ is determined completely by $x$ (and independently of $y$), the subgraph $G_{x, y}[V_y]$ is completely determined by $y$ (and independently of $x$), and the cut $(V_x, V_y)$ is small, i.e., has $O(\log n)$ edges.
It is then shown that if there is an algorithm $\mathcal{A}$ for solving MVC (exactly) in the $\ktzero$ \congest{} model, then Alice and Bob can solve \sd{} on $x$, $y$ by simulating $\mathcal{A}$.
Specifically, Alice and Bob start by respectively constructing $G_{x, y}[V_x]$ and $G_{x, y}[V_y]$ using their private inputs. They then simulate $\mathcal{A}$ round-by-round, communicating with each other only when algorithm $\mathcal{A}$ sends a message from a node in $V_x$ to $V_y$ (or vice versa).
Since the $(V_x, V_y)$ cut has size $O(\log n)$, this means that if $\mathcal{A}$ runs in $T$ rounds, Alice and Bob can solve MVC on $G_{x, y}$ by communicating $O(T \cdot \log^2 n)$ bits.
Finally, the linear (in length of $x$ and $y$) lower bound on the communication complexity of \sd{} \cite{razborov92_distr_compl_disjoin} (even for randomized, Monte Carlo algorithms)
implies an $\tilde{\Omega}(n^2)$ round lower bound on $T$.
The approach for showing an $\tilde{\Omega}(n^2)$ round lower bound for MDS \cite{BachrachCDELP19} is quite similar, the only difference being the construction of the lower bound graph $G_{x,y}$.

We extend the above approach to obtain an $\tilde{\Omega}(n^3)$ message lower bound using a key new idea. Our idea is to ``stretch'' the 
$(V_x, V_y)$ cut by adding vertex subsets $V_2, V_3, \ldots, V_{\ell-1}$ to the graph $G_{x, y}$. 
Renaming $V_x$ as $V_1$ and $V_y$ as $V_{\ell}$, we then replace the edges in the original cut $(V_x, V_y)$ by edges between $(V_i, V_{i+1})$ for $1 \le i \le \ell-1$. 
The size of each cut $(V_i, V_{i+1})$ is still small, i.e., $O(\polylog(n))$ edges.
The first challenge we overcome is showing that the correctness of the mapping from \sd{} instances $(x, y)$ to MVC instances $G_{x, y}$ is preserved. 
We now explain the motivation for ``stretching'' the cut. Suppose there is an algorithm $\mathcal{A}$ that solves MVC on $G_{x, y}$, while sending only $o(n^2/\polylog(n))$ messages across a constant-fraction of the $\ell-1$ cuts $(V_i, V_{i+1})$. 
Then Alice and Bob can simulate $\mathcal{A}$ with low communication complexity. Specifically, Alice and Bob can coordinate to (randomly) pick one of the low-message cuts $(V_i, V_{i+1})$. Alice starts by constructing the subgraph of $G_{x,y}$ induced by $V_1 \cup V_2 \cup \ldots \cup V_i$ and similarly Bob constructs the subgraph of $G_{x,y}$ induced by $V_{i+1} \cup V_{i+2} \cup \ldots \cup V_\ell$. Alice and Bob can then simulate $\mathcal{A}$ round-by-round, communicating with each other only when $\mathcal{A}$ needs to send a message from $V_i$ to $V_{i+1}$ (or vice versa). 
Since $o(n^2/\polylog(n))$ messages are sent across the $(V_i, V_{i+1})$ cut, this means that Alice and Bob only communicate $o(n^2)$ bits.
Because of the mapping from \sd{} instances to MVC instances, Alice and Bob can use the solution to MVC obtained by simulating $\mathcal{A}$, to solve \sd{} in $o(n^2)$ bits, something that is not possible. 
This implies that algorithm $\mathcal{A}$ necessarily sends $\tilde{\Omega}(n^2)$ messages across a constant-fraction of the $\ell-1$ cuts $(V_i, V_{i+1})$. By setting $\ell = \tilde{\Theta}(n)$, we obtain an $\tilde{\Omega}(n^3)$ message lower bound for $\mathcal{A}$.

The above high-level description glosses over several technical challenges. One of these is the fact that the 2-party communication complexity is asynchronous, i.e., Alice and Bob have no common notion of time, whereas we are interested in proving lower bounds for the \textit{synchronous} $\ktzero$ \congest{} model. However, in such a synchronous model of distributed computing, the \textit{time-encoding} trick can be used to reduce messages. For example, a node can stay silent for many clock ticks and then send a single bit at clock tick $t$ to a neighbor, thereby using just 1 bit of actual information to implicitly convey $\log t$ bits of information.
To overcome this challenge, we first use the above argument in a synchronous version of the 2-party communication model, showing that Alice and Bob can simulate algorithm $\mathcal{A}$ using a small number of bits in the synchronous 2-party communication model. We then appeal to a result from \cite{PanduranganPS20} which shows that in 2-party communication models, synchrony can be used to compress messages, but only by a $\log(r)$-factor for $r$-round algorithms. Applying this result allows us to translate the communication complexity in the synchronous 2-party communication model to communication complexity in the standard 2-party communication model with a logarithmic-factor loss, if we restrict ourselves to algorithms that run in polynomial rounds. 

We end this subsection by summarizing the scope of these cubic lower bounds.
First, they only hold in the \congest{} model, and not in the \local{} model, because the lower bound technique described above relies on edges having low bandwidth. In fact, it is easy to see that any problem can be solved using $O(n^2)$ messages in the \local{} model because a single node can gather the entire graph topology, and broadcast it to all nodes, using $O(n^2)$ messages.
Second, our use of communication complexity techniques to obtain lower bounds in the synchronous setting implies that our cubic lower bounds only hold for algorithms that run in polynomial rounds. 
Third, while the above argument has been sketched in the $\ktzero$ \congest{} model, it can be generalized to work in the $\ktone$ \congest{} model as well. We present this generalized argument in the appendix.

\paragraph*{B. Tight Quadratic Lower Bounds for Approximate Computation:}
Recall that we show quadratic message lower bounds for approximation algorithms not just in the $\ktzero$ \congest{} model, but even in the $\ktzero$ \local{} model, and our bounds hold not just for polynomial-round algorithms, but \textit{unconditionally}, i.e., even for algorithms that use arbitrarily many rounds.  
Unfortunately, communication-complexity-based approaches cannot be used for these types of powerful lower bounds.
Communication complexity reductions typically show a lower bound of, say $\Omega(b)$ bits, on the volume
of information that travels across a cut in the graph in any algorithm for the problem. However, in the \KTZero{} \local{} model, this does not
translate to a message complexity lower bound because there is no upper bound on the bandwidth of an edge
and in fact $b$ or more bits can travel across an edge in a \textit{single} message in the \KTZero{} \local{} model.
Furthermore, as mentioned previously, since information can be encoded in clock ticks, 
communication-complexity-based lower bounds degrade with the number of rounds.
So any message complexity lower bound obtained in the \textit{synchronous} $\ktzero$ \congest{} or \local{} model necessarily only applies to algorithms that are round-restricted.
For these reasons we use approaches different from communication-complexity-based techniques.

Our first technique, which we apply to the \mxm problem, involves showing that finding a large ``planted matching'' in the network is impossible without $\Omega(n)$ of the nodes identifying incident ``planted matching'' edges.
Further, we show using a symmetry argument that identifying a specific edge incident on a node requires messages
to pass over many incident edges.
In general, we show in Theorem~\ref{thm:mm_lb} that there is an inherent dependency between the number of discovered ``planted matching'' edges and the message complexity of any algorithm for approximating a maximum matching.
We believe that this technique could be of independent interest because it can be used to show the difficulty
of identifying other ``planted subgraphs'' in the \KTZero{} model using few messages. This in turn can lead to message  complexity lower bounds in the \KTZero{} \local{} model for other problems.

For the other problems, namely MDS, MVC, and \mxis, we use the so-called \emph{edge-crossing} technique,
that has been used to prove a variety of distributed computing lower bounds (see \cite{KorachMZSICOMP1987,AwerbuchGPV90, kuttenjacm15, abboud17_foolin_views, patt-shamir17_proof_label_schem} for some examples). 
For MVC and \mxis, our constructions build upon the lower bound graphs used in \cite{PaiPPRPODC2021} for proving message complexity lower bounds for MIS and $(\Delta+1)$-coloring.
Our use of this technique for MDS, which heavily borrows from communication-complexity-based lower bound constructions, seems novel. 
Below we sketch the 2-step approach we use to obtain the $\Omega(n^2)$ message lower bound for a $(\frac{5}{4}-\epsilon)$-approximation 
for MDS in the \KTZero{} \local{} model.
\begin{itemize}
    \item[(i)] 
    In \cite{BachrachCDELP19} the authors present a reduction from the 2-party communication complexity problem
    \textsc{SetDisjointness} to MDS and use this to show an $\tilde{\Omega}(n^2)$ round lower bound on computing an exact MDS in the \congest{} model. For their round lower bound argument, they construct a family of lower bound graphs that have a small cut --- of size $O(\polylog(n))$ --- across which $\Omega(n^2)$ bits have to flow. This small cut is needed to translate the lower bound on the number of bits to a round complexity lower bound. But, to show message complexity lower bounds, we do not need a small cut and this provides much greater flexibility in the construction of the lower bound graph family.
    For positive integer $n$, $x, y \in \{0, 1\}^{n^2}$, the construction in \cite{BachrachCDELP19} maps the instance $(x, y)$ of \textsc{SetDisjointness} to a graph $G_{x,y}$ with $n$ vertices and $\Theta(n^2)$ edges. 
    We take advantage of the flexibility mentioned above and extend the lower bound construction in \cite{BachrachCDELP19} to create a relatively large gap in the size of the MDS in graphs $G_{x,y}$ for which $\textsc{SetDisjointness}(x, y) = \mathrm{FALSE}$ versus
    graphs $G_{x,y}$ for which $\textsc{SetDisjointness}(x, y) = \mathrm{TRUE}$\footnote{This relatively large gap is created by forcing small minimum dominating sets in both cases; 4 when $\textsc{SetDisjointness}(x, y) = \mathrm{FALSE}$ and 5 when $\textsc{SetDisjointness}(x, y) = \mathrm{FALSE}$.}.
    However, at this stage this is still a communication-complexity-based reduction, and as observed earlier we cannot obtain unconditional \KTZero{} message complexity lower bounds via this construction.
\item[(ii)] Our goal now is to circumvent the need for a communication-complexity-based reduction, while still using this lower bound graph construction. To achieve this goal, we pick a graph $G = G_{x,y}$, $x, y \in \{0, 1\}^{n^2}$ for which $\textsc{SetDisjointness}(x, y) = \mathrm{TRUE}$. 
    We then show that for many pairs of edges $(e, e')$ in $G$, the graph $G(e,e')$ obtained by ``crossing'' the edges $e$ and $e'$ satisfies the property that $G(e, e') = G_{x', y'}$ for $x', y' \in \{0, 1\}^{n^2}$ where $\textsc{SetDisjointness}(x', y') = \mathrm{FALSE}$. The gap in the MDS sizes mentioned earlier implies that the MDS sizes in $G$ and $G(e, e')$ are relatively different.
    Finally, we rely on the well-known feature of ``edge-crossing'' arguments, which is that an algorithm that does not send messages on $e$ and $e'$ cannot distinguish between $G$ and $G(e, e')$. 
    This leads to the result (see Theorems \ref{thm:kt0-mds-det-lb} and \ref{thm:kt0-mds-random-lb}) that $\tilde{\Omega}(n^2)$ messages are needed to obtain an $(5/4-\epsilon)$-approximation for MDS, for any $\epsilon > 0$, in $\ktzero$ $\local$ model. \end{itemize}

We end this subsection by briefly mentioning our technique for obtaining $\tilde{O}(n^2)$ message upper bounds for $(1 \pm \epsilon)$-approximations for all 4 problems in the $\ktzero$ \congest{} model. 
A ``ball growing'' approach has been widely used in distributed computing for problems such as network decomposition in \cite{LinialSaksCombinatorica1993,AwerbuchPelegFOCS1990,MillerPengXuSPAA2013}.
Combining this approach with local exponential-time computations, Ghaffari, Kuhn, and Maus \cite{GhaffariKMSTOC2017} devised $(1 \pm \epsilon)$-approximations algorithms for covering and packing integer linear programs in the $\ktzero$ \local{} model, running in polylogarithmic rounds.
Since all 4 problems we consider are instances of covering and packing integer linear programs, the results in \cite{GhaffariKMSTOC2017} apply to these problems.
Our contribution is to show that this $\ktzero$ \local{} algorithm can also be implemented in the $\ktzero$ \congest{} model (i.e., using small messages) using only $\tilde{O}(n^2)$ messages, while running in polynomial time.

\paragraph*{C. Tight Linear Bounds for Random Graphs.} The $\Omega(n^2)$ message lower bounds that we show hold  on some specifically constructed graph families. We complement our lower  bounds by presenting essentially linear, i.e.,
$\tilde{O}(n)$, message complexity  algorithms on $G(n,p)$ (Erd\"{o}s-R\'{e}nyi) random graphs\cite{bollobas}
for all four problems that work (even) in the $\ktzero$ \congest{} model and give constant-factor approximations with high probability. Our message bounds are essentially tight, since it is easy to see that $\Omega(n)$ is a message lower bound
for all these problems. Furthermore, all our algorithms are fast, the algorithms for MVC, MDS, and \mxis run in $O(\log^2 n)$ rounds, whereas the \mxm algorithm runs in  $O(1)$ rounds.  
These results apply for all $G(n,p)$ random graphs
above the connectivity threshold, i.e., $p = \Theta(\log n/n)$ (see Section \ref{section:UBRG}). These results imply
that  \emph{almost all} graphs\footnote{Since we show high probability bounds on $G(n,p)$ for every $p$
above the connectivity threshold, one can interpret bounds
on random graphs in a deterministic manner as applying to all graphs (of every edge density), except for a vanishingly small fraction. For example, $G(n,1/2)$ is a uniform distribution on all graphs of size $n$ and our bounds show that almost all graphs admit $\tilde{O}(n)$ message algorithms for the four problems.} of every edge density (above $\Theta(\log n)$) admit very message-efficient (essentially linear) algorithms for   exact (for \mxm{}) or constant-factor approximation (for \mxis{}, MDS, and MVC). 
In other words, this means for the vast majority of graphs one needs to use a small fraction of
the edges to solve these problems. We also show that in general  graphs (Theorem \ref{th:mm}), maximum matching can be solved in  $O(n)$ messages and $O(1)$ rounds in $\ktzero$ \congest{} giving an expected $O((\Delta/\delta)^2)$-factor approximation, where  $\Delta$ and $\delta$  are the maximum and minimum degrees of the graph.

 Our main technical contribution is to show that the randomized greedy MIS algorithm can be implemented in random graphs using $\tilde{O}(n)$ messages and in $O(\log^2 n)$ rounds (where the first bound holds with high probability). This implies constant-factor distributed approximation for \mxis{}, MVC, and MDS within the same bounds. We note that while random graphs (above the connectivity threshold) have low diameter (i.e., $O(\log n)$), our $\tilde{O}(n)$ upper bounds are not exclusively due to this property. To compare, we point out that all of the lower bound graphs we construct for quadratic lower bounds for approximation algorithms have constant diameter.

\subsection{Related Work}
\label{section:related}

Significant progress has been made in understanding and improving the \textit{round} complexity of fundamental
``local'' problems such as MIS, maximal matching, $(\Delta+1)$-coloring, and ruling sets (see e.g., \cite{coloring-book, barenboim14_distr_delta_color_linear_delta_time, barenboim16_local_distr_symmet_break, barenboim12_local_distr_symmet_break, ChangLiPettieSICOMP2020, GhaffariKMSTOC2017, GhaffariSODA2016, RozhonGSTOC2020, GhaffariSODA19, HalldorssonKuhnMausTonoyanSTOC2021, BishtKP13, KothapalliP12, pai17_symmet_break_conges_model, PaiPPRPODC2021}) in both the \local{} and \congest{} models. 
This research on ``local'' problems has nice connections to distributed approximation for graph optimization problems and more recently this has become a highly active area of research. This line of research includes round complexity upper bounds for constant-factor and $(1-\epsilon)$-factor approximation for \mxm \cite{barenboim16_local_distr_symmet_break, Bar-YehudaCensor-HillelGhaffariSchwartzmanPODC2017, FischerDC2020, FischerMUSTOC2022, KMWJACM2016, LotkerPatt-ShamirPettieJACM2015},
constant-factor approximations for MVC \cite{Bar-YehudaCSJACM2017, KMWJACM2016} and logarithmic-approximations for MDS \cite{Censor-HillelD18, JiaRS02, DKM19, GK18, KMWJACM2016, PanditPemmarajuPODC2010}. It also includes round complexity lower bounds for solving MVC, MDS, and \mxis{}, approximately \cite{KMWJACM2016, KMWPODC2004,EfronGKPODC20} as well as exactly \cite{BachrachCDELP19, Censor-HillelKP17}.

We compare and contrast only those results in communication complexity that are relevant to our work. The classical 2-party communication complexity where two parties communicate via an (asynchronous) link has been studied extensively for computing various boolean functions, including equality, set disjointness etc. See the books of \cite{rao_yehudayoff_2020,KushilevitzN:book96} for a detailed treatment. The work of Tiwari \cite{TiwariJACM1987} studies the 2-party communication complexity where the two players are connected
by an arbitrary network. The network is assumed to be asynchronous. Tiwari shows that the lower bound on the communication complexity of {\em deterministic} protocols in a $n$-node network can be $\tilde{\Omega}(n)$ times the standard 2-party communication complexity (where the two
players are connected by a direct link).  The high-level idea of Tiwari's lower bound  is relating  the communication complexity in a network to that of a  single-link setting  by arguing that the two players have to communicate across several vertex-disjoint cuts. The ``stretching technique'' we use to obtain cubic lower bounds uses similar ideas, but we make this work even for {\em randomized algorithms}, and furthermore, we also circumvent the time-encoding trick mentioned earlier, so that our bounds also hold for synchronous algorithms.

Another work that is relevant to ours is that of  Chattopadhyay,  Radhakrishnan, and Rudra \cite{ChattopadhyayRRFOCS2014}
who study multi-party communication complexity where the $k$ players are connected by an arbitrary network (see also
related follow-up works \cite{rudra1,rudra2}).
This work builds on the earlier work of Woodruff and Zhang \cite{WoodruffZhangDIST2013} who study
the same problems, but under the assumption that the network is a clique. In these works, 
in the context of  graph problems, the input graph --- which is {\em different} from the communication network (which 
also has $k$ nodes) ---
is {\em edge-partitioned} among the $k$ players and the goal is to compute some property of the input graph, e.g., whether it is connected or bipartite etc. Chattopadhyay et al.~show that the lower bounds on 
the multi-party communication complexity of such problems on an arbitrary network connecting the $k$ players  is at least $\Omega(k)$ times the same complexity
when the players are connected by a clique. They also  exploit communication over several disjoint cuts to
show their stronger lower bounds. While their results hold also for randomized protocols (unlike the results of Tiwari \cite{TiwariJACM1987}),
their results do not apply to the synchronous setting. It is easy to show that in this setting, one can solve
all their problems in $O(m)$ messages, where $m$ is the number of edges of the communication network.
It is important to stress that the above results hold only when the input graph is {\em edge partitioned}. Providing the input graph using a vertex partition is closer to the distributed computing setting where one can associate vertices of the input graph (and their incident edges) to players. Indeed this is the assumption used in distributed computing models such as the congested clique \cite{HegemanPPSSPODC15, HegemanPTCS15, lotker2005mstJournal, DruckerKOPODC2014} and $k$-machine models \cite{bandyapadhyay18_near_optim_clust,KlauckNPRSODA15,PanduranganRS21,BandyapadhyayIPPTCS2022}.
Generally, lower bounds in the vertex partition setting are harder to show and in fact, 
Drucker, Kuhn, and Oshman \cite{DruckerKOPODC2014} prove that showing non-trivial lower bounds in the congested clique model will imply breakthrough circuit complexity lower bounds.

 \section{Tight Cubic Bounds for Exact Computations}
\label{section:cubicLowerBounds}
In this section we present near cubic, i.e. $\tilde{\Omega}(n^3)$ lower bounds on the message complexity of \ktzero{} \congest{} algorithms that compute exact solutions to MVC and \mxis{} in Section~\ref{section:exact-mvc-lb} and exact solution to MDS in Section~\ref{section:exact-mds-lb}.

We first present a generic framework for proving message complexity lower bounds in \ktzero{} \congest{} model by reduction from $2$-party communication complexity lower bounds. We begin by defining a lower bound graph family which we call an \emph{$\ell$-separated family of lower bound graphs}.
Definition~\ref{def:lb-graph-family} is a generalization of the lower bound graph family defined in \cite{Censor-HillelKP17} which is used in to obtain round complexity lower bounds in the $\congest$ model. In particular, the family defined in \cite{Censor-HillelKP17} is a $2$-separated family of lower bound graphs (i.e. they only consider $\ell=2$).

\begin{definition}[$\ell$-Separated Family of Lower Bound Graphs]\label{def:lb-graph-family}
    Let $f:X \times Y \to \{\mathrm{TRUE},\mathrm{FALSE}\}$ be a function and $P$ be a graph predicate. For an integer $\ell>1$, a family of graphs $\{G_{x,y}=(V,E_{x,y})\mid x \in X, y \in Y\}$ is said to be an \emph{$\ell$-separated family of lower bound graphs w.r.t. $f$ and $P$} if $V$ can be partitioned into $\ell$ disjoint and non-empty subsets $V_1, V_2,\dots, V_{\ell-1}, V_\ell$ such that the following properties hold:
\begin{enumerate}
  \item \label{lb-fw:alice} Only the existence or the weight of edges in $V_1 \times V_1$ depend on $x$;
  \item \label{lb-fw:bob} Only the existence or the weight of edges in $V_\ell \times V_\ell$ depend on $y$;
  \item \label{lb-fw:cuts} For all $1 \le i \le \ell$, the vertices in $V_i$ are only connected to vertices in $V_{i-1} \cup V_i \cup V_{i+1}$ (where $V_0 = V_{\ell+1} = \emptyset$). 
  \item \label{lb-fw:pred} $G_{x,y}$ satisfies the predicate $P$ iff $f(x,y) = \mathrm{TRUE}$.
\end{enumerate}
\end{definition}

We will now show a theorem (see Theorem \ref{thm:general-lb-framework}) which says that the existence of an $n$-vertex $\ell$-separated family of lower bound graphs w.r.t. $f$ and $P$ implies a lower bound of roughly $\ell \cdot CC(f)$ on the message complexity of a \ktzero{} \congest{} algorithm for deciding $P$, where $CC(f)$ is the 2-party communication complexity of the function $f$. We are ignoring many technical details in the previous statement for the sake of intuition, and we will spend the rest of the section adding these details. Note that $CC(f)$ is trivially bounded by $O(n^2)$ since $f(x,y)$ can be decided by evaluating $P$ on the $n$-vertex graph $G_{x,y}$, which can be represented using $n^2$ bits. Therefore, the extra $\ell$ factor crucially allows us to prove $\omega(n^2)$ lower bounds on message complexity.

We will prove Theorem \ref{thm:general-lb-framework} by efficiently simulating a \congest{} algorithm in the 2-party communication complexity model. In the standard 2-party model, there are two entities, usually called Alice and Bob. Alice has an input $x \in X$, unknown to Bob, and Bob has an input $y \in Y$, unknown to Alice. They wish to collaboratively compute a function $f(x, y)$ by following an agreed-upon protocol $\Pi$, which can be possibly randomized with error probability $\varepsilon$. The \emph{communication complexity} of this protocol $CC(\Pi)$ is the number of bits communicated by the two parties for the worst-case choice of inputs $x \in X$ and $y \in Y$. The \emph{deterministic communication complexity} of the function $f$, denoted as $CC^{\mathrm{det}}(f)$ is the minimum communication complexity of the deterministic protocol $\Pi$ that correctly computes $f$. And the \emph{randomized $\varepsilon$-error communication complexity} of the function $f$, denoted as $CC^{\mathrm{rand}}_{\varepsilon}(f)$ is the minimum communication complexity of the randomized protocol $\Pi$ that correctly computes $f$ with error probability at most $\varepsilon$. It is important to notice that this simple model is inherently asynchronous, since it does not provide the two parties with a common clock. 

Since the \congest{} model is synchronous, it is helpful to first simulate the $\congest$ algorithm in the \emph{synchronous 2-party model}, where the two parties also have a common clock. The time interval between two consecutive clock ticks is called a round. The computation proceeds in rounds: at the beginning of each synchronous round, (1) both parties send (possibly different) messages to each other, (2) both parties then receive the messages sent to it in the same round, and (3) both parties perform local computation, which will determine the messages it will send in the next round. The synchronous communication complexity $SCC(\Pi)$ of an $r$-round protocol $\Pi$ is the total number of bits sent by the two parties to compute $f(x,y)$ for the worst-case choice of inputs $x \in X$ and $y \in Y$. Note that if Alice (or Bob) decides to not send a message in a particular round, it does not contribute to the communication complexity, but Bob (or Alice) still receives some information, i.e., the fact that Alice (or Bob) chose to remain silent in this round. The \emph{deterministic $r$-round synchronous communication complexity} of the function $f$, denoted as $SCC^{\mathrm{det}}_r(f)$ is the minimum communication complexity of an $r$-round deterministic protocol $\Pi$ that correctly computes $f$. And the \emph{randomized $\varepsilon$-error $r$-round synchronous communication complexity} of the function $f$, denoted as $SCC^{\mathrm{rand}}_{r,\varepsilon}(f)$ is the minimum communication complexity of the $r$-round randomized protocol $\Pi$ that correctly computes $f$ with error probability at most $\varepsilon$.

We use a known relation between $CC^{\mathrm{det}}(f)$ and $SCC^{\mathrm{det}}_r(f)$, and between $CC^{\mathrm{rand}}_{\varepsilon}(f)$ and $SCC^{\mathrm{rand}}_{r,\varepsilon}(f)$. To do so, we use the Synchronous Simulation Theorem (SST) from~\cite{PanduranganPS20} to convert the synchronous 2-party protocol into an asynchronous 2-party protocol. Note that although \cite{PanduranganPS20} considers more than two parties, it also applies to the 2-party communication complexity setting. The below Lemma \ref{lemma:SST} is a simplified restatement of SST (Theorem 2 in~\cite{PanduranganPS20}), obtained by setting the number of parties to $2$ in both the synchronous and asynchronous models.

\begin{lemma}[Theorem 2 in~\cite{PanduranganPS20}]\label{lemma:SST}
    Let $f:X \times Y \to \{\mathrm{TRUE},\mathrm{FALSE}\}$ be a function that requires $CC^{\mathrm{rand}}_{\varepsilon}(f)$ bits for $\varepsilon$-error randomized protocols in the asynchronous 2-party communication complexity model, and $SCC^{\mathrm{rand}}_{\varepsilon,r}(f)$ bits for $\varepsilon$-error randomized, $r$-round protocols in the synchronous 2-party communication complexity model. The following relation holds (also for the deterministic setting),
    $$SCC^{\mathrm{rand}}_{\varepsilon,r}(f) = \Omega\left(\frac{CC^{\mathrm{rand}}_{\varepsilon}(f)}{1 + \log r}\right).$$
\end{lemma}

We are now ready to state and prove Theorem~\ref{thm:general-lb-framework}. We use the existence of an $\ell$-separated family of lower bound graphs w.r.t.~$f$ and $P$ to show that a \ktzero{} \congest{} algorithm that uses $r$ rounds and $M$ messages implies a synchronous randomized $2$-party protocol that uses $O(r)$ rounds and roughly $O(M/\ell)$ messages. In other words, the existence of the family implies that $SCC^{\mathrm{rand}}_{\varepsilon,r}(f) = O(M/\ell)$. Therefore any lower bound on the synchronous communication complexity gives a lower bound on $M$ that is a factor $\ell$ larger. The synchronous $2$-party protocol is pretty straightforward: Alice and Bob agree on a cut $(V_A, V_B)$ where $V_A = V_1,\dots, V_{i}$ and $V_B = V_{i+1}, \dots, V_\ell$, where $i$ is chosen uniformly at random in $\{1,\dots, \ell-1\}$, and they simulate the \congest{} algorithm, one round at a time, by exchanging the messages sent across this cut. Then we use Lemma~\ref{lemma:SST} to turn this into an asynchronous randomized $2$-party protocol. 

\begin{theorem}\label{thm:general-lb-framework}
    Fix a function $f:X \times Y \to \{\mathrm{TRUE},\mathrm{FALSE}\}$, a predicate $P$, a constant $0 < \delta < 1$, and a positive integer $\ell>1$. Suppose there exists an $\ell$-separated family of lower bound graphs $\{G_{x,y}=(V,E_{x,y})\mid x \in X, y \in Y\}$ w.r.t.~$f$ and $P$. Then any $r$-round deterministic (or randomized with error probability at most constant $0 < \varepsilon < 1$) algorithm for deciding $P$ in the \ktzero{} \congest{} model has message complexity 
\[M = \Omega \left(\frac{(\ell-1)}{\log |V|}\cdot\frac{ CC^{\mathrm{rand}}_{\delta + \varepsilon}(f)}{(1 + \log r)} - \frac{\ell \log \ell}{\log |V|}\right).\]
\end{theorem}
\begin{proof}
    Let $\mathcal{A}$ be an $r$-round and $M$-message deterministic (or randomized with error probability $0 < \varepsilon < 1$) \ktzero{} $\congest$ algorithm that decides $P$. We first simulate $\mathcal{A}$ in the synchronous $2$-party communication complexity model in order to find the value of $f(x,y)$ in $(r+1)$-rounds and $\log \ell + (M \log |V|) / \delta(\ell-1)$ bits of communication. This simulation has an error probability of at most $\delta + \varepsilon$, and along with the SST result (Lemma~\ref{lemma:SST}) gives us,
    \[\log \ell + \frac{M \log |V|}{\delta(\ell-1)} = \Omega \left(\frac{CC^{\mathrm{rand}}_{\delta+\varepsilon}(f)}{(1 + \log r)}\right) \implies M = \Omega \left(\frac{\delta (\ell-1)}{\log |V|} \cdot \frac{CC^{\mathrm{rand}}_{\delta+\varepsilon}(f)}{(1 + \log r)} - \frac{\delta \ell \log \ell}{\log |V|} \right)\]
    
    The simulation proceeds as follows: Alice and Bob together create $G_{x,y}$, where Alice is responsible for constructing the edges/weights in $V_1 \times V_1$ and Bob is responsible for constructing the edges/weights in $V_\ell \times V_\ell$. The rest of $G_{x,y}$ does not depend on $x$ and $y$, and is constructed by both Alice and Bob. Alice picks a random number $i \in [1,\ell)$ and sends it to Bob, and they both fix the cut $(V_A, V_B)$ where $V_A = V_1,\dots, V_{i}$ and $V_B = V_{i+1}, \dots, V_\ell$. Then they simulate $\mathcal{A}$ round by round, and in each round only exchange the messages sent over the cut $(V_A,V_B)$. Since the messages are $O(\log |V|)$ bits, we can assume without loss of generality that each message, along with the content, also contains the ID of the source and the destination, with appropriate delimiters, so Alice and Bob know this information for each message.

    Property \ref{lb-fw:cuts} in Definition \ref{def:lb-graph-family} implies that of all the $\ell-1$ cuts that Alice and Bob can fix in $G_{x,y}$, no pair of cuts can have a common edge crossing them. In other words, all the $\ell-1$ cuts have a distinct set of edges crossing them. Since $\mathcal{A}$ uses $M$ messages overall, at most a $\delta$ fraction of these cuts can have at least $M/\delta(\ell-1)$ messages passing through them. Therefore, the probability of the bad event that Alice and Bob fix a cut with at least $M/\delta(\ell-1)$ messages passing through it is at most $\delta$. Moreover, if $\mathcal{A}$ is randomized, the simulation will have error probability at most $\varepsilon$. Therefore, with probability at least $(1-\delta-\varepsilon)$, Alice and Bob will know whether $G_{x,y}$ satisfies $P$ or not, and by property \ref{lb-fw:pred} of Definition \ref{def:lb-graph-family}, they know whether $f(x,y)$ is TRUE or FALSE.

    The simulation of $\mathcal{A}$ requires $r$ synchronous rounds and $(M \log |V|) / \delta(\ell-1)$ bits of communication. And at the beginning they use one round and $\log \ell$ bits to agree on the cut.
\end{proof}

\subsection{Cubic Lower Bounds for Exact MVC and \mxis}\label{section:exact-mvc-lb}
We define an $\ell$-separated lower bound graph family $\{G_{x, y} \mid x \in \{0, 1\}^{k^2}, y \in \{0, 1\}^{k^2}\}$ w.r.t. $f=\sd$ and predicate $P$ which can be decided by computing the MVC (we will describe $P$ more formally later). The \sd function is defined as: $\sd(x,y) = \mathrm{FALSE}$ iff there exists an index $i$ such that $x_i = y_i = 1$. We will assume $k$ is a power of $2$ so that $\log k$ is an integer.
Note that this construction is a generalization of the MVC construction in \cite{Censor-HillelKP17}. More precisely, their construction can be directly obtained from ours by setting $\ell = 2$.

For positive integer $k$, fix $x, y \in \{0, 1\}^{k^2}$. We define the graph $G_{x, y}$ as follows. The vertex set of $G_{x, y}$ is 
$$A_1 \cup A_2 \cup B_1 \cup B_2 \cup C_1 \cup C_2 \cup \dots \cup C_{2\log k}$$
where $A_1 = \{a_1^i \mid 1 \le i \le k\}$, $A_2 = \{a_2^i \mid 1 \le i \le k\}$,
$B_1 = \{b_1^i \mid 1 \le i \le k\}$ and $B_2 = \{b_2^i \mid 1 \le i \le k\}$ are called the ``row vertices''. And
$C_i = \{t_i^j, f_i^j \mid 1 \le j \le \ell\}$ for all $1\le i \le 2\log k$ are called the ``bit gadget vertices''. Therefore, $G_{x,y}$ has $4k+4\ell\log k$ vertices.

We now describe the edges of $G_{x, y}$. The vertices in the sets $A_1$ form a clique, and so do the vertices in the sets $A_2, B_1, B_2$. Assuming $\ell$ is even, the vertices in the set $C_i$ form a cycle with the following order:
$$f_i^1, t_i^2, f_i^3, t_i^4, f_i^5,\dots, t_i^{\ell-2}, f_i^{\ell-1}, t_i^\ell, f_i^\ell, t_i^{\ell-1}, f_i^{\ell-2}, \dots, t_i^5, f_i^4, t_i^3, f_i^2, t_i^1, f_i^1$$ This cycle is also illustrated for the set $C_1$ in Figure \ref{fig:MVC-lb}.
Each vertex $a_1^i \in A_1$ is connected to bit gadgets $C_1, \dots, C_{\log k}$ according to the binary representation of the index $i$. In particular, $a_1^i$ is connected to $t_h^1$ if bit position $h$ of index $i$ is $1$ and to $f_h^1$ if bit position $h$ of index $i$ is $0$. Similarly, each vertex $b_1^i \in B_1$ is connected to $t_h^\ell$ if the bit position $h$ of index $i$ is $1$ and to $f_h^\ell$ if the bit position $h$ of index $i$ is $0$. The vertices in $A_2$ and $B_2$ are connected to $C_{\log k + 1}, \dots, C_{2\log k}$ in a symmetric manner. To be explicit, each vertex $a_2^j \in A_2$ is connected to $t_{h + \log k}^1$ if bit position $h$ of index $j$ is $1$ and to $f_{h + \log k}^1$ if bit position $h$ of index $j$ is $0$. Similarly, each vertex $b_2^j \in B_2$ is connected to $t_{h + \log k}^\ell$ if bit position $h$ of index $j$ is $1$ and to $f_{h + \log k}^\ell$ if bit position $h$ of index $j$ is $0$.

This completes the ``fixed'' edges in the graph, i.e., the edges that do not depend on bit vectors $x$ and $y$.
The edges between $A_1$ and $A_2$ depend on $x$ as follows: $\{a_1^i, a_2^j\}$ is an edge iff $x_{ij} = 0$.
The edges between $B_1$ and $B_2$ depend on $y$ as follows: $\{b_1^i, b_2^j\}$ is an edge iff $y_{ij} = 0$.
The construction is illustrated in Figure~\ref{fig:MVC-lb}.

\begin{figure}
    \centering
    \includegraphics[width=\linewidth]{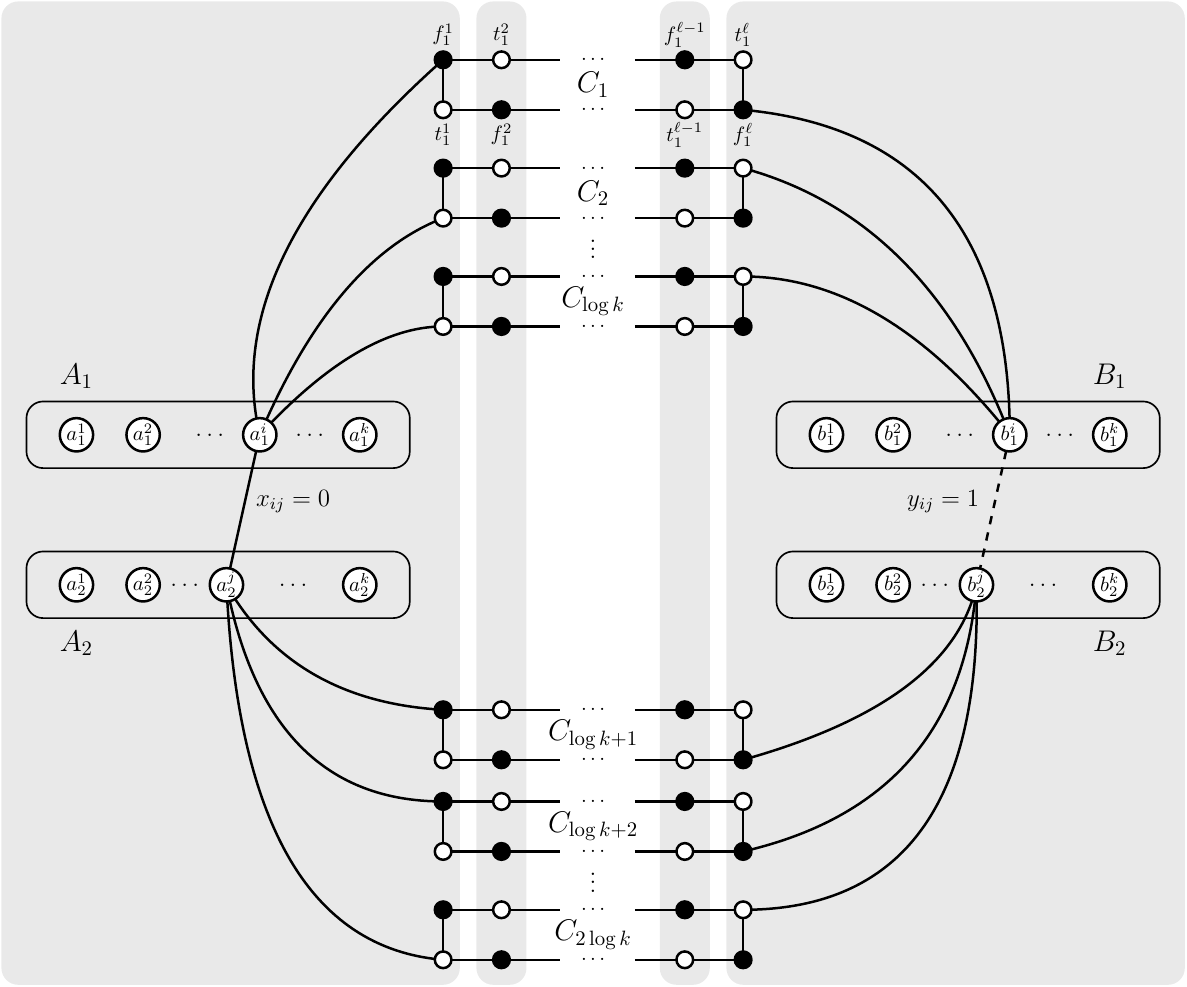}
    \caption{Illustration of one graph $G_{x,y}$ from the $\ell$-separated lower bound graph family we defined to show a cubic message lower bound for exact MVC. Many edges are omitted for the sake of clarity. The gray boxes from left to right denote the sets $V_1, V_2, \dots, V_\ell$ defined in proof of Theorem~\ref{thm:exact-mvc-cubic-lb}.}
    \label{fig:MVC-lb}
\end{figure}

\begin{claim} \label{claim:exact-mvc-gxy-size}
    Any vertex cover of $G_{x,y}$ must contain at least $k-1$ vertices from each set $A_1, A_2, B_1, B_2$ and at least $\ell$ vertices from each set $C_i$, $1 \le i \le 2\log k$.
\end{claim}
\begin{proof}
    Since the vertices in $A_1$ (and $A_2, B_1, B_2$) form a clique of size $k$, any vertex cover of $G_{x,y}$ must contain at least $k-1$ vertices from $A_1$ (and $A_2, B_1, B_2$ respectively) And since the bit gadget $C_i$ is a cycle on $2\ell$ vertices, any vertex cover of $G_{x,y}$ must contain at least $\ell$ vertices from $C_i$ to cover all the edges of the cycle.
\end{proof}

\begin{lemma}\label{lem:mvc-exact-lb-pred}
   For $x, y \in \{0,1\}^{k^2}$, if $\emph{\sd}(x, y) = \mathrm{TRUE}$ then the MVC of $G_{x, y}$ has size at least $4k+2\ell\log k-3$, and if $\emph{\sd}(x, y) = \mathrm{FALSE}$ then the MVC of $G_{x, y}$ has size exactly $4k+2\ell\log k-4$.
\end{lemma}
\begin{proof}
    If $\emph{\sd}(x, y) = \mathrm{TRUE}$ then there are no indices $i,j \in [1,k]$ such that $x_{ij} = y_{ij} = 1$. By Claim~\ref{claim:exact-mvc-gxy-size}, any vertex cover of $G_{x,y}$ must have size at least $4k+2\ell\log k-4$. Let us assume that a vertex cover $U$ of size exactly $4k+2\ell\log k-4$ exists in $G_{x,y}$. By Claim~\ref{claim:exact-mvc-gxy-size}, $U$ must contain at least $k-1$ vertices from each set $A_1, A_2, B_1, B_2$, and because of the restriction on size of $U$, it must contain exactly $k-1$ vertices from each set. Let $a_1^i, a_2^j, b_1^{i'}, b_2^{j'}$ be the vertices not in $U$. Note that we cannot simultaneously have $i=i'$ and $j=j'$ because at least one of the two edges $(a_1^i, a_2^j)$ and $(b_1^i, b_2^j)$ must exist in $G_{x,y}$ and $U$ covers neither of them. 
    
    So let us assume $i \neq i'$.
    The argument for the case $j \neq j'$ is symmetric.
    The edges between $a_1^i$ and $C_1, \dots, C_{\log k}$ are not covered by $a_1^i$, so $U$ must contain $t_h^1$ if bit position $h$ of index $i$ is $1$ and to $f_h^1$ if bit position $h$ of index $i$ is $0$. By Claim~\ref{claim:exact-mvc-gxy-size}, we only have a budget of $\ell$ vertices per cycle, so $U$ must contain $t_h^1, \dots, t_h^\ell$ if bit position $h$ of index $i$ is $1$ and to $f_h^1, \dots, f_h^\ell$ if bit position $h$ of index $i$ is $0$ to cover all the edges in the cycle formed by the vertices in $C_h$. But then $U$ will not cover all the edges incident on $b_1^{i'}$, since $i$ and $i'$ differ on at least one bit position. Therefore, $U$ must have at least one additional vertex in order to be a valid vertex cover, which means that the MVC of $G_{x, y}$ has size at least $4k+2\ell\log k-3$.
    
    If $\emph{\sd}(x, y) = \mathrm{FALSE}$ then there exists $i,j \in [1,k]$ such that $x_{ij} = y_{ij} = 1$. In this case both the edges $(a_1^i, a_2^j)$ and $(b_1^i, b_2^j)$ do not exist in $G_{x,y}$. Therefore the set of vertices $U = (A_1 \setminus \{a_1^i\}) \cup (B_1 \setminus \{b_1^i\}) \cup (A_2 \setminus \{a_2^j\}) \cup (B_2 \setminus \{b_2^j\})$ covers all edges between $A_1, A_2$ and all edges between $B_1, B_2$. 
    
    For all $1 \le h \le \log k$, if bit position $h$ of the binary representation of index $i$ is $1$, we add to $U$ the vertices $t_h^1, t_h^2, \dots, t_h^\ell$ and if bit position $h$ of the binary representation of index $i$ is $0$, we add to $U$ the vertices $f_h^1, f_h^2, \dots, f_h^\ell$. These $\ell\log k$ vertices cover all edges between $a_1^i, b_1^i$ and $C_1, \dots, C_{\log k}$ which were the only edges incident to $A_1, B_1$ that were previously not covered by $U$. Moreover, since we pick every other vertex in each $2\ell$-vertex cycle $C_1, \dots, C_{\log k}$, we also cover all edges incident on $C_1, \dots, C_{\log k}$.

    Symmetrically, for all $1 \le h \le \log k$, if bit position $h$ of the binary representation of index $j$ is $1$, we add to $U$ the vertices $t_{h+\log k}^1, t_{h+\log k}^2, \dots, t_{h+\log k}^\ell$ and if bit position $h$ of the binary representation of index $j$ is $0$, we add to $U$ the vertices $f_{h+\log k}^1, f_{h+\log k}^2, \dots, f_{h+\log k}^\ell$. These $\ell\log k$ vertices cover all edges between $a_2^j, b_2^j$ and $C_{1+\log k}, \dots, C_{2\log k}$ which were the only edges incident to $A_2, B_2$ that were previously not covered by $U$. Moreover, since we pick every other vertex in each $2k$-vertex cycle $C_{1+\log k}, \dots, C_{2\log k}$, we also cover all edges incident on $C_{1+\log k}, \dots, C_{2\log k}$. Hence $U$ is a vertex cover of $G_{x,y}$ of size $4k+2\ell\log k-4$ which is also the MVC of $G_{x,y}$ by Claim~\ref{claim:exact-mvc-gxy-size}.
\end{proof}

\begin{theorem}\label{thm:exact-mvc-cubic-lb}
    For any $0 < \varepsilon < 1/6$, any $\varepsilon$-error randomized Monte-Carlo $r$-round \ktzero{} \congest{} algorithm that computes an MVC or \mxis{} on an $n$-vertex communication graph has message complexity $\tilde{\Omega}(n^3/(1+\log r))$.
\end{theorem}
\begin{proof}Let $V_1 = A_1 \cup A_2 \cup \{t_i^1, f_i^1 \mid 1 \le i \le 2\log k\}$, $V_\ell = B_1 \cup B_2 \cup \{t_i^\ell, f_i^\ell \mid 1 \le i \le 2\log k\}$ and for $1 < j < \ell$, $V_j = \{t_i^j, f_i^j \mid 1 \le i \le 2\log k\}$. These sets are highlighted in gray in Figure~\ref{fig:MVC-lb}.

    Let $P$ be a predicate which is true for graphs with $4k+4\ell\log k$ vertices that have an MVC of size $4k + 2\ell\log k - 4$. Any \ktzero{} \congest{} algorithm that computes the MVC or \mxis{} of a graph can easily decide $P$ using an extra $O(n)$ rounds and $O(n^2)$ messages, by aggregating the MVC size to a leader and the leader broadcasting the answer to everyone. Note that for any graph with set of vertices $V$, if $C$ is an MVC then $V \setminus C$ is a \mxis{}.
    
    It is easy to see that the lower bound graph family $\{G_{x,y} \mid x,y \in \{0,1\}^{k^2}\}$ described above is an $\ell$-separated lower bound graph family w.r.t. function $f = \sd : \{0,1\}^{k^2} \times \{0,1\}^{k^2} \to \{\mathrm{TRUE}, \mathrm{FALSE}\}$ and predicate $P$. In particular, properties \ref{lb-fw:alice}, \ref{lb-fw:bob}, and \ref{lb-fw:cuts} of  Definition~\ref{def:lb-graph-family} are immediate consequences of the construction and Lemma~\ref{lem:mvc-exact-lb-pred} implies that property \ref{lb-fw:pred} is also satisfied. We substitute $k = n/8$ and $\ell = n/(8\log k)$ so that all graphs in the family have exactly $n$ vertices (for $n$ large enough so that $k$ is a power of $2$ and $\ell$ is an even integer). Therefore, we can apply Theorem~\ref{thm:general-lb-framework} with the parameters $\ell = n/(8\log (n/8))$, and $\delta = 1/6$. We use the well known fact that the $\delta + \varepsilon \le 1/3$-error randomized communication complexity $\sd$ on input size $k^2$ is $\Omega(k^2)$ in order to get that any randomized $\varepsilon$-error \ktzero{} \congest{} algorithm that computes an MVC of an $n$-vertex graph has message complexity \(\Omega \left(\frac{n^3}{(1 + \log r) \cdot \log^2 n}\right).\)
\end{proof}

\subsection{Cubic Lower Bound for Exact MDS}\label{section:exact-mds-lb}

Similar to the previous section, we define an $\ell$-separated lower bound graph family $\{G_{x, y} \mid x \in \{0, 1\}^{k^2}, y \in \{0, 1\}^{k^2}\}$ w.r.t. $f=\sd$ and predicate $P$ which can be decided by computing the MDS. The \sd function is defined as follows: $\sd(x,y) = \mathrm{FALSE}$ iff there exists an index $i$ such that $x_i = y_i = 1$. Note that this construction is an generalization of the MDS construction in \cite{BachrachCDELP19}. More precisely, their construction can be directly obtained from ours by setting $\ell = 2$.
Moreover, this MDS construction is almost exactly the same as the MVC construction described in the previous section, with a few important differences which we now describe.

For positive integer $k$, fix $x, y \in \{0, 1\}^{k^2}$. We will assume $k$ is a power of $2$ so that $\log k$ is an integer. We define the graph $G_{x, y}$ as follows. The vertex set of $G_{x, y}$ is 
$A_1 \cup A_2 \cup B_1 \cup B_2 \cup C_1 \cup C_2 \cup \dots \cup C_{2\log k}$
where $A_1 = \{a_1^i \mid 1 \le i \le k\}$, $A_2 = \{a_2^i \mid 1 \le i \le k\}$,
$B_1 = \{b_1^i \mid 1 \le i \le k\}$, and $B_2 = \{b_2^i \mid 1 \le i \le k\}$ are called the row vertices. And
$C_i = \{t_i^j, f_i^j, u_i^j \mid 1 \le j \le \ell\}$ are called the bit gadget vertices for all $1\le i \le 2\log k$. Therefore, $G_{x,y}$ has $4k+6\ell\log k$ vertices.

We now describe the edges of $G_{x, y}$. There are no edges between the row vertices in $A_1$ (and similarly $A_2, B_1, B_2$). Assuming $\ell$ is even, the bit gadget vertices $C_i$ form a cycle with the following order:
$$f_i^1, u_i^2, t_i^2, f_i^3, u_i^4, t_i^4, f_i^5, \dots, f_i^{\ell-1}, u_i^\ell, t_i^\ell, f_i^\ell, u_i^{\ell-1}, t_i^{\ell-1}, \dots, u_i^5, t_i^5, f_i^4, u_i^3, t_i^3, f_i^2, u_i^1, t_i^1, f_i^1$$ This cycle is also illustrated in  Figure~\ref{fig:MDS-lb-bit-gadget}.

Each vertex $a_1^i \in A_1$ is connected to bit gadgets $C_1, \dots, C_{\log k}$ according to the binary representation of the index $i$. In particular, $a_1^i$ is connected to $t_h^1$ if bit position $h$ of index $i$ is $1$ and to $f_h^1$ if bit position $h$ of index $i$ is $0$. Similarly, each vertex $b_1^i \in B_1$ is connected to $t_h^\ell$ if the bit $h$ of index $i$ is $1$ and to $f_h^\ell$ if the bit position $h$ of index $i$ is $0$. The vertices in $A_2$ and $B_2$ are connected to $C_{\log k + 1}, \dots, C_{2\log k}$ in a symmetric manner. To be explicit, each vertex $a_2^j \in A_2$ is connected to $t_{h + \log k}^1$ if bit position $h$ of index $j$ is $1$ and to $f_{h + \log k}^1$ if bit position $h$ of index $j$ is $0$. Similarly, each vertex $b_2^j \in B_2$ is connected to $t_{h + \log k}^\ell$ if bit position $h$ of index $j$ is $1$ and to $f_{h + \log k}^\ell$ if bit position $h$ of index $j$ is $0$.

This completes the ``fixed'' edges in the graph, i.e., the edges that do not depend on bit vectors $x$ and $y$.
The edges between $A_1$ and $A_2$ depend on $x$ as follows: $\{a_1^i, a_2^j\}$ is an edge iff $x_{ij} = 1$.
The edges between $B_1$ and $B_2$ depend on $y$ as follows: $\{b_1^i, b_2^j\}$ is an edge iff $y_{ij} = 1$.

\begin{figure}
    \centering
    \includegraphics[width=\linewidth]{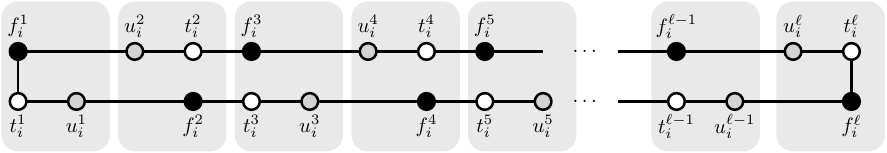}
    \caption{Illustration of one bit gadget $C_i$ in a graph from the $\ell$-separated lower bound graph family we defined to show a cubic message lower bound for exact MDS. The rest of the construction is visually similar to Figure~\ref{fig:MVC-lb}, so we do not illustrate those parts again. But there are important differences that are described in detail at the beginning of Section~\ref{section:exact-mds-lb}}
    \label{fig:MDS-lb-bit-gadget}
\end{figure}

\begin{claim} \label{claim:exact-mds-gxy-size}
    Any dominating set of $G_{x,y}$ must contain at least $\ell$ vertices from each set $C_i$, $1 \le i \le 2\log k$. In other words, a dominating set of $G_{x,y}$ has at least $2\ell \log k$ vertices.
\end{claim}
\begin{proof}
    For each $C_i$, all the $\ell$ vertices with label $u$ have a unique neighbor with label $t$ and a unique neighbor with label $f$. In order to dominate all vertices with label $u$, we need to pick one of these three vertices. Therefore, each dominating of $G_{x,y}$ has at least $2\ell \log k$ vertices and must contain at least $\ell$ vertices from each $C_i$.
\end{proof}

\begin{lemma}\label{lem:mds-exact-lb-pred}
   For $x, y \in \{0,1\}^{k^2}$, if $\emph{\sd}(x, y) = \mathrm{TRUE}$ then the MDS of $G_{x, y}$ has size at least $2\ell\log k + 3$, and if $\emph{\sd}(x, y) = \mathrm{FALSE}$ then the MDS of $G_{x, y}$ has size at most $2\ell\log k + 2$.
\end{lemma}
\begin{proof}
    If $\emph{\sd}(x, y) = \mathrm{TRUE}$ then there are no indices $i,j \in [1,k]$ such that $x_{ij} = y_{ij} = 1$. Let us assume for the sake of contradiction that in this case, $G_{x,y}$ has a dominating set $U$ of size $2\ell\log k + 2$. By Claim~\ref{claim:exact-mds-gxy-size}, $2\ell\log k$ of these vertices are required just to dominate the bit gadget vertices in $C_1, \dots, C_{2\log k}$. Therefore using only two ``extra'' vertices, $U$ needs to dominate all vertices in $A_1, A_2, B_1, B_2$. This means that for all but two sets in $C_1, \dots, C_{2\log k}$, $U$ must contain exactly $\ell$ vertices. Therefore, there are three cases based on the number of extra vertices in $U$ that belong to the bit gadgets $C_1, \dots, C_{2\log k}$.
    
    {\bf Case 1:} (no extra vertices) Here, $U$ contains exactly $\ell$ vertices from each set $C_1, \dots, C_{2\log k}$. For a cycle $C_h$, if we have a budget of $\ell$ vertices to put in $U$, we have three choices: we can pick all vertices labeled $u$, or all vertices labeled $t$, or all vertices labeled $f$. Vertices labeled $u$ do not dominate any vertices outside $C_h$, so we assume that for $C_h$ we pick either all vertices labeled $t$ or all vertices labeled $f$. For all values $h$ such that $1 \le h \le \log k$, if we arbitrarily pick either all vertices labeled $t$ or all vertices labeled $f$ in $C_h$, we will cover all vertices in $A_1 \setminus \{a_1^i\}$ and $B_1 \setminus \{b_1^i\}$, where $i$ is the index whose binary representation has $0$ in bit position $h$ where we picked all vertices labeled $t$ in $C_h$ and $1$ in bit position $h$ where we picked all vertices labeled $f$ in $C_h$. Similarly, we will cover all vertices in $A_2 \setminus \{a_2^j\}$ and $B_2 \setminus \{b_2^j\}$ for some index $1 \le j \le k$. Since there are no indices $i,j \in [1,k]$ such that $x_{ij} = y_{ij} = 1$, one of the two edges $(a_1^i, a_2^j)$ and $(b_1^i, b_2^j)$ does not exist in $G_{x,y}$. We need to dominate these four vertices, by picking at most two vertices in $A_1, A_2, B_1, B_2$ which is impossible.

    {\bf Case 2:} (one extra vertex) Here, $U$ contains $\ell+1$ vertices from some set $C_h$, $1 \le h \le 2\log k$. Let us first assume that $1 \le h \le \log k$, the argument for $h > \log k$ is symmetric. We can use this vertex to include both $t_h^1, f_h^1$ in $U$, thereby dominating all vertices in $A_1$. If we try to dominate the remaining vertices in $C_h$ with $\ell-1$ vertices, it turns out that we still need to choose either $t_h^\ell$ or $f_h^\ell$, but not both, to be a part of $U$. This leaves three undominated vertices $b_1^{i}, a_2^j, b_2^j$ for some indices $i, j$. Similarly, if we include both $t_h^\ell, f_h^\ell$ in $U$ to dominate all vertices in $B_1$, it still leaves three undominated vertices $a_1^{i}, a_2^j, b_2^j$ for some indices $i, j$. We need to add at least two additional vertices in $U$ to dominate these three vertices, which exceeds our size bound. Another possibility is to insert the extra vertex in such a way that $U$ contains $t_h^1, f_h^\ell$ (or $t_h^\ell, f_h^1$) and dominates all vertices in $C_h$. This means that the undominated vertices are $a_1^i, b_1^{i'}, a_2^j, b_2^j$, where $i \neq i'$ and both the edges $(a_1^i, a_2^j)$ and $(b_1^{i'}, b_2^j)$ exist in $G_{x,y}$. But even in this case, we need at least two more vertices in $U$ to dominate these four vertices, which exceeds our size bound. Therefore, even in this case, it is not possible for $U$ to dominate all vertices in $A_1, A_2, B_1, B_2$.
    
    {\bf Case 3:} (two extra vertices) In this case, both the extra vertices cannot belong to the same set $C_h$, as a single $C_h$ can only dominate either $A_1,B_1$ or $A_2,B_2$. For the same reason, one extra vertex must belong to $C_h$ for $1 \le h \le \log k$, and the other must belong to $C_h$ for $h > \log k$. A similar analysis as Case 2 shows us that $U$ cannot dominate all vertices in $A_1, A_2, B_1, B_2$, which is a contradiction. Therefore, the MDS of $G_{x, y}$ has size at least $2\ell\log k + 3$.
    
    If $\emph{\sd}(x, y) = \mathrm{FALSE}$ then there exists $i,j \in [1,k]$ such that $x_{ij} = y_{ij} = 1$. In this case both the edges $(a_1^i, a_2^j)$ and $(b_1^i, b_2^j)$ exist in $G_{x,y}$. In this case we construct a dominating set $U$ as follows: We first add to $U$, the vertices $a_1^i$ and $b_1^i$, to specifically dominate the vertices $a_1^i, a_2^j, b_1^i, b_2^j$.
    
    For all $1 \le h \le \log k$, if bit position $h$ of the binary representation of index $i$ is $0$, we add to $U$ the vertices $t_h^1, t_h^2, \dots, t_h^\ell$ and if bit position $h$ of the binary representation of index $i$ is $1$, we add to $U$ the vertices $f_h^1, f_h^2, \dots, f_h^\ell$. These $\ell\log k$ vertices dominate all vertices in the sets $C_1, \dots, C_{\log k}$. Moreover, these vertices also dominate the vertices in sets $A_1 \setminus \{a_1^i\}$ and $B_1 \setminus \{b_1^i\}$, since all of their indices differ from $i$ in at least one bit position.

    Symmetrically, for all $1 \le h \le \log k$, if bit position $h$ of the binary representation of index $j$ is $0$, we add to $U$ the vertices $t_{h+\log k}^1, t_{h+\log k}^2, \dots, t_{h+\log k}^\ell$ and if bit position $h$ of the binary representation of index $j$ is $1$, we add to $U$ the vertices $f_{h+\log k}^1, f_{h+\log k}^2, \dots, f_{h+\log k}^\ell$. By similar arguments as before, these $\ell\log k$ vertices dominate all vertices in the sets $A_2 \setminus \{a_2^j\}$, $B_2 \setminus \{b_2^j\}$, $C_{1 + \log k}, \dots, C_{2\log k}$.
    
    Hence $U$ is a dominating set of $G_{x,y}$ of size $2\ell\log k + 2$ which implies that the MDS of $G_{x,y}$ has size at most $2\ell\log k + 2$.
\end{proof}

\begin{theorem}\label{thm:exact-mds-cubic-lb}
    For any $0 < \varepsilon < 1/6$, any $\varepsilon$-error randomized Monte-Carlo $r$-round \ktzero{} \congest{} algorithm that computes an MDS on an $n$-vertex communication graph has message complexity $\tilde{\Omega}(n^3/(1+\log r))$.
\end{theorem}
\begin{proof}Let $V_1 = A_1 \cup A_2 \cup \{u_i^1, t_i^1, f_i^1 \mid 1 \le i \le 2\log k\}$, $V_\ell = B_1 \cup B_2 \cup \{u_i^\ell, t_i^\ell, f_i^\ell \mid 1 \le i \le 2\log k\}$ and for $1 < j < \ell$, $V_j = \{u_i^j, t_i^j, f_i^j \mid 1 \le i \le 2\log k\}$.

    Let $P$ be a predicate which is true only for graphs with $4k+6\ell\log k$ vertices that have an MDS size at most $2\ell\log k +2$. Any \ktzero{} \congest{} algorithm that computes the MDS of a graph can easily decide $P$ using an extra $O(n)$ rounds and $O(n^2)$ messages, by aggregating the dominating set size to a leader vertex and the leader broadcasting the answer to everyone.
    
    It is easy to see that the lower bound graph family $\{G_{x,y} \mid x,y \in \{0,1\}^{k^2}\}$ described above is an $\ell$-separated lower bound graph family w.r.t. function $f = \sd : \{0,1\}^{k^2} \times \{0,1\}^{k^2} \to \{\mathrm{TRUE}, \mathrm{FALSE}\}$ and predicate $P$. In particular, properties \ref{lb-fw:alice}, \ref{lb-fw:bob}, and \ref{lb-fw:cuts} of  Definition~\ref{def:lb-graph-family} are immediate consequences of the construction and Lemma~\ref{lem:mds-exact-lb-pred} implies that property \ref{lb-fw:pred} is also satisfied. We substitute $k = n/8$ and $\ell = n/(12\log k)$ so that all graphs in the family have exactly $n$ vertices (for $n$ large enough so that $k$ is a power of $2$ and $\ell$ is an even integer). Therefore, we can apply Theorem~\ref{thm:general-lb-framework} with the parameters $\ell = n/(12\log (n/8))$, and $\delta = 1/6$. We use the well known fact that the $\delta + \varepsilon \le 1/3$-error randomized communication complexity $\sd$ on input size $k^2$ is $\Omega(k^2)$ in order to get that any randomized $\varepsilon$-error \ktzero{} \congest{} algorithm that computes an MDS of an $n$-vertex graph has message complexity \(\Omega \left(\frac{n^3}{(1 + \log r) \cdot \log^2 n}\right).\)
\end{proof}

\begin{remark}
\label{remark:maxm}
One reason why this method does not give a near-cubic lower bound for exact \mxm{} is that we first have to show existence of an $\ell$-separated lower bound graph family with $n$-vertices for $\ell = n/\polylog(n)$ and $k = n^2/\polylog(n)$. However, with this setting of parameters, property \ref{lb-fw:cuts} of Definition~\ref{def:lb-graph-family}, implies that there must exist an index $1 \le i < \ell$ such that the cut $(V_A, V_B)$, where $V_A = V_1,\dots, V_{i}$ and $V_B = V_{i+1}, \dots, V_\ell$, has $\polylog(n)$ edges that cross it. (Otherwise, if every such cut has $\omega(\polylog n)$ edges, it implies $|V_i \cup V_{i+1}| = \omega(\polylog n)$, which is not possible as the number of vertices in the graph becomes larger than $n$.)
And thus, for this family of lower bound graphs, the framework of \cite{Censor-HillelKP17} would imply that any randomized algorithm that computes an exact solution to \mxm{} requires $\tilde{\Omega}(n^2)$ rounds. This is not possible as \cite{KitimuraIzumiITIS2022} show a randomized \ktzero{} \congest{} algorithm that computes an exact \mxm{} in $O(n^{1.5})$ rounds. So the family of lower bound graphs required to show near-cubic lower bound for exact \mxm{} does not exist.
\end{remark} \section{Tight Quadratic Bounds for Approximate Computations}
\label{section:KTZeroLB}

We start this section by proving $\tilde{O}(n^2)$ message complexity upper bounds for $(1 \pm \epsilon)$-approximations for all four problems, \mxm, MVC, MDS, and \mxis. These results serve as a contrast to the cubic lower bounds for exact computation, shown in the previous section.
We then show that these upper bounds are tight, by showing that $\tilde{\Omega}(n^2)$ messages are required for {\em constant-factor} approximation algorithms for all four problems. 
For \mxm{} and MVC, these bounds hold for any constant-factor approximation, whereas for MDS and \mxis{} they hold for any approximation factor better than $5/4$ and $1/2$ respectively. 
These lower bounds hold even in the \local{} model (in which messages can be arbitrarily large) and they apply not just to polynomial-round algorithms, but to algorithms that take arbitrarily many rounds. 

\subsection{Quadratic Upper Bounds for Approximate Computations}

\textbf{Notation:} For any graph $G$, let $\alpha(G)$ denote the size of the largest independent set in $G$. For any node $v$ in $G$ and integer $r \ge 0$, let $B_r(v)$ denote the set of all nodes in $G$ at distance at most $r$ from $v$.

Consider the following sequential algorithm, called the ``ball growing'' algorithm in \cite{GhaffariKMSTOC2017} that gives a $1/(1+\epsilon)$-approximate solution for \mxis{}, for any constant $\epsilon > 0$. 
Let $I$ denote the solution constructed by the algorithm; initialize $I$ to $\emptyset$.
Pick an arbitrary vertex $v_1$ and find a smallest radius $r_1$ such that 
$$\alpha(G[B_{r_1+1}(v_1)]) \le  (1 + \epsilon) \cdot \alpha(G[B_{r_1}(v_1)]).$$ 
Add a maximum-sized independent set of $G[B_{r_1}(v_1)]$ to $I$ and delete $G[B_{r_1+1}(v_1)]$ from the graph.
This completes the first iteration of the algorithm.
For the second iteration, find an arbitrary vertex $v_2$ in the graph that remains and repeat an iteration of ``ball growing'' until a radius $r_2$ is found. Continue these iterations until the graph becomes empty. Note that each radius $r_i = O(\log n/\log (1+\epsilon))$ and the constructed solution $I$ is an independent set of $G$. Furthermore, a simple argument (see \cite{GhaffariKMSTOC2017}) shows that $|I| \ge 1/(1+\epsilon) \cdot \alpha(G)$. 
Note that this is not a polynomial-time algorithm because each $v_i$ needs to compute the exact maximum-sized independent set in $B_r(v_i)$ for different values of $r$.

\begin{lemma}
There is a deterministic $1/(1+\epsilon)$-approximation algorithm for \mxis in the \ktzero\ \congest{} model that uses $\tilde{O}(n^2/\epsilon)$ messages and runs in $O(\poly(m + n))$ rounds.
\end{lemma}
\begin{proof}
We show that the sequential algorithm described above can be implemented in the \ktzero\ \congest{} model, using $\tilde{O}(n^2/\epsilon)$ messages and running in $O(\poly(m + n))$ rounds.

A rooted spanning tree $T$ of $G$ can be constructed using $O(m)$ messages and $O(n)$ rounds \cite{peleg00}. 
Let $T$ be rooted at node $root$. Initially, all nodes in $G$ are \textit{active}. At the start of iteration $i$, the node $root$ initiates a broadcast-echo on $T$ so as to find a node $v_i$ of lowest ID from among the active nodes in $G$. Then $root$ informs $v_i$ that it is the next initiator of the ``ball growing'' algorithm, responsible for initiating iteration $i$. All of this takes $O(n)$ messages because both during the broadcast and during the echo, exactly one message travels along each edge of $T$.
Further, this broadcast-echo procedure takes $O(n)$ rounds.
Since there can be at most $n$ iterations, the portion of the algorithm devoted to finding initiators and informing them, uses
$O(n^2)$ messages and runs in $O(\poly(m + n))$ rounds.

Once the initiator $v_i$ has been identified and informed, it initiates an iteration of ``ball growing'' algorithm. We show that this runs on $O(\poly(m+n))$ rounds and uses $\tilde{O}((n_i + m_i)/\epsilon)$ messages, where $n_i$ is the number of nodes in $G[B_{r_i+1}(v_i)]$ and
$m_i$ is the number of edges incident on nodes in $G[B_{r_i+1}(v_i)]$. 
This means that over all iterations the ``ball growing'' portion of the algorithm uses $\tilde{O}(n^2/\epsilon)$ messages and runs in $O(\poly(m+n))$ rounds.

For any integer $r \ge 0$, suppose that $v_i$ knows $G[B_{r+1}(v_i)]$.
Using (exponential-time) local computation $v_i$ checks if $\alpha(G[B_{r+1}(v_i)]) \le  (1 + \epsilon) \cdot \alpha(G[B_{r}(v_i)])$. If the condition is satisfied, then this iteration of the ``ball growing'' algorithm stops and $r_i$ is set to $r$. 
Furthermore, having computed a maximum-sized independent set $I'$ of $G[B_r(v_i)]$, node $v_i$ informs all nodes in $I'$ that they are part of the solution.
Otherwise, $v_i$ will grow the ball by gathering information so that it knows $G[B_{r+2}(v_i)]$. 
To do this, $v_i$ initiates a broadcast asking all nodes at distance $r+2$ for their incident edges. This broadcast uses $O(|B_{r+2}(v_i)|)$ messages and $O(r)$ rounds. Each node $u \in B_{r+2}(v) \setminus B_{r+1}(v)$ marks itself inactive and then sends its incident edges back to the initiator $v_i$. Since $r = O(\log n/\log (1 + \epsilon)) = O(\log n/\epsilon)$, each edge sent to $v_i$ travels a distance of $O(\log n/\epsilon)$, and thus for any constant $\epsilon > 0$, the message as well as round complexity of this step is $\tilde{O}(\ell/\epsilon)$, where $\ell$ is the number of edges incident on nodes at distance $r+2$ from $v_i$.
This completes the proof.
\end{proof}

While this result is described for \mxis, the authors of \cite{GhaffariKMSTOC2017} also show that this ``ball growing'' algorithm is able to produce a $(1 + \epsilon)$-approximation for MDS. Furthermore, \cite{GhaffariKMSTOC2017} claims that this approach produces a $(1+\epsilon)$
or a $1/(1+\epsilon)$-approximation for any problem that can be expressed as certain type of packing or covering integer linear program. In addition to \mxis and MDS, this framework also includes \mxm and MVC. So we get the following theorem.\footnote{In this theorem statement we use the more convenient $(1-\epsilon)$ rather than $1/(1+\epsilon)$. This is justified by the fact that $1/(1+\epsilon)$ can be written as $1-\epsilon'$, where $\epsilon/2 \le \epsilon' \le \epsilon$.}

\begin{theorem}
There is a deterministic $(1-\epsilon)$-approximation algorithm for \mxis and \mxm in the \ktzero\ \congest{} model that uses $\tilde{O}(n^2/\epsilon)$ messages and runs in $O(\poly(m + n))$ rounds. Similarly, there is a a deterministic $(1+\epsilon)$-approximation algorithm for MDS and MVC in the \ktzero\ \congest{} model that uses $\tilde{O}(n^2/\epsilon)$ messages and runs in $O(\poly(m + n))$ rounds.
\end{theorem}

\subsection{Unconditional Quadratic Lower Bound for \mxm{} Approximation}
\label{section:KTZeroMaxMLB}

\label{sec:mm_lb}

Next, we show that $\tilde{\Omega}(n^2)$ messages are required to compute a constant-factor approximation for \mxm{}.

\begin{theorem} \label{thm:mm_lb}
  Consider any randomized algorithm for the \mxm{} problem in the \KTZero{} \local{} model. Assume that every matched edge is output by at least one of its endpoints; either by outputting a port number or the ID of the corresponding neighbor.
  If, for some $\epsilon \in \lb(\frac{1}{n^{1/3}},1\rb)$, the algorithm sends at most $\frac{\epsilon^3 n^2}{7^3 \cdot 8} = O\lb(\epsilon^3n^2\rb)$ messages with probability at least $1 - \frac{\epsilon}{7}$, then there exists a graph on $2n$ nodes such that the approximation ratio is at most $\epsilon$ in expectation.
\end{theorem}

In the remainder of this section, we give a class of graphs on which finding a large matching is hard, and then we state the details of the assumed port numbering model and discuss the output specification of a given matching algorithm in this setting. We make use of these definitions when proving Theorem~\ref{thm:mm_lb}.

\subsubsection*{The Lower Bound Graph} \label{sec:mm_lb_graph}

Let $\gamma = \frac{\epsilon}{7}$.
We consider the following $2n$-node graph $G$ consisting of vertex sets $A$ and $B$, where $A=\{u_1,\dots,u_n\}$ and $B=\{v_1,\dots,v_n\}$.
We further partition $A$ into $C_A$ and $N_A$ such that $N_A=\{u_1,\dots,u_{n-\lfloor \gamma n\rfloor}\}$ and $C_A = A \setminus N_A$.
Each node in $N_A$ is connected to all nodes in $C_A$, whereas the nodes in $C_A$ form a clique.
Analogously, we define $C_B$ and $N_B$, and the edges between them.
In addition, we add the set of \emph{valuable edges} $\{\set{u_1,v_1},\dots,\set{u_n,v_n}\}$.
Figure~\ref{fig:mm_lb} depicts an example of this construction.
Notice that the set of valuable edges corresponds to a perfect matching of size $n$ and hence forms an optimum solution for the maximum matching problem.

\begin{figure}[t]
  \centering
  \includegraphics[width=0.8\textwidth]{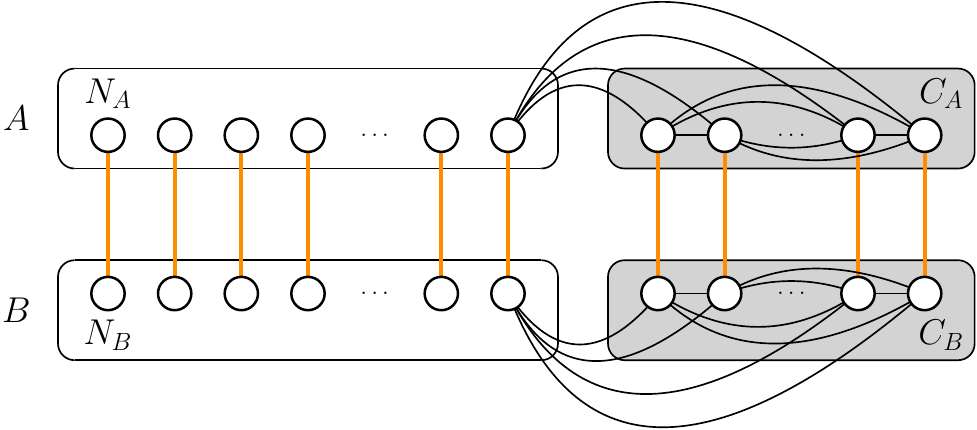}
  \caption{\small
  The lower bound construction for proving Theorem~\ref{thm:mm_lb}.
  There are $2n$ nodes in total equally partitioned into sets $A$ and $B$. Each one of the grey-shaded areas contains $\lfloor \gamma n \rfloor$ nodes that form the cliques $C_A$ and $C_B$, respectively.
  Every node in $N_A = A \setminus C_A$ has an edge to all nodes in $C_A$, and the nodes in $N_B$ and $C_B$ are connected similarly.
  The thick orange edges are the valuable edges that form a perfect matching.
  }
  \label{fig:mm_lb}
\end{figure}

\subsubsection*{The Port Numbering Model}
We consider the standard port numbering model, where the incident edges of a node $u$ are numbered $1,\dots,\deg(u)$.
This means that, in order to send a message across the edge $\set{u,v}$, node $u$ would need to send over some port $p_{(u,v)}$, whereas node $v$ would need to use port $p_{(v,u)}$, for some (possibly distinct) integers $p_{(u,v)} \in [\deg(u)]$\footnote{We use the standard notation $[m] := \set{1,\dots,m}$.} and $p_{(v,u)} \in [\deg(v)]$.

We say that the port $p_{(u,v)}$ is \emph{used} if $u$ sends a message over $p$ or receives a message that was sent on $p_{(v,u)}$; otherwise we say that it is \emph{unused}.
A crucial property of the $\ktzero$ assumption is that, initially, a node $u$ does not know that it is connected to $v$ via $p_{(u,v)}$. However, we assume that $u$ learns that $p_{(u,v)}$ connects to $v$ upon receiving a message directly from $v$.
To obtain a concrete lower bound graph, we fix the node IDs to correspond to a uniformly random permutation of $[2n]$, and, for each node $u$, we choose an assignment of its incident ports to the corresponding endpoints by  independently and uniformly selecting a random permutation of the set $[\deg(u)]$.

There are two standard ways how an algorithm may output a matching in $\ktzero$. The first possibility is that at least one of the endpoints of each matched edge outputs the corresponding port number. The second one is that a node $u$ outputs the ID of some neighbor $v$ to indicate that $\set{u,v}$ is in the matching.
We point out that our lower bound result holds under either output assumption.

\subsubsection*{Proof of Theorem~\ref{thm:mm_lb}} \label{sec:thm_mm_lb}
Consider any algorithm that satisfies the premise of the theorem.
Suppose that it sends at most $\frac{\gamma^3n^2}{8}$ messages with probability at least $1 - \gamma$, and let $\texttt{Sparse}$ denote the event that this happens.

\begin{lemma} \label{lem:mm_lb_existsS}
Let $J \subseteq [n - \lfloor \gamma n\rfloor]$ be the set of indices such that, for all $i \in J$, the edge $\set{u_i,v_i}$ is not part of the computed matching.
  Then $\mathbf{E}\lb[|J|\ \middle|\ \text{\rm\texttt{Sparse}}\rb]
  \ge \lb( 1 - 5\gamma \rb)n.
  $
\end{lemma}
\begin{proof}
  Let $\alpha = \frac{\gamma^2}{\gamma+3}$.
  We will first show that there exists a large subset of indices $I \subseteq [n- \lfloor\gamma n\rfloor ]$ such that the following two properties hold, for all $i \in I$:
  \begin{compactenum}
    \item[(a)] $u_i$ and $v_i$ each use at most $\alpha n$ of their incident ports throughout the execution.
    \item[(b)] No message is sent over the valuable edge $\set{u_i,v_i}$.
  \end{compactenum}
  Assume towards a contradiction that there exists a subset $T \subseteq N_A$ such that  more than ${\alpha n}$ incident ports of every node in $T$ are used, and $|T| > \gamma|N_A|$.
Since the number of used ports is a lower bound on the message complexity, it follows that the total number of messages sent is strictly greater than
  \[
    {\alpha\gamma n |N_A|}
    \ge
    \alpha\gamma n (n - \gamma n)
=
    \lb(\alpha\gamma - \alpha\gamma^2\rb) n^2
    \ge
    \frac{\alpha\gamma n^2}{2}
    \ge
    \frac{\gamma^3 n^2}{8},
  \]
  where we have used the fact that $\alpha \ge \tfrac{\gamma^2}{4}$ in the last inequality.
  This, however, contradicts the conditioning on event \texttt{Sparse}, which tells us that there must exist a set $S_A' \subseteq N_A$  of size $(1 - \gamma)|N_A|$ such that, for every $u_i \in S_A'$, the algorithm uses at most $\alpha n$ incident ports.
  By a symmetric argument, we can also show the existence of a set of $S_B' \subseteq B$ with the same properties and size as $S_A'$.

  Let $I'$ be the set of indices $i$ such that $u_i \in S_A'$ and $v_i \in S_B'$, which means that every $i \in I'$ satisfies property (a). 
  Note that
	\begin{align}
	|I'| \ge (1 - 2\gamma)|N_A|. \label{eq:mm_lb_I}
  \end{align}

	We will now identify a large subset of the indices in $I'$ that also satisfy (b), which form the sought subset $I$.
  Consider any $i \in I'$ and the corresponding pair $u_i$ and $v_i$. 
  Let $X_i$ be the indicator random variable that is $1$ if and only if neither $u_i$ nor $v_i$ sends a message over the valuable edge $\set{u_i,v_i}$.
  Node $u_i$ has degree $\lfloor\gamma n\rfloor+1 \ge \gamma n$ and due to condition (a), a set $P$ of at least $(\gamma - \alpha)n$ of its incident ports remain unused  throughout the execution.
  Since the port connections are chosen independently and uniformly at random for each node, it follows that $u_i$'s valuable edge has uniform probability to be connected to any one of $u_i$'s ports, as long as $u_i$ has not yet sent or received a message over its valuable edge.
To see why $u_i$ cannot learn any information about the port of $\set{u_i,v_i}$ from other nodes, observe that any message $\mu$ sent by some node $w$ is a function of the local states of the nodes, which may include all currently-known port assignments, and, conditioned on the fact that all ports in $P$ are unused, it follows that $\mu$ is independent of the distribution of the endpoint connections of $P$.

  By assumption, $u_i$ and $v_i$ send messages over at most $\alpha n$ of their respective incident ports. 
  Hence, $\mathbf{E}\lb[ X_i \!=\! 1\ \middle|\ \texttt{Sparse} \rb] \ge 1 - \frac{2\alpha}{\gamma - \alpha}$.
  It follows that 
 \begin{align}
  \mathbf{E}\lb[ |I|\ \middle|\ \texttt{Sparse} \rb]
	\ge
  \mathbf{E}\lb[ \sum_{i \in I'} X_i\ \middle|\ \texttt{Sparse}\rb] 
  &\ge \lb(1 -\frac{2\alpha}{\gamma - \alpha}\rb)\mathbf{E}\lb[ |I'| \ \middle|\ \texttt{Sparse} \rb] \notag\\
      &\ge \lb(1 - \frac{2\alpha}{\gamma - \alpha}\rb)(1 - 2\gamma)(1 - \gamma)n \tag{\small by \eqref{eq:mm_lb_I}}\\ 
      &\ge \lb(1 - 3\gamma - \frac{2\alpha}{\gamma - \alpha}\rb)n \label{eq:mm_lb_size}
 \end{align}
where $I$ has the property that, for all $i \in I$, nodes $u_i$ and $v_i$ satisfy (a) and (b).

To complete the proof, we need to show that there is only a small probability of $u_i$ or $v_i$ outputting the valuable edge $\set{u_i,v_i}$, for any index $i \in I$.
Let $Y_i$ be the corresponding indicator random variable that is $1$ iff $\set{u_i,v_i}$ is part of the output.
We first consider algorithms where, for each matched edge, at least one of its endpoint nodes outputs the corresponding port number: 
  As argued above, $u_i$ and $v_i$ have at least $(\gamma-\alpha)n$ unused ports respectively, of which they do not know the endpoints, and their respective valuable edge is uniformly distributed among these ports.
  It follows that $u_i$ outputs its valuable edge with probability at most $\frac{1}{(\gamma-\alpha)n}$ and the same holds for $v_i$.
  Thus, $\Pr\lb[ Y_i \!=\! 1 \rb] \le \frac{2}{(\gamma-\alpha)n}$.
Now consider an algorithm where a node outputs the ID of its matching partner, if any.
Consider any $u_i$ ($i \in I$).
Recall that $I$ satisfies (a) and (b) and the assumption that the IDs were assigned uniformly at random, which implies that the valuable edges between $A$ and $B$ correspond to a random matching between the corresponding ID pairs. 
Thus, at any point in the execution, from the point of view of a node $u_i$ ($i \in I$), the probability distribution of the ID of $v_i$ is still uniform over the IDs of nodes $\set{v_j \in B \mid j \in I }$. 
Consequently, \eqref{eq:mm_lb_size} tells us that $u_i$ has probability at most $\frac{1}{(1 - 3\gamma - \frac{2\alpha}{\gamma-\alpha})n}$ to output the correct ID of $v_i$, and a similar argument applies to $v_i$. 

Combining both possible case regarding the output of the algorithm, we obtain that
\begin{align}
 \Pr\lb[ Y_i \!=\! 1 \rb] \le \max\lb\{ \frac{2}{(\gamma-\alpha)n}, \frac{2}{(1 - 3\gamma - \frac{2\alpha}{\gamma-\alpha})n}\rb\}
 \le \frac{2}{n}\max\lb\{ \frac{1}{\gamma-\alpha}, \frac{1}{1 - 4\gamma}\rb\}, \notag \end{align}
where the second inequality follows because $\alpha = \frac{\gamma^2}{\gamma+3}$ ensures $\frac{2\alpha}{\gamma-\alpha}\le\gamma$.

  Let $J \subseteq I$ be the set of indices such that neither $u_i$ nor $v_i$ outputs the valuable edge $\set{u_i,v_i}$.
  Recalling \eqref{eq:mm_lb_size}, we get
  \begin{align}
    \mathbf{E}\lb[ |J|\ \middle|\ \texttt{Sparse} \rb]
      &\ge
      \lb( 1 - 3\gamma - \frac{2\alpha}{\gamma-\alpha} \rb)
      \lb( 1 -  \Pr\lb[ Y_i \!=\! 1 \rb]\rb)n \notag\\
      &\ge
      \lb( 1 - 3\gamma - \frac{2\alpha}{\gamma-\alpha} \rb)
      \lb( 1 -  \frac{2}{n}\max\lb\{ \frac{1}{\gamma-\alpha}, \frac{1}{1 - 4\gamma}\rb\}\rb)n \notag\\
      &\ge
      \lb( 1 - 3\gamma - \frac{2\alpha}{\gamma-\alpha} -  \frac{2}{n}\max\lb\{ \frac{1}{\gamma-\alpha}, \frac{1}{1 - 4\gamma}\rb\}\rb)n \notag\\
      &\ge
      \lb( 1 - 4\gamma - \frac{2\alpha}{\gamma-\alpha} \rb)n 
      \tag{\small since $\gamma \ge \frac{2}{n}\max\lb\{ \frac{1}{\gamma-\alpha}, \frac{1}{1 - 4\gamma}\rb\}$}\\
      &\ge
      \lb( 1 - 5\gamma\rb)n,\notag
  \end{align}
  where, in the last step, we have used the fact that the assumption that $\alpha = \frac{\gamma^2}{\gamma+3}$ implies $\frac{2\alpha}{\gamma-\alpha}\le\gamma$.
\end{proof}

We are now ready to complete the proof of Theorem~\ref{thm:mm_lb}.
Let $M$ be the matching computed by the algorithm and recall that the perfect matching has size $n$. 
We have
\begin{align*}
  \mathbf{E}\lb[ |M| \rb]
  &\le
  \mathbf{E}\lb[ |M|\ \middle|\ \texttt{Sparse}\rb]
  +
  n\cdot\Pr[ \neg\texttt{Sparse}] \\
  &\le
  |C_A| + |N_A| - |J| + \gamma n\\
  &\le
  2\gamma n + \lb(n - \gamma n +1\rb) - \lb( 1 - 5\gamma \rb)n \quad\quad\text{\small (by Lemma~\ref{lem:mm_lb_existsS})}\\
  &\le  7\gamma n,
\end{align*}
which implies the claimed upper bound on the approximation ratio since $\gamma = \frac{\epsilon}{7}$ and the maximum matching has size $n$.

\subsection{Unconditional Quadratic Lower Bound for MDS Approximation}
\label{section:KTZeroMDSLB}

A well known and powerful tool for proving lower bounds in \KTZero{} is the notion of a port preserving crossing which is used to prove message complexity lower bounds in~\cite{AwerbuchGPV90}. We use this tool to prove lower bounds for MDS approximation in this section, and also for MVC and \mxis{} approximation in Sections~\ref{section:kt0-mvc-lb}
and~\ref{section:kt0-maxIS-lb}.

\begin{definition}[Port-Preserving Crossing]\label{def:kt-0-port-preserving-crossing}
    Let $H$ be an arbitrary graph with two edges $e = \{u,v\}$ and $e' = \{u',v'\}$ for distinct nodes $u,u',v,v'$. Let $e$ be connected to $u$ at port $p$ and to $v$ at port $q$, and similarly let $e'$ be connected to $u'$ at port $p'$ and to $v'$ at port $q'$. The port preserving crossing of $e$ and $e'$ is the graph $H_{e,e'}$ which is obtained by removing the edges $e$, $e'$ from $H$ and adding the edge $\{u, u'\}$ connected to $u$ at port $p$ and to $u'$ at port $p'$ and the edge $\{v, v'\}$ connected to $v$ at port $q$ and to $v'$ at port $q'$. See Figure~\ref{fig:crossing} for an illustration of this definition
\end{definition}

\begin{figure}[th]
\begin{center}
\includegraphics[width=0.8\textwidth]{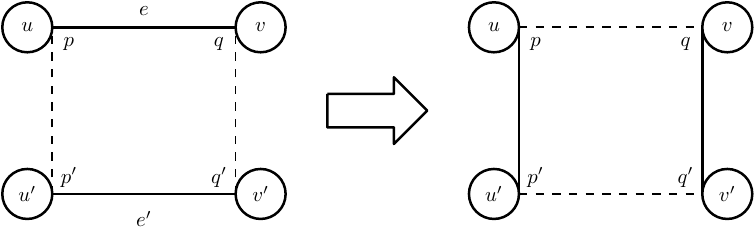}
\end{center}
\caption{Illustration of port preserving crossing in Definition~\ref{def:kt-0-port-preserving-crossing}}
\label{fig:crossing}
\end{figure}

\begin{theorem}\label{thm:kt-0-crossing-similarity}
Let $\mathcal{A}$ be a deterministic \KTZero{} \local{} algorithm. Let $H$ and $H_{e, e'}$ be the two graphs described in Definition~\ref{def:kt-0-port-preserving-crossing} with the same ID assignment. If no messages pass over $e$ and $e'$, $\mathcal{A}$ behaves identically on the graphs $H$ and $H_{e, e'}$.
\end{theorem}
\begin{proof}
    We will prove by induction that before each round of algorithm $\mathcal{A}$, the state of each node $x$ in $H$ is identical to the state of the corresponding node $x$ in $H_{e,e'}$. This is true trivially before round $1$ as the state of each node is just its ID and the list of ports. Inductively assume that the claim is true before some round $i \ge 1$. In round $i$, every node $x$ will send the same message through each port in both graphs $H$ and $H_{e,e'}$. Since no message passes through $e$ and $e'$, the nodes $u,v,u',v'$ will not send any message over their corresponding ports $p,q,p',q'$. Since every other edge is identically connected in both graphs, we have that in round $i$, each node $x$ will also receive the same messages through each port in both graphs $H$ and $H_{e,e'}$. So before round $i+1$, the state of each node $x$ will be identical on both graphs $H$ and $H_{e, e'}$.
\end{proof}

For the approximate MDS lower bound, we define a family $\{G_{x, y} \mid x \in \{0, 1\}^{n^2}, y \in \{0, 1\}^{n^2}\}$ of lower bound graphs. 
This construction is inspired by the construction in \cite{BachrachCDELP19}, but with a critical difference, that we highlight below.
For positive integer $n$, let $x, y \in \{0, 1\}^{n^2}$. We will now define a graph $G_{x, y}$ as follows. The vertex set of $G_{x, y}$ is 
$$A_1 \cup A_2 \cup B_1 \cup B_2 \cup C_1 \cup C_2 \cup \{a^*, b^*\}$$
where $A_1 = \{a_1^i \mid 1 \le i \le n\}$, $A_2 = \{a_2^i \mid 1 \le i \le n\}$,
$B_1 = \{b_1^i \mid 1 \le i \le n\}$, $B_2 = \{b_2^i \mid 1 \le i \le n\}$,
$C_1 = \{c_1^i \mid 1 \le i \le n\}$, and $C_2 = \{c_2^i \mid 1 \le i \le n\}$. Therefore, $G_{x,y}$ has $6n+2$ vertices.

We now describe the edges of $G_{x, y}$. The vertices in $C_1$ (and $C_2$) form an $n$-vertex clique.
Each vertex $c_1^i \in C_1$ is connected to all vertices $a_1^j$, $j \not= i$ and to all vertices $b_1^j$, $j \not= i$.
Similarly, each vertex $c_2^i \in C_2$ is connected to all vertices $a_2^j$, $j \not= i$ and to all vertices $b_2^j$, $j \not= i$.
Vertex $a^*$ is connected to all vertices in $A_1$ and vertex $b^*$ is connected to all vertices in $B_1$.
This completes the ``fixed'' edges in the graph, i.e., the edges that do not depend on bit vectors $x$ and $y$.
The edges between $A_1$ and $A_2$ depend on $x$ as follows: $\{a_1^i, a_2^j\}$ is an edge iff $x_{ij} = 1$.
The edges between $B_1$ and $B_2$ depend on $y$ as follows: $\{b_1^i, b_2^j\}$ is an edge iff $y_{ij} = 1$.
The construction is illustrated in Figure~\ref{fig:MDSLBConstruction}.

In \cite{BachrachCDELP19}, the authors use a similar construction to obtain an $\Omega(n^2)$ \textit{round} lower bound for \textit{exact} MDS. Their construction
critically depends on the existence of a small cut (with $O(\poly(\log n))$ edges) across the partition $(V_A, V_B)$ of the vertex set, where $V_A \supseteq A_1 \cup A_2$ and $V_B \supseteq B_1 \cup B_2$. For a message complexity lower bound, we don't need a small cut. In fact, it can be verified that $\Omega(n^2)$ edges connect any $V_A \supseteq A_1 \cup A_2$ and $V_B \supseteq B_1 \cup B_2$ in our lower bound graph. This flexibility plays a key role in our ability to force a constant-sized dominating set in the lower bound graph and obtain a relatively large gap in the MDS sizes between different types of instances, as shown in the following lemma.

\begin{figure}
\begin{center}
\includegraphics[width=\textwidth]{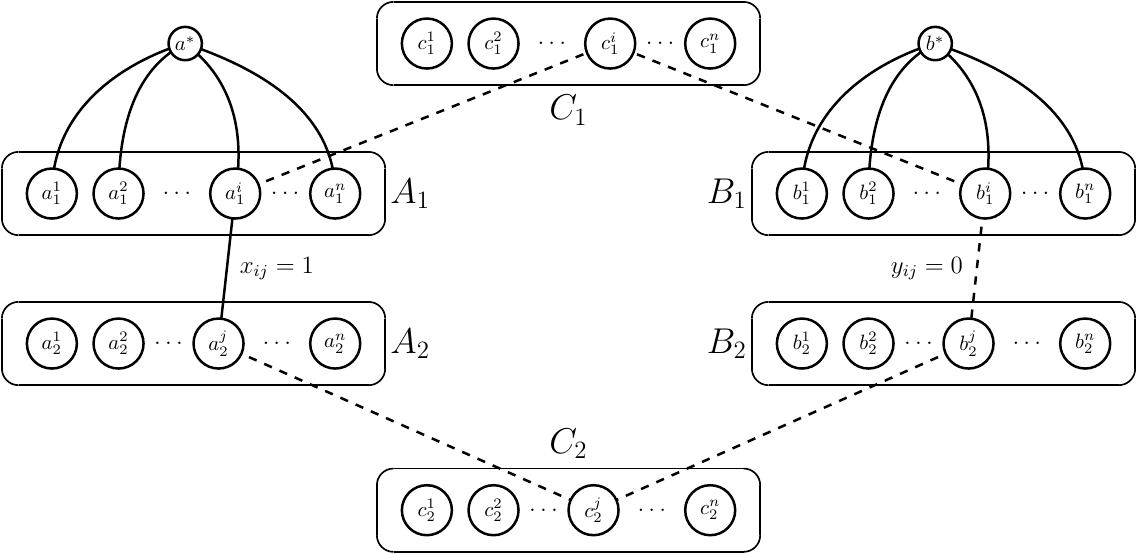}
\end{center}
\caption{Illustration of the graph $G_{x,y}$ used in the MDS lower bound proof. For clarity, many edges are not shown. Dashed line segments represent the absence of edges.}
\label{fig:MDSLBConstruction}
\end{figure}

\begin{lemma}
\label{thm:mdscrossing}
If $\exists (i,j)$ such that $x_{ij} = y_{ij} = 1$ then $G_{x, y}$ has a dominating set of size at most $4$.
If $\nexists (i,j)$ such that $x_{ij} = y_{ij} = 1$ then $G_{x, y}$ has a dominating set of size at least $5$.
\end{lemma}
\begin{proof}
If $\exists (i,j)$ such that $x_{ij} = y_{ij} = 1$, then the edges $\{a_1^i,a_2^j\}$ and $\{b_1^i,b_2^j\}$ exist in $G_{x,y}$ and so the nodes $a_1^i,c_1^i,b_1^i,c_2^j$ form a dominating set of size $4$. 

We aim to show that on the other hand, when $\nexists (i,j)$ such that $x_{ij} = y_{ij} = 1$, the minimum dominating set size is at least $5$. More precisely, we show that no set $S$ of 4 nodes can dominate the lower bound graph $G$.

First, it is clear that there must be at least one node of $S$ in $A_1 \cup \{a^*\}$ to dominate $a^*$. Similarly, there must be at least one node of $S$ in $B_1 \cup \{b^*\}$ to dominate $b^*$. Moreover, if $a^* \in S$ or $b^* \in S$ (and $|S| = 4$) then in addition to the first two nodes, at least 3 more nodes are required to dominate $C_1 \cup B_2 \cup C_2 \cup A_2$. More precisely, one more node is required to dominate $C_1$ entirely and two more nodes are required to dominate $B_2 \cup C_2 \cup A_2$ entirely. 
Hence, we assume $a_1^i \in S$ and $b_1^j \in S$ for some integers $i,j \in \{1,\ldots,n\}$ for the remainder of the proof. 

First, assume $i = j$. Then, a third node in $A_1 \cup C_1 \cup B_1$ is required to dominate $C_1$. As a result, the last node must be in $C_2$ (as any node in $A_2$ or $B_2$ cannot dominate $C_2$ entirely). Let that node be $c_2^p$ for some integer $p \in \{1,\ldots,n\}$. This implies that $S$ is dominating only if both edges $(a_1^i,a_2^p)$ and $(b_1^j,b_2^p)$ are present in $G$ (for $i = j$), which is a contradiction as $x_{ik} = y_{ik} = 1$. 

Second, assume $i \neq j$. Then, no single node of $C_1$ (and thus in fact no single node of $G$) can dominate $A_1' \triangleq A_1 \setminus \{a_1^i\}$ and $B_1' \triangleq B_1 \setminus \{b_1^i\}$. Next, we show the following claim: there must be one node of $S$ in each of the sets $A_1$, $A_2$, $B_1$ and $B_2$. Suppose first that at least one node of $S$ is in $C_1$. Then without loss of generality, one node of $A_1$ is not dominated by the first three nodes of $S$. Thus, the fourth node must be in $A_2$. However, there is a node in $C_2$ that is not dominated by $S$, which is a contradiction. Next, suppose that at least one node of $S$ is in $C_2$ (and thus dominates no node in $A_1$ nor $B_1$). Then, the fourth node cannot dominate both $A_1'$ and $B_1'$, which means $S$ cannot dominate $G$. Finally, the claim follows from the fact that either $A_1'$ or $B_1'$ is not dominated if there is no node of $S$ either in $A_2$ or in $B_2$.

To finish with the case $i \neq j$, note that given the above claim, we can define the nodes of $S$ as $a_1^i$, $b_1^j$, $a_2^p$ and $b_2^q$. Moreover, it must hold that $i \neq j$ and $p \neq q$, otherwise either $C_1$ or $C_2$ is not dominated. In addition, $a_1^i$ must dominate $A_2 \setminus a_2^p$ and $b_2^q$ must dominate $B_1 \setminus b_1^j$.
Since $i \neq j$ and $p \neq q$, this implies that both edges $(a_1^i, a_2^q)$ and edges $(b_1^i,b_2^q)$ must be present in $G$, which is a contradiction as $x_{iq} = y_{iq} = 1$.
\end{proof}

We will pick a member $G$ of this family such that the subgraph $G[A_1 \cup A_2]$ is a fixed $n/2$-regular bipartite graph where for all $1 \le i \le n$, $a_1^i$ is connected to $a_2^j$ for all $j = i, i+1, \dots, (i + n/2-1)$ (if $j$ becomes larger than $n$, we wrap around back to $1$). And the subgraph $G[B_1 \cup B_2]$ is the complement of $G[A_1 \cup A_2]$. We claim that that this graph has a dominating set of size at least 5. This is because $G$ can also be viewed as $G_{x,y}$ for some string $x \in \{0, 1\}^{n^2}$ and $y = \overline{x}$. Since there is no index $(i,j)$ such that $x_{ij} = y_{ij} = 1$, by Lemma~\ref{thm:mdscrossing}, $G$ has a dominating set of size at least 5.

Let $\mathcal{A}$ be a deterministic dominating set algorithm in the \KTZero{} \local{} model that uses $o(n^2)$ messages. Note that $\mathcal{A}$ outputs a dominating set of size at least 5 on $G$. Then, there exists an edge $e = \{a_1^i, a_2^j\}$ in $G$, $a_1^i \in A_1$, $a_2^j \in A_2$ such that no message passes through this edge during the execution of algorithm $\mathcal{A}$.
Let $\overline{N}(a_1^i) \subseteq A_2$ be the subset of $A_2$ containing vertices that are not neighbors of $a_1^i$. Similarly, let $\overline{N}(a_2^j) \subseteq A_1$ be the subset of $A_1$ containing vertices that are not neighbors of $a_2^j$. Note that $|\overline{N}(a_1^i)| = |\overline{N}(a_2^j)| = \frac{n}{2}$, and there are $\Theta(n^2)$ edges in $G$ between $\overline{N}(a_1^i)$ and $\overline{N}(a_2^j)$. Since $\mathcal{A}$ uses $o(n^2)$ messages, there is an edge $e' = \{a_1^p, a_2^q\}$, $a_1^p \in \overline{N}(a_2^j)$, $a_2^q \in \overline{N}(a_1^i)$, such that no message passes over edge $e'$ during the execution of $\mathcal{A}$.  

Let $G_{e, e'}$ be the graph obtained from $G$ by crossing the edges $e = \{a_1^i, a_2^j\}$ and $e' = \{a_2^q, a_1^p\}$ according to Definition~\ref{def:kt-0-port-preserving-crossing}. In particular, $e$ and $e'$ are replaced by edges $\{a_1^i, a_2^q\}$ and $\{a_1^p, a_2^j\}$, with the port-numbering preserved. Note that by Theorem~\ref{thm:kt-0-crossing-similarity}, algorithm $\mathcal{A}$ behaves identically in $G$ and $G_{e, e'}$ and so $\mathcal{A}$ outputs a dominating set of size at least 5 for $G_{e,e'}$ also.
However, $G_{e, e'}$ has a dominating set of size 4. This follows from Lemma~\ref{thm:mdscrossing}, and the fact that $\{b_1^i, b_2^q\}$ (and $\{b_1^p, b_2^j\}$) is an edge in $G_{e, e'}$. 

This means that $\mathcal{A}$ outputs a dominating set of size at least 5 for a graph with a dominating set of size at most 4. Thus, a $(5/4-\epsilon)$-approximation, $\epsilon > 0$, for MDS requires $\Omega(n^2)$ messages.

\begin{theorem}
\label{thm:kt0-mds-det-lb}
For any constant $\epsilon > 0$, any deterministic \KTZero{} \local{} algorithm $\mathcal{A}$ that computes
a $(5/4-\epsilon)$-approximation of MDS on $n$-vertex graphs has $\Omega(n^2)$ message complexity.
\end{theorem}

\begin{theorem}
\label{thm:kt0-mds-random-lb}
For any constant $\epsilon > 0$, any randomized Monte-Carlo \KTZero{} \local{} algorithm $\mathcal{A}$ that computes a $(5/4-\epsilon)$-approximation of MDS on $n$-vertex graphs with constant error probability $\delta < 1/2$ has $\Omega(n^2)$ message complexity.
\end{theorem}
\begin{proof}
By Yao's minimax theorem~\cite{Yao77_probab}, it suffices to show a lower bound on the message complexity of a deterministic algorithm under a hard distribution $\mu$.

The hard distribution $\mu$ consists of the graph $G$ with probability $1/2$ and the remaining probability is distributed uniformly over the set $\mathcal{F}$ which contains all graphs of the form $G_{e, e'}$ where $e = \{a_1^i, a_2^j\}, e'=\{a_1^p, a_2^q\}$ for $1 \le i \le n, i < p, j \neq q$ such that $e \in G, e' \in G, \{a_1^i, a_2^q\} \notin G,$ and $\{a_1^p, a_2^j\} \notin G$. 

Note that a deterministic algorithm cannot produce an incorrect output on $G$ as it will have error probability greater than $\delta$. Therefore, the deterministic algorithm outputs a dominating set of size at least 5 on $G$.
So we now need to show that any deterministic algorithm that produces an incorrect output on at most a $2\delta$ fraction of graphs in $\mathcal{F}$ has $\Omega(n^2)$ message complexity. So suppose for sake of contradiction that there exists such an algorithm that outputs a $(5/4 - \epsilon)$-approximate dominating set with message complexity $\alpha n^2$ for some small constant $\alpha$. 

In $G$, the number of edges between $A_1$ and $A_2$ is $n^2/2$. Therefore, the algorithm can only send messages on $2\alpha$ fraction of the edges between $A_1$ and $A_2$. Recall our previous argument that for each edge $e$ between $A_1$ and $A_2$, there are $\Theta(n^2)$ edges $e'$ such that $G_{e,e'} \in \mathcal{F}$. Therefore, for each edge $e$ between $A_1$ and $A_2$, there are $1 - O(\alpha)$ fraction of edges $e'$ such that the algorithm does not send a message over $e'$, and $G_{e,e'} \in \mathcal{F}$. We consider all pairs of edges $e,e'$ such that $G_{e,e'} \in \mathcal{F}$ and the algorithm does not send any message over $e, e'$. Hence, it is easy to see that the number of such graphs is at least $(1 - O(\alpha))|\mathcal{F}|$ where the $O(\alpha)$ here hides a larger constant.

The proof of Theorem~\ref{thm:kt0-mds-det-lb} says that the deterministic algorithm outputs the same solution as it would for $G$ on all these graphs $G_{e,e'}$, which is a dominating set of size $5$. But the minimum dominating set in all these graphs has size at most $4$, which means the output of the algorithm is not a $(5/4 - \epsilon)$-approximation.

The algorithm is therefore incorrect on at least $(1 - O(\alpha))$ fraction of graphs in $\mathcal{F}$. This is a contradiction if $2\delta < 1 - O(\alpha)$, that is, $\delta < 1/2 - O(\alpha)$. This can be made arbitrarily close to $1/2$ by choosing a small enough constant $\alpha$.
\end{proof}

\subsection{Unconditional Quadratic Lower Bound for MVC Approximation}
\label{section:kt0-mvc-lb}
Here we show that for any parameter $c \ge 1$, any randomized \KTZero{} \local{} $c$-approximation algorithm for the minimum vertex cover (MVC) problem uses $\Omega(n^2/c)$ messages for some $n$-vertex graph.

Define a graph $G = (V, E)$ where $V$ is divided into three parts $X, Y, Z$ such that $|X| = |Z| = t$ and $|Y| = t/(4c)$. We add all possible edges between $X$ and $Y$ and all possible edges between $Y$ and $Z$. We then add a copy $G' = (V', E')$ of $G$ (where the three parts of $V'$ are $X', Y'$ and $Z'$). Note that $G$ and $G'$ have exactly $t^2/(2c)$ edges each. We will call $G \cup G'$ the \emph{base graph}, using which we create the lower bound graphs. Let $n = |V \cup V'| = 4t + t/(2c)$, thus $t = 2cn/(8c + 1)$.

Each node in the base graph is assigned a unique ID in the range $[1,n]$, and the ports of each node are also assigned in some arbitrary way.

We create a \emph{crossed graph} $G_{e, e'}$ by starting with the base graph $G \cup G'$ and then replacing edges $e = \{y,z\}$ and $e' = \{x',y'\}$ with edges $\{y, y'\}$ and $\{z, x'\}$ in a port preserving manner according to Definition~\ref{def:kt-0-port-preserving-crossing}. The base graph and crossed graph are illustrated in Figure~\ref{fig:MVCLBConstruction}. Note that the ID and port assignments to the nodes remain unchanged in the crossed graph.

The base and crossed graphs are similar to the lower bound graphs used in \cite{PaiPPRPODC2021} to prove message complexity lower bounds for MIS and $(\Delta+1)$-coloring (here $|Y| = |Y'| = t$). The motivation for shrinking the size of $Y$ (and $Y'$) is to ensure that in any approximate vertex cover of $G \cup G'$, there are lots of vertices that are not in the cover. This is made precise in the following claim.

\begin{figure}[t]
\begin{center}
\includegraphics[width=0.8\textwidth]{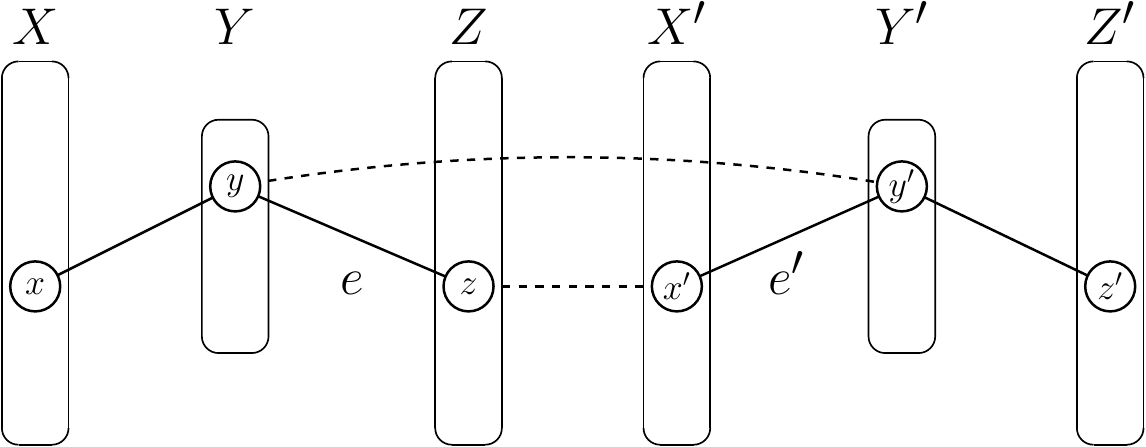}
\end{center}
\caption{Illustration of the graphs $G \cup G'$, and $G_{e,e'}$ used in the MVC lower bound proof.}
\label{fig:MVCLBConstruction}
\end{figure}

\begin{claim}\label{claim:kt0-mvc-size}
Any $c$-approximate vertex cover in $G \cup G'$ has size at most $t/2$.
\end{claim}
\begin{proof}
This follows from the fact that the optimal vertex cover in $G \cup G'$ is just $Y \cup Y'$ (and its size is $t/(2c)$). There can be no vertex cover of size $< t/(2c)$ because there are $t^2/c$ edges to cover and the maximum degree a vertex has in $G \cup G'$ is $2t$.
\end{proof}

Now suppose (to obtain a contradiction) that there is a deterministic \KTZero{} \local{} algorithm $\mathcal{A}$ that computes a $c$-approximate vertex cover $C$ in $G \cup G'$ using $o(t^2/c)$ messages.

\begin{claim}\label{claim:kt0-mvc-triplets}
There exist $y \in Y$, $z \in Z \setminus C$, $x' \in X' \setminus C$ and $y' \in Y'$ such that no message passes over edges $\{y, z\}$ and $\{x',y'\}$ in algorithm $\mathcal{A}$.
\end{claim}
\begin{proof}
  By Claim~\ref{claim:kt0-mvc-size}, there must be at least $t/2$ nodes in $Z$ and at least $t/2$ nodes in $X'$ that are not in $C$. And for at least one $z \in Z \setminus C$, there must be at least one $y \in Y$ such that no message passes over edge $\{y, z\}$. Otherwise, for every $y \in Y$ at least $t/2$ edges incident on $y$ have a message passing over the edge. This implies that $\Omega(t^2/c)$ edges in $G \cup G'$ have been used to send a message which is a contradiction. Similarly, for at least one $x' \in X' \setminus C$, there must be at least one $y' \in Y'$ such that no message passes over edges $\{x', y'\}$.
\end{proof}

\begin{lemma}\label{lemma:kt0-mvc-correctness}
Let $y \in Y$, $z \in Z \setminus C$, $x' \in X' \setminus C$ and $y' \in Y'$ be four nodes such that no message passes over the edges $e = \{y, z\}$ and $e' = \{x', y'\}$ in algorithm $\mathcal{A}$. Then $\mathcal{A}$ cannot compute a correct vertex cover on $G_{e,e'}$. 
\end{lemma}
\begin{proof}
Theorem~\ref{thm:kt-0-crossing-similarity} implies that both $z$ and $x'$ which are not in $C$ when $\mathcal{A}$ is run on $G \cup G'$ continue to not be in $C$ when $\mathcal{A}$ is run on $G_{e, e'}$. But, now there is an edge between $z$ and $x'$ and this is a violation of the correctness of algorithm $\mathcal{A}$. 
\end{proof}

Claim~\ref{claim:kt0-mvc-triplets} implies the existence of $e$ and $e'$ assumed in Lemma~\ref{lemma:kt0-mvc-correctness}. Thus $\mathcal{A}$ must use $\Omega(t^2/c)$ messages when run on $G \cup G'$. The following theorem formally states this lower bound.

\begin{theorem}
For any constant $c \ge 1$, any deterministic \KTZero{} \local{} algorithm $\mathcal{A}$ that computes a $c$-approximation of MVC on $n$-vertex graphs has $\Omega(n^2/c)$ message complexity.
\end{theorem}

We extend this deterministic lower bound to randomized Monte-Carlo algorithms.
\begin{theorem}
For any constant $c \ge 1$, any randomized Monte-Carlo \KTZero{} \local{} algorithm $\mathcal{A}$ that computes a $c$-approximation of MVC on $n$-vertex graphs with constant error probability $0 \le \delta < 1/12 - o(1)$ has $\Omega(n^2/c)$ message complexity.
\end{theorem}
\begin{proof}
By Yao's minimax theorem~\cite{Yao77_probab}, it suffices to show a lower bound on the message complexity of a deterministic algorithm under a hard distribution $\mu$.

The hard distribution $\mu$ consists of the base graph $G \cup G'$ with probability $1/2$ and the remaining probability is distributed uniformly over the set $\mathcal{F}$, where $\mathcal{F} = \{G_{e, e'} \mid e = \{y, z\}, e'=\{x', y'\}, y \in Y, z \in Z, x' \in X', y' \in Y'\}$. Note that $|\mathcal{F}| = t^4/(16c^2)$, since there are $t$ choices for each $z$ and $x$, and $t/(4c)$ choices for each $y$ and $y'$. Note that a deterministic algorithm cannot produce an incorrect output on $G \cup G'$ as it will have error probability greater than $\delta$. Therefore, we now need to show that any deterministic algorithm that produces an incorrect output on at most a $2\delta$ fraction of graphs in $\mathcal{F}$ has $\Omega(t^2/c)$ message complexity. So suppose for sake of contradiction that there exists such an algorithm that outputs a $c$-approximate vertex cover $C$ with message complexity $o(t^2/c)$

There can only be $o(t)$ vertices in $Y$, and $o(t)$ vertices in $Y'$ that send more than $\delta t$ messages. For each of the $(1 - o(1))t/(4c)$ nodes $y \in Y$ there are at least $(1 - \delta)t$ nodes $z \in Z$, and consequently at least $(1/2 - \delta)t$ nodes $z \in Z \setminus C$ (as $|C| \le t/2$), such that no message is sent over edge $e = \{y,z\}$. Similarly, for each of the $(1 - o(1))t/(4c)$ nodes $y' \in Y'$ there are at least $(1/2 - \delta)t$ nodes $x' \in X' \setminus C$, such that no message is sent over edge $e' = \{x',y'\}$. The proof of Lemma~\ref{lemma:kt0-mvc-correctness} says that the deterministic algorithm is incorrect on all these graphs $G_{e,e'}$. 

The algorithm is incorrect on at least $(1 - o(1))^2(1/2 - \delta)^2t^4/(16c^2)$ graphs in $\mathcal{F}$, that is, on at least a $(1/4 - \delta - o(1))$ fraction of graphs in $\mathcal{F}$. This is a contradiction if $2\delta < (1/4-\delta - o(1))$ or in other words, if $\delta < 1/12 - o(1)$. \end{proof}

\subsubsection{Unconditional Quadratic Lower Bound for Maximal Matching}
\label{section:kt-0-max-matching}
A maximal matching immediately gives us a $2$-approximation to MVC if we pick both end points of each matched edge into the cover. The maximality condition implies that all edges are covered, and the $2$-approximation guarantee is obtained by the observation that the optimal solution needs to pick at least one node to cover each edge of the matching. Since $2$-approximate MVC requires $\Omega(n^2)$ messages, we get the same lower bound for maximal matching.\footnote{This argument yields the same lower bound for exact \mxm{}, but it fails for approximation. This is because we obtain a contradiction to maximality, and an approximate matching need not be maximal.}

\begin{theorem}
Any randomized Monte-Carlo \KTZero{} \local{} algorithm $\mathcal{A}$ computing a maximal matching with constant error probability $0 \le \delta < 1/8 - o(1)$ where each matched edge is output by at least one of its end points has $\Omega(n^2)$ message complexity.
\end{theorem}

\subsection{Unconditional Quadratic Lower Bound for \mxis{} Approximation}
\label{section:kt0-maxIS-lb}
For the \mxis{} approximation lower bound we use the same base graph $G \cup G'$ as the MVC lower bound but we set $|Y| = |Y'| = \epsilon t$. The crossed graph $G_{e,e'}$ obtained by crossing the edges $e$ and $e'$ in a port preserving manner according to Definition~\ref{def:kt-0-port-preserving-crossing}. 

Suppose that there is a deterministic $(1/2+\epsilon)$-approximation algorithm $\mathcal{A}$ for \mxis{} that produces a solution (independent set) $S$ on input $G \cup G'$. In order to obtain a contradiction, we assume that $\mathcal{A}$ uses $o(n^2)$ messages. 

\begin{claim}
Any $(1/2 + \epsilon)$-approximate independent set $S$ in $G \cup G'$ is such that $|S \cap Z| \ge \epsilon \cdot t$ and $|S \cap X'| \ge \epsilon \cdot t$.
\end{claim}
\begin{proof}
It is easy to see that the largest independent set in $G \cup G'$ is $X \cup Z \cup X' \cup Z'$ of size $4t$.
We first claim that $S \subseteq X \cup Z \cup X' \cup Z'$ because, even if a single node of $Y$ (or $Y'$) is in $S$, it will not allow us to put any node of $X \cup Z$ (or $X' \cup Z'$) in $S$. In which case, the largest independent set we can create is $Y \cup X' \cup Z'$ (or $Y' \cup X \cup Z$) which has size $(2 + \epsilon)t$, which is only a $(1/2 + \epsilon/4)$-approximation.

If any two sets out of $X, Z, X', Z'$ contribute less than $\epsilon t$ nodes to $S$, then $|S| < (2 + 2\epsilon)t$, and $S$ is at most a $(1/2 + \epsilon/2)$-approximation. Therefore at least three sets must contribute at least $\epsilon t$ nodes to $S$. Now we can claim without loss of generality that $|S \cap Z| \ge \epsilon \cdot t$ and $|S \cap X'| \ge \epsilon \cdot t$ by swapping the sets $X, Z$ or $X', Z'$ as required.
\end{proof}

\begin{claim}\label{claim:kt0-maxIS-edges}
There exist $y \in Y$, $z \in S \cap Z$, $x' \in S \cap X'$, $y \in Y$ such that no message passes over edges $\{y, z\}$ and $\{x, y\}$ in algorithm $\mathcal{A}$.
\end{claim}
\begin{proof}
If this is not the case, then either for every $y \in Y$, all edges $\{y, z\}$ for $z \in S \cap Z$ have a message sent over them, or for every $y' \in Y'$ all edges $\{x', y'\}$ for $x' \in S \cap X'$ have a message sent over them. This implies that at least $\epsilon^2 \cdot t^2 = \Theta(n^2)$ edges have a message sent over them, contradicting the assumption that $\mathcal{A}$ has $o(n^2)$ message complexity.
\end{proof}

\begin{lemma}\label{lemma:kt0-maxIS-correctness}
Let $y \in Y$, $z \in S \cap Z$, $x' \in S \cap X'$, $y \in Y$ be four nodes such that no message passes over edges $e = \{y, z\}$ and $e' = \{x, y\}$ in algorithm $\mathcal{A}$. Then $\mathcal{A}$ cannot compute a correct independent set on $G_{e,e'}$. 
\end{lemma}
\begin{proof}
By Theorem~\ref{thm:kt-0-crossing-similarity}, algorithm $\mathcal{A}$ produces the same output on $G_{e, e'}$ and $G \cup G'$. This means that $z$ and $x'$ are both declared to be in the independent set in $G_{e,e'}$, but this is a violation to independence since they are neighbors in $G_{e,e'}$.
\end{proof}

Claim~\ref{claim:kt0-maxIS-edges} implies the existence of $e$ and $e'$ assumed in Lemma~\ref{lemma:kt0-maxIS-correctness}. Thus $\mathcal{A}$ must use $\Omega(t^2)$ messages when run on $G \cup G'$. The following theorem formally states this lower bound.

\begin{theorem}
For any constant $\epsilon > 0$, any deterministic \KTZero{} \local{} algorithm $\mathcal{A}$ that computes
a $(1/2+\epsilon)$-approximation of \mxis{} on $n$-vertex graphs has $\Omega(n^2)$ message complexity.
\end{theorem}

\begin{theorem}
For any constant $\epsilon > 0$, any randomized Monte-Carlo \KTZero{} \local{} algorithm $\mathcal{A}$ that computes a $(1/2 + \epsilon)$-approximation of \mxis{} on $n$-vertex graphs with constant error probability $\delta < \epsilon^2/8 - o(1)$ has $\Omega(n^2)$ message complexity.
\end{theorem}
\begin{proof}
By Yao's minimax theorem~\cite{Yao77_probab}, it suffices to show a lower bound on the message complexity of a deterministic algorithm under a hard distribution $\mu$.

The hard distribution $\mu$ consists of the base graph $G \cup G'$ with probability $1/2$ and the remaining probability is distributed uniformly over the set $\mathcal{F}$, where $\mathcal{F} = \{G_{e, e'} \mid e = \{y, z\}, e'=\{x', y'\}, y \in Y, z \in Z, x' \in X', y' \in Y'\}$. Note that $|\mathcal{F}| = \epsilon^2 t^4$, since there are $t$ choices for each $z$ and $x$, and $\epsilon t$ choices for each $y$ and $y'$. Note that a deterministic algorithm cannot produce an incorrect output on $G \cup G'$ as it will have error probability greater than $\delta$. Therefore, we now need to show that any deterministic algorithm that produces an incorrect output on at most a $2\delta$ fraction of graphs in $\mathcal{F}$ has $\Omega(t^2)$ message complexity. So suppose for sake of contradiction that there exists such an algorithm that outputs a $(1/2 + \epsilon)$-approximate independent set $S$ with message complexity $o(t^2)$

There can only be $o(t)$ vertices in $Y$, and $o(t)$ vertices in $Y'$ that send more than $\delta t$ messages. For each of the $(1 - o(1))\epsilon t$ nodes $y \in Y$ there are at least $(1 - \delta)t$ nodes $z \in Z$, and consequently at least $(\epsilon/2)t$ nodes $z \in Z \cap S$ (as $|Z \cap S| \ge \epsilon t$, and $\delta < \epsilon/2$), such that no message is sent over edge $e = \{y,z\}$. Similarly, for each of the $(1 - o(1))\epsilon t$ nodes $y' \in Y'$ there are at least $(\epsilon/2)t$ nodes $x' \in X' \cap S$, such that no message is sent over edge $e' = \{x',y'\}$. The proof of Lemma~\ref{lemma:kt0-maxIS-correctness} says that the deterministic algorithm is incorrect on all these graphs $G_{e,e'}$. 

The algorithm is incorrect on at least $(1 - o(1))^2(1/2)^2 \epsilon^4 t^4$ graphs in $\mathcal{F}$, that is, on at least a $(\epsilon^2/4 - o(1))$ fraction of graphs in $\mathcal{F}$. This is a contradiction if $2\delta < \epsilon^2/4 - o(1)$. 
\end{proof}

 \section{Near-Linear Bounds for Random Graphs}
\label{section:UBRG}

We first provide a message- and round-efficient randomized greedy MIS in $G(n,p)$ random graphs, in Subsection \ref{subsec:greedyMIS}.\footnote{These are graphs for which each (possible) edge $e = \{x,y\} \in V^2$ occurs independently with probability $0 < p < 1$ (where $p$ may be a function of $n$).} With it, we give distributed algorithms using only $\tilde{O}(n)$ messages to compute with high probability constant-factor approximations for MaxIS, MDS, MVC and MaxM in $G(n,p)$ random graphs, in Subsections \ref{subsec:approxInRandomGraphs} and \ref{subsec:maxMatchingApproxRG}.

\paragraph*{Preliminaries}
We state some well-known Chernoff bounds as well as a basic (probabilistic) claim about the maximum and minimum degree of random graphs. Note that the claim's probabilistic guarantees depend on the randomly chosen graph, or more precisely, on the indicator variables $\{\Ind_e\}_{e \in V^2}$ that denote whether edge $e$ occurs in $G$.
Later in this section, probabilities are taken over both the natural probability space of the randomly chosen graph and the natural probability space of the algorithm's random choices.

\begin{lemma}[Chernoff Bounds \cite{Upfalbook}]
\label{lem:ChernoffBound}
Let $X_1,\ldots,X_k$ be independent $\{0,1\}$ random variables. Let $X$ denote the sum of the random variables, $\mu$ the sum's expected value and $\mu_L, \mu_H$ be any value respectively smaller and greater than $\mu$. Then,
\begin{enumerate}
    \item $\Pr[X \leq (1-\delta) \mu_L] \leq \exp(- \delta^2 \mu_L / 2)$  for $0 \leq \delta \leq 1$,
    \item $\Pr[X \geq (1+\delta) \mu_H] \leq \exp(- \delta^2 \mu_H / 3)$ for $0 \leq \delta \leq 1$,
    \item $\Pr[X \geq (1+\delta) \mu_H] \leq \exp(- \delta^2 \mu_H / (2+\delta))$ for $\delta \geq 1$.
\end{enumerate}

\end{lemma}

\begin{claim}
\label{claim:randomGraphMaxDegree}
The maximum degree of $G$ is $O(\max\{np,\log n\})$ with high probability.
Additionally, if $p \geq 16 \log(n)/n$, the minimum degree of $G$ is $\Omega(np)$ with high probability.
\end{claim}

\begin{proof}
For any node $v \in V$, let $d_v$ denote the degree of $v$ in $G$. Note that $d_v$ is a random variable (that depends only on the randomly chosen graph). Moreover, it is the sum of $n-1$ independent, identically distributed Bernoulli random variables with parameter $p$. Thus, $\E[d_v] = (n-1)p \leq np$. Consequently, it follows from the Chernoff bound (see Lemma \ref{lem:ChernoffBound}) that for any $\delta \geq 0$, $\Pr[d_v \geq (1+\delta) np] \leq \exp(- \delta^2 np / (2+\delta))$.

In a first time, we assume that $p \geq \log(n)/n$, which implies that $np \geq \log n$. Hence,  $\Pr[d_v \geq (1+\delta) np] \leq  \exp(- \delta^2 \log n / (2+\delta))$. Setting $\delta = 4$, $\Pr[d_v \geq 5 np] \leq 1/n^2$.
Next, assume that $p \leq \log(n)/n$ and take $\delta = c \log(n)/(np)$ (for some constant $c > 1$). Since $\delta \geq 1$, rewriting the Chernoff bound gives $\Pr[d_v \geq 2 \delta np] \leq \exp(- \delta np / 3)$. Replacing the value of $\delta$ in the expression results in  $\Pr[d_v \geq 2 c \log n] \leq \exp(- c \log n / 3)$. Thus, for $c = 6$, $\Pr[d_v \geq 12 \log n] \leq 1/n^2$.
Therefore, by a union bound on the nodes of $G$, the maximum degree is upper bounded by $\max\{5np,12 \log n\}$ with probability at least $1-1/n$. Greater values of $\delta$ allow for better probability bounds, say $1-1/n^2$.

Finally, assume $p \geq 16 \log(n)/n$ and note that $\E[d_v] = (n-1)p \geq np/2$ for $n \geq 2$. Hence, it follows from the Chernoff bound (see Lemma \ref{lem:ChernoffBound}) that for any $\delta \geq 0$, $\Pr[d_v \leq (1-\delta) np/2] \leq \exp(- \delta^2 np / 4) \leq \exp(- 4 \delta^2 \log(n))$. Setting $\delta = 1/2$, $\Pr[d_v \leq np/4] \leq 1/n^2$.
\end{proof}

The above claim implies that for $p = O(\log(n)/n)$, the communication graph has at most $O(n \log n)$ edges w.h.p. Hence, any well-known fast MIS algorithm (say, randomized greedy MIS) takes $O(\log n)$ rounds and uses only $O(n \log^2 n)$ messages. As a result, from now on, we consider that $p = \Omega(\log(n)/n)$.

\subsection{Randomized Greedy MIS in \texorpdfstring{$O(\log^2 n)$}{O(log^2 n)} Rounds and \texorpdfstring{$\tilde{O}(n)$}{Õ(n)} Messages}
\label{subsec:greedyMIS}

The (randomized greedy) MIS algorithm works in $O(\log n)$ phases  (where $n$ is known to all nodes initially). More precisely, there are $\phaseBound+1$ phases, where $q = \cproba / p n$ is a parameter (fixed to ensure correctness with high probability in the analysis) and $0 < q \leq 1$. Initially, all nodes start undecided. As the algorithm progresses, more and more nodes become decided by either joining the MIS or having a neighboring node join the MIS.

The first $\phaseBound$ phases each consist of $15$ iterations, where each iteration takes $O(\log n)$ rounds. Each iteration decreases the number of undecided nodes by a constant factor, and each phase decreases the number of undecided nodes by half. On the other hand, the final phase consists of a single $O(\log n)$ round iteration. During that iteration, all undecided nodes---at most $O(\log(n)/p)$ of them---run an MIS algorithm and become decided. Since nodes have maximum degree $O(np)$ w.h.p. (see Claim \ref{claim:randomGraphMaxDegree}), the last phase uses only $\tilde{O}(n)$ messages w.h.p.

Next, we describe in more detail one iteration, say $j$, of phase $i \in \{1,\ldots,\phaseBound\}$.\footnote{If nodes only know some polynomial upper bound $N$ on $n$, the algorithm can be adapted by setting $q = (100 \log N) / (p N)$. The analysis remains mostly the same. Indeed, although some of the early phases with small probability $q_i$ provide no progress, they also use few messages---more precisely, $\tilde{O}(n)$ messages.} In the first round, each (up to now) undecided node becomes active with probability $q_i = 2^{i-1} q$ --- which should result with high probability in $O(\log(n) / p)$ active nodes --- and active nodes send a message to all neighboring nodes. In the second round, nodes (that received a message) answer whether they are active or not. Hence, after the first two rounds, active nodes know for each incident edge whether it leads to an active node or not.
For the next $T_{MIS}(n) = O(\log n)$ rounds (where $T_{MIS}(n)$ is an upper bound on the runtime of distributed randomized greedy MIS over $n$-node graphs), the active nodes (denoted by $A_{i,j}$) execute the distributed randomized greedy MIS algorithm on $G[A_{i,j}]$.\footnote{We use $G[U]$ to denote the subgraph of $G$ induced by set $U \subset V$.}
Let us denote the computed MIS by $M_{i,j}$. Crucially, computing $M_{i,j}$ does not depend on the edges $E_{i,j}$ between active and non-active nodes. More precisely, active nodes ignore any information they may have on $E_{i,j}$, such as which and how many incident edges are in $E_{i,j}$, when executing distributed randomized greedy MIS on $G[A_{i,j}]$.
Finally, in the phase's last round, nodes in $M_{i,j}$ send a message to all of their neighbors (not only the active ones) and become decided (as part of the MIS). Nodes that receive a message in this round also become decided (as nodes with a neighbor in the MIS).

\paragraph*{Analysis} We denote, for any phase $i \in \{1, \ldots, \phaseBound + 1\}$, by $V_i$ the set of undecided nodes at the start of phase $i$ and by $n_i = |V_i|$ the size of $V_i$. Moreover, for any iteration $j \in \{1,\ldots,15\}$ of phase $i \leq \phaseBound$, we denote by $V_{i,j}$ the set of undecided nodes at the start of iteration $j$ and by $n_{i,j} = |V_{i,j}|$ the size of $V_{i,j}$.

We start by bounding the number of active nodes in each iteration (see Lemma \ref{lem:activeSubgraphSize}), as well as by bounding the maximum degree on the subgraph induced by these active nodes (see Lemma \ref{lem:activeSubgraphDegree}).

\begin{lemma}
\label{lem:activeSubgraphSize}
For any iteration $j \in \{1,\ldots,15\}$ of phase $i \in \{1, \ldots, \phaseBound\}$, if $n_{i,j} \leq n/2^{i-1}$ then $|A_{i,j}| \leq 125 \log(n)/p$ \withProba{with probability at least $1-1/(120 n^2)$}. Additionally, if $n_{i,j} \geq n/2^{i}$ then $|A_{i,j}| \geq 35 \log(n)/p$ \withProba{with probability at least $1-1/(60 n^2)$}.
\end{lemma}

\begin{proof}
Assume that for iteration $j$ of phase $i$, $n_{i,j} \leq n/2^{i-1}$.
The indicator random variables $\{\Ind[v \in A_{i,j}]\}_{v \in V_{i,j}}$ are independent, identically distributed Bernoulli random variables with parameter $q_i$.
Moreover, their sum is the number of active nodes $|A_{i,j}|$: i.e., $|A_{i,j}| = \sum_{v \in V_{i,j}} \Ind[v \in A_{i,j}]$. The expected number of active nodes is $\E[|A_{i,j}|] = n_{i,j} q_i \leq \cproba/p$. Thus, the Chernoff bound (see Lemma \ref{lem:ChernoffBound}) shows that for any $0 \leq \delta \leq 1$:
\[ \Pr[|A_{i,j}| \geq (1+\delta) \cproba/p] \leq \exp(- \delta^2 \cproba / (3 p)) \leq \exp(- \delta^2 \cproba / 3) \]
The second inequality holds since $0 < p < 1$. Setting $\delta = 1/4$, we get $\Pr[|A_{i,j}| \geq 125 \log(n)/p] \leq \exp(- 100 \log(n)/48) \leq 1/(120 n^2)$, for large enough $n$.

Next, consider $n_{i,j} \geq n/2^{i}$. Then, the expected number of active nodes is $\E[|A_{i,j}|] = n_{i,j} q_i \geq 50 \log(n) / p$. Thus, the Chernoff bound shows that for any $0 \leq \delta \leq 1$:
\[ \Pr[|A_{i,j}| \leq (1-\delta) 50 \log(n)/p] \leq \exp(- \delta^2 50 \log(n) / (2p)) \leq \exp(- \delta^2 25 \log n) \]
Set $\delta = 3/10$. Then, $\Pr[|A_{i,j}| \leq 35 \log(n)/p] \leq \exp(- \frac{9}{4} \log(n)) \leq 1/(60 n^2)$ for large enough $n$.
\end{proof}

\begin{lemma}
\label{lem:activeSubgraphDegree}
For any iteration $j \in \{1,\ldots,15\}$ of phase $i \in \{1, \ldots, \phaseBound\}$, if $n_{i,j} \leq n/2^{i-1}$ then the maximum degree of $G[A_{i,j}]$ is $O(\log n)$ with high probability; in fact, it is at most $175 \log n$ \withProba{with probability at least $1-1/(60 n^2)$. }
\end{lemma}

\begin{proof}
Choosing the communication graph determines the value of all indicator variables $\{\Ind[(v,u) \in E]\}_{v,u \in V}$ (denoting whether edge $(v,u)$ is part of the communication graph). Thus, these variables are decided before the algorithm even starts.
However, we shall use the principle of deferred decisions \cite{Upfalbook} to argue that in the analysis, the indicator random variables $\{\Ind[(v,u) \in E]\}_{v,u \in A_{i,j}}$ can be revealed after the set of active nodes $A_{i,j}$ is decided.

More concretely, prior to iteration $j$ of phase $i$, the algorithm sends messages only on edges incident to (prior) active nodes and these (prior) active nodes become decided---as a MIS node or the neighbor of one. Hence, prior to iteration $j$ (of phase $i$), the algorithm sends no messages over any edges between nodes in $V_{i,j}$. As a result, by the principle of deferred decisions, all indicator random variables $\{\Ind[(v,u) \in E]\}_{v,u \in V_{i,j}}$ can be revealed after $V_{i,j}$ is decided.
Moreover, nodes in $V_{i,j}$ become active (i.e., join $A_{i,j}$) independently of any information on the communication graph.
Hence, by the principle of deferred decisions, the variables $\{\Ind[(v,u) \in E]\}_{v,u \in A_{i,j}}$ can be revealed after the set of active nodes $A_{i,j}$.

By Lemma \ref{lem:activeSubgraphSize}, $|A_{i,j}| \leq 125 \log(n) / p$ \withProba{with probability at least $1-1/(120 n^2)$}. We condition on this for the remainder of the proof. Consider an arbitrary $v \in A_{i,j}$. The degree of $v$ in $G[A_{i,j}]$, now denoted by $d_{i,j}(v)$, satisfies $d_{i,j}(v) = \sum_{u \in A_{i,j}, u \neq v} \Ind[(v,u) \in E]$. Since the indicator variables $\{\Ind[(v,u) \in E]\}_{v,u \in A_{i,j}}$ are independent, identically distributed Bernoulli random variables with parameter $p$, the expectation of $d_{i,j}(v)$ is $\E[d_{i,j}(v)] \leq |A_{i,j}| p \leq 125 \log n$.
We can show, using the Chernoff bound (see Lemma \ref{lem:ChernoffBound}), that for any $0 \leq \delta \leq 1$, $\Pr[d_{i,j}(v) \geq (1+\delta) 125 \log n] \leq \exp(- \delta^2 125 \log(n) / 3)$. Setting $\delta = 4/10$, we get $\Pr[d_{i,j}(v) \geq 175  \log(n)] \leq \exp(- 20 \log(n) / 3) \leq 1/(120 n^3)$, for large enough $n$. Finally, we obtain the second half of the lemma statement by a union bound on the nodes of $A_{i,j}$ and adding the probability that $|A_{i,j}| > 125 \log(n)/p$.
\end{proof}

Given the previous two lemmas, we can now show that each phase decreases the number of undecided nodes by at least half.

\begin{lemma}
\label{lem:progress}
Consider some phase $i \in \{1, \ldots, \phaseBound\}$ such that $n_{i,j} \geq n/2^{i}$. Then, at least half of the nodes of $V_i$ become decided \withProba{with probability at least $1-1/n^2$}.
\end{lemma}

\begin{proof}
Below, we show that in each iteration $j \in \{1,\ldots,15\}$ of phase $i$, the number of undecided nodes at the end of the iteration is at most $0.95 n_{i,j}$ \withProba{with probability at least $1-1/(15n^2)$}. The lemma statement follows straightforwardly, since the number of undecided nodes at the end of phase $i$ is at most $(0.95)^{15} n_i \leq n_i/2$ \withProba{with probability at least $1-1/n^2$}.

For the remainder of the proof, consider some iteration $j$ of phase $i$.
First, we show that any MIS on $G[A_{i,j}]$ contains at least $1/(5 p)$ nodes \withProba{with probability at least $1-1/(30n^2)$}. Indeed, suppose by contradiction that there exists a MIS $M$ on $G[A_{i,j}]$ with strictly less than $1/(5 p)$ nodes. By Lemma \ref{lem:activeSubgraphDegree}, the maximum degree of $G[A_{i,j}]$ is $175 \log n$ \withProba{with probability at least $1-1/(60 n^2)$} and thus $M$ covers strictly less than $175 \log(n) / (5 p) = 35 \log(n)/ p$ nodes in $G[A_{i,j}]$. However, $|A_{i,j}| \geq 35 \log(n)/ p$ \withProba{with probability at least $1-1/(60 n^2)$} by Lemma \ref{lem:activeSubgraphSize}. Hence, \withProba{with probability at least $1-1/(30 n^2)$}, there exists a node in $A_{i,j}$ with no neighbors in $M$, leading to a contradiction.

Next, we show the following claim: at least $0.05 n_{i,j}$ nodes become decided \withProba{with probability at least $1-1/(15n^2)$} at the end of iteration $j$. (Or in other words, the number of undecided nodes is at most $0.95 n_{i,j}$ \withProba{with probability at least $1-1/(15n^2)$} at the end of iteration $j$.) 

Let $U_{i,j} = V_{i,j} \setminus A_{i,j}$. As shown in the proof of Lemma \ref{lem:activeSubgraphDegree}, the indicator random variables $\{\Ind[(v,u) \in E]\}_{v \in A_{i,j}, u \in U_{i,j}}$ can be revealed after $V_{i,j}$ is decided. Nodes in $V_{i,j}$ become active independently of any information on the communication graph. Moreover, $M_{i,j}$ is computed independently of any information on $E_{i,j} = \{(v,u)\}_{v \in A_{i,j}, u \in U_{i,j}}$. Hence, by the principle of deferred decisions, the variables $\{\Ind[(v,u) \in E]\}_{v \in A_{i,j}, u \in U_{i,j}}$ can be revealed after $A_{i,j}$ and $M_{i,j}$ are decided.
Consequently, let us fix $M_{i,j}$ (without revealing $E_{i,j}$) and condition on $|M_{i,j}| \geq 1/(5 p)$, which holds \withProba{with probability at least $1-1/(30 n^2)$} (as shown above).
Consider the indicator random variables $\{\Ind_v\}_{v \in V_{i,j}}$ denoting that $v$ is in $M_{i,j}$ or has at least one neighbor in $M_{i,j}$. Then, the number of decided nodes at the end of iteration $j$ is the random variable $X = \sum_{v \in V_{i,j}} \Ind_v$. Note that $\Pr[\Ind_v = 1] = 1$ for all nodes $v \in A_{i,j}$ and that $\Pr[\Ind_v = 1] = 1 - (1-p)^{|M_{i,j}|} \geq 1- \exp(-p |M_{i,j}|) \geq 1- e^{-1/5}$ for all nodes $v \in  U_{i,j}$ (where we now reveal $E_{i,j}$). Hence, variable $X$ stochastically dominates the sum of $n_{i,j}$ independent, identically distributed random variable with parameter $1- e^{-1/5}$ and $\E[X] \geq n_{i,j} (1- e^{-1/5})$. Following which, the Chernoff bound (see Lemma \ref{lem:ChernoffBound}) implies that for any $0 \leq \delta \leq 1$,
\[\Pr[X \leq (1-\delta) n_{i,j} (1- e^{-1/5})] \leq \exp(- \delta^2 n_{i,j} (1- e^{-1/5})/2)
\]
Since $n_{i,j} \geq n/2^i \geq n/2^{\phaseBound} \geq nq/2 = 50 \log(n)/p$ and $0 < p < 1$,
\[\Pr[X \leq (1-\delta) n_{i,j} (1- e^{-1/5})] \leq \exp(- \delta^2 50 \log(n) (1- e^{-1/5})/2)
\]
Setting $\delta = 7/10$, we get $\delta^2 50 (1- e^{-1/5})/2 > 2$ and $(1-\delta) (1- e^{-1/5}) \geq 0.05$. Hence, $\Pr[X \leq 0.05 n_{i,j}] \leq 1/(30 n^2)$, for large enough $n$.
Since $|M_{i,j}| \geq 1/(5 p)$ fails \withProba{with probability at most $1/(30 n^2)$}, we get the desired probability bound.
\end{proof}

Finally, we prove that the proposed MIS algorithm is message- and round-efficient.

\begin{theorem}
In $G(n,p)$ random graphs, MIS can be solved with high probability
with $O(\log^2 n)$ round complexity and $\tilde{O}(n)$ message complexity (w.h.p.).
\end{theorem}

\begin{proof}
The correctness is straightforward. In each iteration of the first $\phaseBound$ phases, we compute an independent set (w.h.p.) and neighboring nodes become decided. In the final phase, we compute a maximal independent set (w.h.p.) on the subgraph induced by the remaining undecided nodes. The union of all independent sets form a maximal independent set of the communication graph with high probability.

Next, let us bound the round and message complexity. The round complexity bound is easily obtained. As for the message complexity bound, we first upper bound the number of undecided nodes in each phase. With this, we upper bound the number of active nodes and their incident edges in each iteration of the first $\phaseBound$ phases. Since only active nodes send messages (over all incident edges) in these phases, this bounds that first half of the message complexity. As for the final phase, the upper bound on the number of undecided nodes allows us to directly show that there remains few undecided nodes at the start of the final phase, bound their incident edges and thus also bound the second half of the message complexity.

To do so, we first show by induction that the following claim holds \withProba{with probability at least $1-1/(2n)$}: for any phase $i \in \{1, \ldots, \phaseBound+1\}$, $n_i \leq n/2^{i-1}$. This inequality holds trivially for the first phase. As for the induction step, consider some phase $i \in \{1, \ldots, \phaseBound\}$ for which the induction hypothesis holds. Then, if $n_i < n/2^{i}$, it holds straightforwardly that $n_{i+1} \leq n/2^{i}$ and the induction hypothesis holds for phase $i+1$. Otherwise, if $n_i \geq n/2^{i}$, Lemma \ref{lem:progress} implies that the number of undecided nodes decreases by half---that is, $n_{i+1} \leq n/2^{i}$---\withProba{with probability at least $1-1/n^2$}. Thus, the induction step fails \withProba{with probability at least $1-1/n^2$}, and hence by union bound the statement fails \withProba{with probability at most $1/(2n)$}, for large enough $n$.
In the remainder of the proof, we condition on the above claim.

For any iteration $j \in \{1,\ldots,15\}$ of phase $i \in \{1, \ldots, \phaseBound\}$, the above claim implies that $n_i \leq n/2^{i-1}$. In which case, Lemma \ref{lem:activeSubgraphSize} implies that there are at most $|A_{i,j}| = O(\log(n) / p)$ active nodes \withProba{with probability at least $1-1/n^2$}. Moreover, all nodes---and thus active nodes---have maximum degree $O(np)$ \withProba{with probability at least $1-1/n^2$}, by Claim \ref{claim:randomGraphMaxDegree}. Hence, there are at most $O(n \log n)$ edges incident to $A_{i,j}$ \withProba{with probability at least $1-2/ n^2$}.
As for the final phase $i = \phaseBound + 1$, the above claim implies that $n_i \leq n/2^{i-1} = n q / 2 = 50 \log(n) / p$. Since all nodes have maximum degree $O(np)$ \withProba{with probability at least $1-1/n^2$} (see Claim \ref{claim:randomGraphMaxDegree}), there are at most $O(n \log n)$ edges incident to undecided nodes in the final phase \withProba{with probability at least $1-1/n^2$}.
Finally, a union bound over all iterations of all phases, and adding in the error probability of the above claim, implies that the message complexity is $O(n \log^2 n)$ \withProba{with probability at least $1-1/n$}, for large enough $n$.
\end{proof}

\begin{lemma}
The proposed MIS algorithm implements randomized greedy MIS.
\end{lemma}

\begin{proof}
    We first define a uniformly random ordering of the nodes and then argue that the proposed algorithm implements the sequential greedy MIS algorithm with respect to that order.
    
Note that by the algorithm definition, each node $v \in V$ becomes active in at most one iteration (and one phase). Then, fix such a phase and iteration pair in advance for all nodes; for node $v \in V$, denote the pair by $p(v)$. More concretely, each node $v \in V$ independently chooses the pair $p(v) = (1,1)$ with probability $q$, $p(v) = (1,2)$ with probability $(1-q) q$, and so on.
    The chosen phase and iteration pairs (through their lexicographical ordering) induce a partial ordering on all nodes, with the property that each node gets independent and equal probability to choose any given rank in that ordering. Furthermore, within each set $P_{i,j}$ of nodes that have the same phase and iteration pair $(i,j)$, assume nodes choose IDs independently and uniformly at random in $[1,n^3]$ --- these are unique with high probability --- and order correspondingly. The resulting (combined) ordering of $V$ is a uniformly random node ordering with high probability.

    Finally, it is clear that the proposed algorithm implements sequential greedy MIS with respect to that order.
Indeed, within each iteration active nodes run the distributed greedy MIS algorithm (with IDs chosen independently and uniformly at random in $[1,n^3]$) and in the last round, nodes that have entered the MIS in that iteration inform all of their neighbors.
\end{proof}

\subsection{Solving Approximation Problems in \texorpdfstring{$O(\log^2 n)$}{O(log^2 n)} Rounds and \texorpdfstring{$\tilde{O}(n)$}{Õ(n)} Messages}
\label{subsec:approxInRandomGraphs}

In this subsection, we consider random graphs in the connectivity regime, in particular we consider $G(n,p)$ random graphs with $p \geq 40 (\log n) / n$.\footnote{For sparse random graphs, $\tilde{O}(m) = \tilde{O}(n)$ messages suffice to solve \mxm, MVC, MDS, and \mxis (with time encoding).}
For these graphs, the distributed randomized greedy MIS algorithm from the previous subsection solves (with high probability) $(2+o(1))$-approximate maximum independent set and $(1+o(1))$-approximate minimum dominating set.
This follows from well-known results regarding random graphs that we state below. (More concretely, see Grimmett and McDiarmid \cite{GM75} and Frieze and McDiarmid \cite{FM97Survey} for the independent set bounds, and  Glebov, Liebenau and Szab\'{o} \cite{GLS15} as well as Wieland and Godbole \cite{WG01} for the dominating set bound.)

\begin{lemma}[\cite{GM75,FM97Survey, GLS15, WG01}]
    For any random graph $G(n,p)$ in the connectivity regime, let $\alpha$ be the size of the maximum independent set, $\gamma$ the size of the minimum dominating set and $\sigma$ be the size of the independent set computed by the sequential randomized greedy MIS algorithm. 
    
    Then, it holds with high probability that $\alpha \sim 2 \log_{1/(1-p)} np $, $\gamma = (1+o(1))\ln(np)/p \geq (1+o(1)) \log_{1/(1-p)} np$ and $\sigma \sim \log_{1/(1-p)} np$.
\end{lemma}

Note that $G(n,p)$ random graphs in the connectivity regime are also known to admit perfect matchings \cite{ER66} when $n$ is even and near-perfect matchings when $n$ is odd. As a result, any vertex cover must contain at least $n/2$ or $(n-1)/2$ nodes. Taking the complement of the MIS computed by our round- and message-efficient distributed MIS algorithm, we obtain a $(2-o(1))$-approximate solution to minimum vertex cover. We summarize our approximation results in the claim below.

\begin{lemma}
In $G(n,p)$ random graphs (with $p \geq 40(\log n)/n$), $(1/2-o(1))$-approximate MaxIS, $(1+o(1))$-approximate MDS and $(2-o(1))$-approximate MVC can be solved (w.h.p.) with $\tilde{O}(n)$ messages and $O(\log^2 n)$ rounds in $\ktzero$ $\congest$.
\end{lemma}

\subsection{Maximum Matching}
\label{subsec:maxMatchingApproxRG}

\subsubsection{Perfect Matching in $\tilde{O}(n)$ Rounds and $\tilde{O}(n)$ Messages}

In $G(n,p)$ random graphs with $p \geq c\log n/n$ (for some fixed constant $c$), it is well-known that there exists an
Hamiltonian cycle. Thus a perfect matching (or near perfect matching if $n$ is odd)  exists
in such graphs. The Hamiltonian cycle can also be constructed by using the rotation algorithm due
to Angluin and Valiant \cite{valiant} (see also \cite{Upfalbook}). The rotation algorithm takes $O(n \log n)$ steps
in the sequential setting.  The algorithm always maintains a path. Initially
the path consists of a single (arbitrary) node. One step of the algorithm consists of  extending the \emph{head} of the path (which is the last vertex in the path) by doing 
a simple random walk. Specifically, if the path is $v_0, \dots v_k$, then a random incident edge of the head ($v_k$), say 
$(v_k, v_i)$ is
chosen to extend the path. If $v_i$ is a vertex that does not belong to the part, we continue the process with 
by extending the path to  $v_0, \dots v_k, v_{k+1}$ with the new head $v_{k+1} = v_i$.
Otherwise, if $v_i$ belongs to the path, say $v_i= v_j$, we perform a {\em rotation} as follows:
the new path becomes $v_0 \dots, v_j, v_k, \dots, v_{j+1}$, with $v_{j+1}$ being the new head. (In the case then $v_{i} = v_{k-1}$, there is no change in the head of the path). It is known that the rotation algorithm
finds an Hamiltonian cycle in $G(n,p)$ random graphs with $p \geq c \log n/n$ (e.g., constant $c =40 $) with high
probability \cite{Upfalbook}. It can be shown  that the rotation algorithm can be implemented as a distributed algorithm
in a straightforward way (in $\ktzero$) that takes $\tilde{O}(n)$ messages and $\tilde{O}(n)$ rounds (see the
Distributed rotation algorithm of \cite{icdcsreza}).
Once the Hamiltonian cycle is constructed, it is easy to extract a perfect (or near-perfect) matching using
$O(n)$ messages and $O(n)$ rounds. Thus we can state the following theorem.

\begin{theorem}
\label{thm:rgmm}
There is a distributed algorithm that, with high probability, computes a perfect (or near-perfect) matching in
 $G(n,p)$ random graphs (with $p \geq 40(\log n)/n$) $\tilde{O}(n)$ messages and $\tilde{O}(n)$ rounds in $\ktzero$ $\congest$.
\end{theorem}

\subsubsection{Constant Factor Approximation in \texorpdfstring{$O(1)$}{O(1)} Rounds and \texorpdfstring{$O(n)$}{O(n)} Messages}

We present a simple message-efficient algorithm in $\ktzero$ that takes $O(n)$ message, runs in $O(1)$ rounds and
gives a constant factor approximation in regular (and almost regular) graphs and random graphs (which are almost
regular with high probability). Actually, we present a somewhat more general result that
gives an expected $O((\Delta/\delta)^2)$-factor approximation to the maximum matching in an arbitrary graph, where  $\Delta$ and $\delta$ are the maximum
and minimum degrees of the graph. As a corollary of this result, we obtain  an expected constant factor approximation
in regular (or almost regular) graphs as well as random graphs. Moreover, we show that this constant factor approximation holds also with high probability in regular graphs and random graphs.

\paragraph{The algorithm}
Let $G =(V,E)$ be a graph on $n$ nodes with maximum degree $\Delta$ and minimum degree $\delta$.
The following is the algorithm, executed by each node $u \in V$. Let $\alpha$ be a parameter (fixed in the analysis).
\begin{enumerate}
\item $u$ independently, with probability $\alpha$,  chooses a random incident edge  and proposes this
edge  to the other endpoint (say $v$).
\item $v$ accepts this edge, if no other incident edge of $v$ is proposed by $v$ or its neighbors (except $u$).
\item $v$ conveys its acceptance to $u$ and the edge $(u,v)$ is included in the matching.
\end{enumerate}

\paragraph{Analysis}
We show the following theorem.

\begin{theorem}
\label{th:mm}
The above algorithm  computes an expected $O(r^2)$-approximation of maximum matching in a graph $G$,
where  $r = \Delta/\delta$ (ratio of the maximum and minimum degree of the graph). The algorithm takes
constant rounds and $O(n)$ messages (deterministically). Moreover, if $r = o(n^{1/4}/\log n)$, then the approximation factor
holds with high probability.
\end{theorem}

\begin{proof}
It is clear that the edges included is indeed a matching and the algorithm takes 2 rounds and uses $O(n)$ messages.
We next bound the approximation given by the algorithm.

Fix an edge $(u,v)$. Let the indicator random variable $X_{u,v}$ be 1 if the edge is included in  the matching, otherwise 0. We compute $E[X_{u,v}] = \Pr(X_{u,v} =1)$.
The edge $(u,v)$ will be included in the matching if it is proposed by $u$ (say)
and no other edge that is incident on $v$ is proposed. The probability that $u$ proposes edge $(u,v)$ is $\alpha/d(u)$.
We next upper bound the probability that some edge incident on $v$ is proposed by a neighbor other than $u$. This probability is at most
$\sum_{w\in N(v) \setminus \{u\}} \frac{\alpha}{d(w)} \leq \frac{\Delta\alpha}{\delta} = r\alpha$,
where the sum is over all neighbors of $v$ (including $v$ and excluding $u$).
Thus, the probability that no edge  incident on $v$ is proposed is at least  $1 - r\alpha = 1/2$, if $\alpha$
is chosen as $1/2r$.
Since the random choices made by $u$ and $v$'s other neighbors (including $v$) are independent, we have
 $\Pr(X_{u,v} = 1) \geq \frac{\alpha}{2d(u)}  \geq \frac{\alpha}{2\Delta}$.
 Hence by linearity of expectation, the expected number of edges added to the matching is
 at least $\frac{\alpha}{2\Delta} (n\delta/2) \geq n/8r^2$.
 Since a maximum matching can be of size at most $n/2$, this gives an expected approximation factor
 of $O(r^2)$.

 Finally, we give a concentration result on the  number of edges added to the matching when $r$ is small. 
 The number of edges added to the matching is a function that depends on $n$ {\em independent} random variables $\{X_v\}_{v\in V}$, where $X_v$ is the random choice made by node $v$. This function, call it $f$, satisfies the Lipschitz condition with bound 2,
 since changing the value of any random variable changes the value of the function by at most 2. This is because
 changing the random choice of $X_v$ for any $v$ can affect the matched edge incident on $v$ and one other neighbor of
 $v$.  Thus one can apply the McDiarmid's (bounded difference) inequality\cite{Upfalbook}[Theorem 13.7] to show the following concentration
 bound:
 $$\Pr(|f-E[f]| \geq t) \leq 2e^{-2t^2/4n}.$$
 Since $E[f] \geq n/8r^2$, setting $t = n/16r^2$, we have
 $$\Pr(f < n/16r^2)  \leq 2e^{-n/512r^4} = O(1/n),$$
 if $r = o(n^{1/4}/\log n)$, and thus an approximation factor of $O(r^2)$ with high probability.
\end{proof}

\begin{corollary}
The above algorithm gives, with high probability, a constant factor approximation to the maximum matching problem
in regular graphs and in $G(n,p)$ random graphs (with $p \geq 16(\log n)/n$) in $O(1)$ rounds and $O(n)$ messages in $\ktzero$.
\end{corollary}

\begin{proof}
Since $r=1$ in regular graphs, the result is immediate from Theorem \ref{th:mm}.  In $G(n,p)$ random graphs,  as shown in
Claim \ref{claim:randomGraphMaxDegree}, for $p \geq 16(\log n)/n$,  $r = O(1)$ with high probability
and hence the result again follows from Theorem \ref{th:mm}.
\end{proof}

 \section{Conclusion and Open Problems}
\label{sec:conc}

In this work, we almost fully quantify the message complexity of  four fundamental graph optimization problems ---
\mxm, MVC, MDS, and \mxis. These problems represent a spectrum of hardness of approximation in the sequential setting, ranging from \mxm that is exactly solvable in polynomial-time to \mxis that is 
hard to approximate even to a $O(n^{1-\epsilon})$-factor for any $\epsilon > 0$. 
We have shown that $\tilde{\Omega}(n^3)$ messages are needed to solve MVC, MDS, and \mxis exactly in the $\ktzero$ \congest{} model.
The message complexity of exact \mxm is an intriguing open question and as we point out in Remark \ref{remark:maxm}, the lower bound technique we use to obtain the cubic bounds for MVC, MDS, and \mxis, cannot be used for \mxm. Furthermore, there has been recent progress on improving the round complexity of exact \mxm \cite{KitimuraIzumiITIS2022}, though it is not clear if techniques from this line of work can be used to obtain $o(n^3)$ message algorithms for exact \mxm.
Another set of open questions relate to our quadratic lower bounds.
For MDS and \mxis, our lower bounds are for constant-factor approximations, for specific constants. 
Can these lower bounds  be extended to any $\alpha$-approximation algorithm? Such lower bounds would be a function of the graph size as well as $\alpha$ (similar to our lower bounds for \mxm and MVC).
For MDS, such a general lower bound would have to account for the $O(n^{1.5})$ message upper bound for $O(\log \Delta)$-approximation for MDS \cite{GotteKSWAlgoSensors2021,GOTTE2023113756}.

\newpage

\bibliographystyle{plain}
\bibliography{references,references1}

\newpage

\appendix
\section{Lower Bounds in \texorpdfstring{\ktone{} \congest{}}{KT-1 CONGEST}}
\label{section:KTOneLB}

Message complexity is sensitive to initial local knowledge of nodes.
For any integer $\rho > 0$, in the \ktrho{} model
each node $v$ is provided initial  knowledge of (i) the \texttt{ID}s of all nodes at distance at most $\rho$ from $v$ and (ii) the neighborhood of every vertex at distance at most $\rho-1$ from $v$.

The two most commonly used models of initial knowledge in distributed algorithms are $\ktzero$ (which we considered in the main body of the paper) and $\ktone$.
In the $\ktone$  model, each node has initial knowledge of itself and the \texttt{ID}s of its neighbors. Note that only knowledge of the \texttt{ID}s of neighbors is assumed, not other information such as the degree of the neighbors.
In \ktone{} \congest{}, we can solve \textit{any} problem using $\OT(n)$ messages, provided exponentially many rounds are allowed (see Appendix A in \cite{RobinsonSODA21}). Therefore, one cannot expect to obtain the unconditional lower bounds we proved in Section~\ref{section:KTZeroLB} for \ktzero{}.

We circumvent this obstacle and prove lower bounds in two distinct settings: (1) where algorithms are \emph{time-bounded}, where the time complexity is bounded by $\poly(n)$ and (2) where algorithms are \emph{comparison-based}, where nodes can only perform comparison operations on the IDs that constitute their \ktone{} information.

\subsection{Lower Bounds for Time-Bounded Algorithms}
\label{section:KTOneLB-Time-Restricted}
We begin by showing a generalization of Theorem~\ref{thm:general-lb-framework} to the \ktrho{} \congest{} model. The proof is also very similar, but we need to account for the \ktrho information in the simulation.

\begin{theorem}\label{thm:kt-rho-general-lb-framework}
    Fix a function $f:X \times Y \to \{\mathrm{TRUE},\mathrm{FALSE}\}$, a predicate $P$, a constant $0 < \delta < 1$, and a positive integers $\rho \ge 1$ and $\ell$ such that $\ell \ge 2\rho$. Suppose there exists an $\ell$-separated family of lower bound graphs $\{G_{x,y}=(V,E_{x,y})\mid x \in X, y \in Y\}$ w.r.t.~$f$ and $P$. Then any $r$-round deterministic (or randomized with error probability at most constant $0 < \varepsilon < 1$) algorithm for deciding $P$ in the \ktrho{} \congest{} model has message complexity 
\[M = \Omega \left(\frac{(\ell-2\rho+1)}{\log |V|}\cdot\frac{ CC^{\mathrm{rand}}_{\delta + \varepsilon}(f)}{(1 + \log r)} - \frac{\ell \log \ell}{\log |V|}\right).\]
\end{theorem}
\begin{proof}
    Let $\mathcal{A}$ be an $r$-round and $M$-message deterministic (or randomized with error probability $0 < \varepsilon < 1$) \ktrho{} $\congest$ algorithm that decides $P$. We first simulate $\mathcal{A}$ in the synchronous $2$-party communication complexity model in order to find the value of $f(x,y)$ in $(r+1)$-rounds and $\log \ell + (M \log |V|) / \delta(\ell-2\rho+1)$ bits of communication. This simulation has an error probability of at most $\delta + \varepsilon$, and along with the SST result (Lemma~\ref{lemma:SST}) gives us,
    \[\log \ell + \frac{M \log |V|}{\delta(\ell-2\rho+1)} = \Omega \left(\frac{CC^{\mathrm{rand}}_{\delta+\varepsilon}(f)}{(1 + \log r)}\right) \implies M = \Omega \left(\frac{\delta(\ell-2\rho+1)}{\log |V|} \cdot \frac{CC^{\mathrm{rand}}_{\delta+\varepsilon}(f)}{(1 + \log r)} - \frac{\delta \ell \log \ell}{\log |V|} \right)\]
    
    The simulation proceeds as follows: Alice and Bob together create $G_{x,y}$, where Alice is responsible for constructing the edges/weights in $V_1 \times V_1$ and Bob is responsible for constructing the edges/weights in $V_\ell \times V_\ell$. The rest of $G_{x,y}$ does not depend on $x$ and $y$, and is constructed by both Alice and Bob. Alice picks a random number $i \in [\rho,\ell-\rho]$ and sends it to Bob, and they both fix the cut $(V_A, V_B)$ where $V_A = V_1,\dots, V_{i}$ and $V_B = V_{i+1}, \dots, V_\ell$. Then they simulate $\mathcal{A}$ round by round, and in each round only exchange the messages sent over the cut $(V_A,V_B)$. The sets $V_1, \dots, V_{\rho}$ ensure that the \ktrho{} information required by Bob to simulate $\mathcal{A}$ depends at most on the IDs of nodes in $V_1$, and hence do not depend on Alice's input $x$. Similarly, the sets $V_{\ell-\rho+1}, \dots, V_{\ell}$ ensure that the \ktrho{} information required by Alice to simulate $\mathcal{A}$ depends at most on the IDs of nodes in $V_{\ell}$, and hence do not depend on Bob's input $y$. Since the messages are $O(\log |V|)$ bits, we can assume without loss of generality that each message contains the ID of the source and the destination, so Alice and Bob know the recipient of each message.

    Property \ref{lb-fw:cuts} in Definition \ref{def:lb-graph-family} implies that of all the $\ell-2\rho+1$ cuts that Alice and Bob can fix in $G_{x,y}$, no pair of cuts can have a common edge crossing them. In other words, all the $\ell-2\rho+1$ cuts have a distinct set of edges crossing them. Since $\mathcal{A}$ uses $M$ messages overall, at most $\delta$ fraction of these cuts can have at least $M/\delta(\ell-2\rho+1)$ messages passing through them. Therefore, the probability of the bad event that Alice and Bob fix a cut with at least $M/\delta(\ell-2\rho+1)$ messages passing through it is at most $\delta$. Moreover, if $\mathcal{A}$ is randomized, the simulation will have error probability at most $\varepsilon$. Therefore, with probability at least $(1-\delta-\varepsilon)$, Alice and Bob will know whether $G_{x,y}$ satisfies $P$ or not, and by property \ref{lb-fw:pred} of Definition \ref{def:lb-graph-family}, they know whether $f(x,y)$ is TRUE or FALSE.

    The simulation of $\mathcal{A}$ requires $r$ synchronous rounds and $(M \log |V|) / \delta(\ell-2\rho+1)$ bits of communication. And at the beginning they use one round and $\log \ell$ bits to agree on the cut.
\end{proof}

\begin{remark}
    This theorem implies that the cubic lower bounds for exact MVC, \mxis{}, and MDS presented in Section~\ref{section:cubicLowerBounds} also hold for the $\ktrho{}$ \congest{} model where $\rho = c \cdot n$ for any constant $c < 1$.
\end{remark}

\subsubsection{Lower bound for MDS approximation}
We will use the lower bound graph family $G_{x,y}$ described in Section~\ref{section:KTZeroMDSLB} is a $2$-separated lower bound graph family w.r.t. function $f=\sd{}$ on binary strings of length $n^2$, and predicate $P$ which is true for graphs with $6n+2$ vertices having a dominating set of size at least $5$. Recall that the \sd function is defined as: $\sd(x,y) = \mathrm{FALSE}$ iff there exists an index $i$ such that $x_i = y_i = 1$. The corollary below follows from Theorem~\ref{thm:mdscrossing} and the definition of \sd{}.

\begin{corollary}
\label{cor:mdscc}
For $x, y \in \{0,1\}^{n^2}$, if\ \ $\emph{\sd}(x, y) = \mathrm{FALSE}$ then $G_{x, y}$ has a dominating set of size at most 4, and
if\ \ $\emph{\sd}(x, y) = \mathrm{TRUE}$ then $G_{x, y}$ has a dominating set of size at least 5.
\end{corollary}

Let $V_1 = A_1 \cup A_2 \cup C_1 \cup C_2$, $V_2 = B_1 \cup B_2$. With this partition, it is easy to see that we satisfy properties (1) and (2) of Definition~\ref{def:lb-graph-family}, property (3) is automatically satisfied since $\ell=2$, and Corollary~\ref{cor:mdscc} shows that property (4) is also satisfied. We can now invoke Theorem~\ref{thm:kt-rho-general-lb-framework} with $\ell=2, \rho = 1, \delta = 1/6$ and use the fact that $CC^{\mathrm{rand}}_{1/3}(\sd) = \Omega(n^2)$ for inputs of size $n^2$ implies the following lower bound. 

\begin{theorem}\label{thm:mdscc}
For any $0 < \epsilon < 1/6$, the message complexity of an $r$-round \ktone{} \congest{} model algorithm that solves $(5/4-\epsilon)$-approximate MDS on $n$-node networks is lower bounded by $\Omega\left(n^2/(\log r \cdot \log n)\right)$.
\end{theorem}

\subsubsection{Lower Bound for \mxis{} Approximation}
The authors of \cite{BachrachCDELP19} present an $n$-vertex $2$-separated lower bound graph family w.r.t function $f=\sd{}$ on binary strings of length $O(n^2)$, and predicate $P$ which can be determined by a $(7/8 + \epsilon)$-approximate \mxis{}. Theorem~\ref{thm:kt-rho-general-lb-framework}, with $\ell=2, \rho = 1, \delta = 1/6$, and using the fact that $CC^{\mathrm{rand}}_{1/3}(\sd) = \Omega(n^2)$ for inputs of size $O(n^2)$, it implies the following lower bound.

\begin{theorem}\label{thm:maxiscc}
For any $0 < \epsilon < 1/6$, the message complexity of an $r$-round \ktone{} \congest{} model algorithm that solves $(7/8 + \epsilon)$-approximate \mxis{} on $n$-node networks is lower bounded by $\Omega\left(n^2/(\log r \cdot \log n)\right)$.
\end{theorem}

\subsection{Lower Bounds for Comparison Based Algorithms}
\label{section:KTOneLB-Comparison}

Comparison-based algorithms were formally defined by Awerbuch et al.~\cite{AwerbuchGPV90}. In comparison-based algorithms, the algorithm executed by each node contains two types of variables: \emph{ID-type} variables and \emph{ordinary} variables.
In the \ktone{} \congest\ model, the ID-type variables at a node $v$ will store the IDs of $v$ and the neighbors of $v$.
Nodes can send a constant number of ID-type variables each messages, since messages in the \congest{} model are restricted to be \(O(\log n)\) bits. The local computations at any node involve operations of the following two forms only:
\begin{enumerate}
  \item Compare two ID-type variables \(I_{i}, I_{j}\) and store the result of the comparison in an ordinary variable.
  \item Perform arbitrary computations on ordinary variables and storing the result in another ordinary variable.
\end{enumerate}

\subsubsection{Technical Preliminaries}
\label{section:ktone-technical-preliminaries}
We state key definitions and notation from Awerbuch et al.~ \cite{AwerbuchGPV90} and describe the lower bound graph construction from~\cite{PaiPPRPODC2021}, which we will use in our proofs of the message lower bounds in \ktone{} \congest{} model for comparison based algorithms.

\begin{definition}
[Executions]
We denote the execution of a \congest\ algorithm (or protocol) \(\mathcal{A}\) on a graph \(G(V, E)\) with an ID-assignment \(\phi\) by \(EX(\mathcal{A}, G, \phi)\). This execution contains (i) the messages sent and received by the nodes in each round and (ii) a snapshot of the local state of each node in each round.
We denote the state of a node \(v\) in the beginning of round \(i\) of the execution \(EX(\mathcal{A}, G, \phi)\) by \(L_i(EX, v)\).
\end{definition}

The \emph{decoded representation} of an execution is obtained by replacing each occurrence of an ID value \(\phi(v)\) by \(v\) in the execution. This decoded representation allows us to define a similarity of executions. We denote the decoded representations of all messages sent during round \(i\) of an execution \(EX(\mathcal{A}, G, \phi)\) as \(h_i(EX(\mathcal{A}, G, \phi))\).

\begin{definition}
[Similar executions]
Two executions of a \congest\ algorithm \(\mathcal{A}\) on graphs \(G_{1}(V, E_{1})\) and \(G_{2}(V, E_{2})\) with ID-assignments \(\phi_{1}\) and \(\phi_{2}\) respectively are \emph{similar} if they have the same decoded representation. Likewise, we say that two messages are \emph{similar} if their decoded representations are the same.
\end{definition}

\begin{definition}
[Utilized Edge]
An edge \(e=\{u, v\}\) is utilized if any one of the following happens during the course of the algorithm:
(i)  a message is sent along \(e\),
(ii) the node $u$ sends or receives ID $\phi(v)$, or
(iii) the node $v$ sends or receives ID $\phi(u)$.
\end{definition}

By definition, the number of utilized edges is an upper bound on the number of edges along which a message sent. Using a charging argument, Awerbuch et al.~\cite{AwerbuchGPV90} show that the number of utilized edges is also upper bounded by a constant times the number of edges along which a message sent. We restate their claim here.

\begin{lemma}[Lemma 3.4 of \cite{AwerbuchGPV90}]
\label{lem:utilization-message-complexity}
Let \(m_u\) denote the number of utilized edges in an execution \(EX(\mathcal{A}, G, \phi)\). Then the message complexity of the execution is \(\Omega(m_u)\).
\end{lemma}

\paragraph*{Lower Bound Graph:} Consider the base graph $G \cup G'$ described in Section~\ref{section:kt0-mvc-lb}, except that $|Y| = |Y'| = t$. This is the same base graph used in~\cite{PaiPPRPODC2021} to show message complexity lower bounds for comparison-based \ktone{} \congest{} algorithms that solve MIS and $(\Delta + 1)$-coloring. 

For each $v \in V$, let $v'$ denote the corresponding copy in $V'$. For $x \in X, y \in Y$ and $z \in Z$, let $e = \{y, z\}$ and $e' = \{x', y'\}$. We create a crossed graph $G_{e, e'}$ by starting with the base graph $G \cup G'$ and then replacing edges $e$ and $e'$ by edges $\{y, y'\}$ and $\{z, x'\}$.

The authors in~\cite{PaiPPRPODC2021} construct an ID assignment $\phi$ for the vertices $V$ of $G$, and a ``shifted'' ID assignment $\phi'_{e,e'}$ for the vertices $V'$ of $G'$ that produce similar executions for any comparison-based \ktone{} \congest{} algorithm. This is stated as the following lemma.

\begin{lemma}[Lemma 2.8 of~\cite{PaiPPRPODC2021}]
\label{lem:id-assignment-similarity}
Consider an arbitrary vertex $y \in Y$ and an arbitrary pair of edges $e =\{y, z\}$, $z \in Z$ and $e' = \{x',y'\}$, $x' \in X'$. For any comparison-based algorithm $\mathcal{A}$ in the \ktone{} \congest{} model, the executions \(EX_G = EX(\mathcal{A}, G, \phi)\) and \(EX_{G'} = EX(\mathcal{A}, G', \phi'_{e,e'})\) are similar.
\end{lemma}

More importantly, \cite{PaiPPRPODC2021} also show that the ID assignment $\psi_{e,e'} = \phi \cup \phi'_{e,e'}$ produces two similar executions $EX = EX(\mathcal{A}, G \cup G', \psi_{e,e'})$ and $EX_{e,e'} = EX(\mathcal{A}, G_{e,e'}, \psi_{e,e'})$ for any comparison-based \ktone{} \congest{} algorithm $\mathcal{A}$. This is stated as the following corollary.

\begin{corollary}[Corollary 2.7 of \cite{PaiPPRPODC2021}]
\label{cor:utilization-similarity}
Suppose that during the execution \(EX\) neither of the edges \(e = \{y,z\}\) and \(e' = \{x',y'\}\) are utilized, for some vertices \(x \in X\), \(y \in Y\), and \(z \in Z\). Then the executions \(EX\) and \(EX_{e,e'}\) are similar and furthermore in \(EX_{e,e'}\), no messages are sent through the edges \(\{y,y'\}\) and \(\{x',z\}\).
\end{corollary}

\subsubsection{Lower bound for MVC approximation}
We use the base graph $G \cup G'$ and the crossed graph $G_{e,e'}$ described in Section~\ref{section:ktone-technical-preliminaries}, but set $|Y| = |Y'| = t/(2c)$. And we use the same ID assignment $\psi_{e,e'}$ described in the previous section. Note that the ID assignment was defined assuming that the sizes of $Y$ and $Y'$ is $t$, but we have $|Y| = |Y'| = t/(2c)$. To fix this, we can pretend that $Y$ and $Y'$ each contain $t - t/(2c)$ extra isolated dummy nodes. Therefore, the ID assignment $\psi_{e,e'}$ is well defined, and the only difference is that the nodes in $X, Z$ (and $X',Z'$) do not have the IDs of the dummy nodes in $Y$ (and $Y'$) as part of their initial knowledge. So the proofs of Lemma~\ref{lem:id-assignment-similarity} and Corollary~\ref{cor:utilization-similarity} still hold for the modified base and crossed graphs.

\begin{claim}\label{claim:kt1-mvc-size}
Any $c$-approximate vertex cover in $G$ has size at most $t/2$.
\end{claim}
\begin{proof}
This follows from the fact that the optimal vertex cover in $G$ is just $Y$ (and its size is $t/(2c)$). There can be no vertex cover of size $< t/(2c)$ because there are $t^2/c$ edges to cover and the maximum degree a vertex has in $G$ is $2t$.
\end{proof}

Now suppose (to obtain a contradiction) that there is a deterministic comparison-based \ktone{} \congest{} algorithm $\mathcal{A}$ that computes a $c$-approximate vertex cover $C$ in $G \cup G'$ using $o(t^2/c)$ messages. By Lemma~\ref{lem:utilization-message-complexity}, any execution of this algorithm utilizes $o(t^2/c)$ edges.

Lemma~\ref{lem:id-assignment-similarity} implies that $C$ will contain the same vertices in $G$ and $G'$. Thus, $C \cap V$ must be a $c$-approximate vertex cover of $G$ and $C \cap V'$ must be a $c$-approximate vertex cover of $G'$. Therefore, using Claim~\ref{claim:kt1-mvc-size} for $G$ and $G'$ respectively, there are at least $t/2$ vertices in each of the sets $X \setminus C$, $Z \setminus C$, $X' \setminus C$, and $Z' \setminus C$.

\begin{claim}\label{claim:kt1-mvc-triplets}
There exist $x \in X$, $y \in Y$, and $z \in Z$ such that $x$ and $z$ are not in $C$ and the edges $\{x, y\}$ and $\{y, z\}$ are not utilized in execution $EX$ of algorithm $\mathcal{A}$.
\end{claim}
\begin{proof}
There must be at least $t/2$ nodes in $X$ and at least $t/2$ nodes in $Z$ that are not in $C$. And for at least one such pair $x \in X \setminus C$ and $z \in Z \setminus C$, there must be at least one $y \in Y$ such that no message passes over edges $\{x, y\}$ and $\{y, z\}$. Otherwise, for every $y \in Y$ at least $t/2$ edges incident on $y$ have a message passing over the edge. This implies that $\Omega(t^2/c)$ edges in $G$ have been used to send a message which is a contradiction.
\end{proof}

\begin{lemma}\label{lemma:kt1-mvc-correctness}
Let $x \in X \setminus C$, $y \in Y$, and $z \in Z \setminus C$ be three nodes such that the edges $e = \{y, z\}$ and $e' = \{x', y'\}$ are not utilized in execution $EX$ of algorithm $\mathcal{A}$. Then $\mathcal{A}$ cannot compute a correct vertex cover on $G_{e,e'}$. 
\end{lemma}
\begin{proof}
Corollary~\ref{cor:utilization-similarity} implies that both $z$ and $x'$ which are not in $C$ in the execution $EX$ of algorithm $\mathcal{A}$ on $G \cup G'$ continue to not be in $C$ in execution $EX_{e,e'}$ of algorithm $\mathcal{A}$ on $G_{e, e'}$. But in $G_{e,e'}$ there is an edge between $z$ and $x'$ and this is a violation of the correctness of algorithm $\mathcal{A}$. 
\end{proof}

Lemma~\ref{lem:id-assignment-similarity} and Claim~\ref{claim:kt1-mvc-triplets} implies the existence of $e$ and $e'$ assumed in Lemma~\ref{lemma:kt1-mvc-correctness}. Thus $\mathcal{A}$ must use $\Omega(t^2/c)$ messages when run on $G \cup G'$. The following theorem formally states this lower bound.

\begin{theorem}
Any deterministic comparison-based \ktone{} \congest{} algorithm $\mathcal{A}$ that computes a $c$-approximate vertex cover on $n$-vertex graphs has $\Omega(n^2/c)$ message complexity.
\end{theorem}

\begin{theorem}
Any randomized Monte-Carlo comparison-based \ktone{} \congest{} algorithm $\mathcal{A}$ that computes a $c$-approximate vertex cover on $n$-vertex graphs with constant error probability $0 \le \delta < 1/12 - o(1)$ has $\Omega(n^2/c)$ message complexity.
\end{theorem}

\begin{proof}
By Yao's minimax theorem~\cite{Yao77_probab}, it suffices to show a lower bound on the message complexity of a deterministic algorithm under a hard distribution $\mu$.

The hard distribution $\mu$ consists of the base graph $G \cup G'$ with probability $1/2$ and the remaining probability is distributed uniformly over the set $\mathcal{F}$, where $\mathcal{F} = \{G_{e, e'} \mid e = \{y, z\}, e'=\{x', y'\}, y \in Y, z \in Z, x' \in X', y' \in Y'\}$. Note that $|\mathcal{F}| = t^3/(2c)$, since there are $t$ choices for each $x$ and $z$, and $t/(2c)$ choices for each $y$. Note that a deterministic algorithm cannot produce an incorrect output on $G \cup G'$ as it will have error probability greater than $\delta$. Therefore, we now need to show that any deterministic algorithm that produces an incorrect output on at most a $2\delta$ fraction of graphs in $\mathcal{F}$ has $\Omega(t^2)$ message complexity. So suppose for sake of contradiction that there exists such an algorithm that outputs a $c$-approximate vertex cover $C$ with message complexity $o(t^2)$. By Lemma~\ref{lem:utilization-message-complexity}, the algorithm also utilizes no more than $o(t^2)$ edges in any execution.

Therefore, there can only be $o(t)$ vertices in $Y$ that utilize more than $\delta t$ edges. For each of the remaining $(1 - o(1))t/(2c)$ nodes $y \in Y$ there are at least $(1 - \delta)t$ nodes $z \in Z$, and consequently at least $(1/2 - \delta)t$ nodes $z \in Z \setminus C$ (as $|C| \le t/2$), such that the edge $e = \{y,z\}$ is not utilized. Similarly, there are at least $(1/2 - \delta)t$ nodes $x \in X \setminus C$, such that the edge $e' = \{x',y'\}$ is not utilized. The proof of Lemma~\ref{lemma:kt1-maxIS-correctness} says that the deterministic algorithm is incorrect on all such graphs $G_{e,e'}$. 

The algorithm is incorrect on at least $(1 - o(1))(1/2 - \delta)^2 t^3/(2c)$ graphs in $\mathcal{F}$, that is, on at least a $(1/4 - \delta - o(1))$ fraction of graphs in $\mathcal{F}$. This is a contradiction if $2\delta < (1/4-\delta - o(1))$ or in other words, if $\delta < 1/12 - o(1)$. \end{proof}

\subsubsection{Lower Bound for Exact \mxm{}}
Just as in Section~\ref{section:kt-0-max-matching}, maximum matching (which is maximal) immediately gives us a $2$-approximation to MVC if we pick both end points of each matched edge into the cover. We get the following theorem.

\begin{theorem}
Any randomized Monte-Carlo \ktone{} \congest{} algorithm $\mathcal{A}$ computing an exact solution to \mxm{} with constant error probability $0 \le \delta < 1/8 - o(1)$ where each matched edge is output by at least one of its end points has $\Omega(n^2)$ message complexity.
\end{theorem}

\subsubsection{Lower Bound for \mxis{} Approximation}
We use the base graph $G \cup G'$ and the crossed graph $G_{e,e'}$, along with the ID assignment $\psi_{e,e'}$ described in Section~\ref{section:ktone-technical-preliminaries}.

\begin{claim}\label{claim:kt1-maxIS-size}
Any $(1/2 + \epsilon)$-approximate independent set $S$ in $G$ is such that $|S \cap X| \ge \epsilon \cdot t$ and $|S \cap Z| \ge \epsilon \cdot t$.
\end{claim}
\begin{proof}
First we have $S \subseteq X \cup Z$ because, even if a single node of $Y$ is in $S$, it will not allow us to put any node of $X \cup Z$ in $S$. Hence, the largest independent set we can create is $Y$ which has size $t$, which is half the size of the largest independent set in $G$ ($X \cup Z$ having size $2t$). So if $S$ is not a subset of $X \cup Z$, then we can only get a $1/2$ approximation. 
Now the claim follows immediately, as otherwise it is not possible to get a $(1/2+\epsilon)$-approximation.
\end{proof}

Now suppose (to obtain a contradiction) that there is a deterministic comparison-based \ktone{} \congest{} algorithm $\mathcal{A}$ that computes a $(1/2 + \epsilon)$-approximate maximum independent set $S$ in $G \cup G'$ using $o(t^2)$ messages. By Lemma~\ref{lem:utilization-message-complexity}, any execution of this algorithm utilizes $o(t^2)$ edges.

\begin{claim}\label{claim:kt1-maxIS-triplets}
There exist $x \in S \cap X$, $y \in Y$, and $z \in S \cap Z$ such that the edges $\{x, y\}$ and $\{y, z\}$ are not utilized in execution $EX$ of algorithm $\mathcal{A}$.
\end{claim}
\begin{proof}
If this is not the case, then for every $y \in Y$, all edges $\{x, y\}$ for $x \in S \cap X$ or all edges $\{z, y\}$ for $z \in S \cap Z$ have a message sent over them. This implies that at least $\epsilon \cdot t^2 = \Theta(n^2)$ edges have a message sent over them, contradicting the assumption that $\mathcal{A}$ has $o(n^2)$ message complexity.
\end{proof}

\begin{lemma}\label{lemma:kt1-maxIS-correctness}
Let $x \in S \cap X$, $y \in Y$, and $z \in S \cap Z$ be three nodes such that the edges $e = \{y, z\}$ and $e' = \{x', y'\}$ are not utilized in execution $EX$ of algorithm $\mathcal{A}$. Then $\mathcal{A}$ cannot compute a correct independent set on $G_{e,e'}$. 
\end{lemma}
\begin{proof}
Corollary~\ref{cor:utilization-similarity} implies that both $z$ and $x'$ which are in $S$ in the execution $EX$ of algorithm $\mathcal{A}$ on $G \cup G'$ continue to be in $S$ in execution $EX_{e,e'}$ of algorithm $\mathcal{A}$ on $G_{e, e'}$. But in $G_{e,e'}$ there is an edge between $z$ and $x'$ and this is a violation of the independence of $S$, and hence algorithm $\mathcal{A}$ is incorrect on $G_{e,e'}$. 
\end{proof}

Lemma~\ref{lem:id-assignment-similarity} and Claim~\ref{claim:kt1-maxIS-triplets} implies the existence of $e$ and $e'$ assumed in Lemma~\ref{lemma:kt1-maxIS-correctness}. Thus $\mathcal{A}$ must use $\Omega(t^2)$ messages when run on $G \cup G'$. The following theorem formally states this lower bound.

\begin{theorem}
For any constant $\epsilon > 0$, any deterministic comparison-based \ktone{} \congest{} algorithm $\mathcal{A}$ that computes a $(1/2+\epsilon)$-approximation of \mxis{} on $n$-vertex graphs has $\Omega(n^2)$ message complexity.
\end{theorem}

\begin{theorem}
Any randomized Monte-Carlo comparison-based \ktone{} \congest{} algorithm $\mathcal{A}$ that computes a $(1/2 + \epsilon)$-approximation of \mxis{} on $n$-vertex graphs with constant error probability $\delta < \epsilon^2/8 - o(1)$ has $\Omega(n^2)$ message complexity.
\end{theorem}
\begin{proof}
By Yao's minimax theorem~\cite{Yao77_probab}, it suffices to show a lower bound on the message complexity of a deterministic algorithm under a hard distribution $\mu$.

The hard distribution $\mu$ consists of the base graph $G \cup G'$ with probability $1/2$ and the remaining probability is distributed uniformly over the set $\mathcal{F}$, where $\mathcal{F} = \{G_{e, e'} \mid e = \{y, z\}, e'=\{x', y'\}, y \in Y, z \in Z, x' \in X', y' \in Y'\}$. Note that $|\mathcal{F}| = t^3$, since there are $t$ choices for each $x$, $y$ and $z$. Note that a deterministic algorithm cannot produce an incorrect output on $G \cup G'$ as it will have error probability greater than $\delta$. Therefore, we now need to show that any deterministic algorithm that produces an incorrect output on at most a $2\delta$ fraction of graphs in $\mathcal{F}$ has $\Omega(t^2)$ message complexity. So suppose for sake of contradiction that there exists such an algorithm that outputs a $(1/2 + \epsilon)$-approximate independent set $S$ with message complexity $o(t^2)$. By Lemma~\ref{lem:utilization-message-complexity}, the algorithm also utilizes no more than $o(t^2)$ edges in any execution.

Therefore, there can only be $o(t)$ vertices in $Y$ that utilize more than $\delta t$ edges. For each of the remaining $(1 - o(1))t$ nodes $y \in Y$ there are at least $(1 - \delta)t$ nodes $z \in Z$, and consequently at least $(\epsilon/2)t$ nodes $z \in Z \cap S$ (as $|Z \cap S| \ge \epsilon t$, and $\delta < \epsilon/2$), such that the edge $e = \{y,z\}$ is not utilized. Similarly, there are at least $(\epsilon/2)t$ nodes $x \in X \cap S$, such that the edge $e' = \{x',y'\}$ is not utilized. The proof of Lemma~\ref{lemma:kt1-maxIS-correctness} says that the deterministic algorithm is incorrect on all such graphs $G_{e,e'}$. 

The algorithm is incorrect on at least $(1 - o(1))(1/2)^2 \epsilon^2 t^3$ graphs in $\mathcal{F}$, that is, on at least a $(\epsilon^2/4 - o(1))$ fraction of graphs in $\mathcal{F}$. This is a contradiction if $2\delta < \epsilon^2/4 - o(1)$. 
\end{proof}

 \end{document}